\documentclass[reqno,11pt]{amsart}
\usepackage{fullpage}
\usepackage{amsfonts}
\usepackage{amssymb}
\usepackage{arydshln}
\usepackage{latexsym}
\usepackage{cite}
\usepackage{graphicx}
\usepackage{subfigure}
\usepackage{bm}
\usepackage{color}

\usepackage{amsmath,amsthm,accents}
\numberwithin{equation}{section}
\numberwithin{equation}{section}

\newtheorem{Thm}{Theorem}

\newcommand{\wh}{\widehat}
\newcommand{\wt}{\widetilde}
\newcommand{\wb}{\overline}

\def \ty#1{\hbox{\tiny{{\it{#1}}}}}

\newcommand{\nn}{\nonumber}

\newcommand{\bC}{\boldsymbol{C}}
\newcommand{\bD}{\boldsymbol{D}}
\newcommand{\bI}{\boldsymbol{I}}
\newcommand{\bK}{\boldsymbol{K}}

\newcommand{\bM}{\boldsymbol{M}}

\newcommand{\bP}{\boldsymbol{P}}
\newcommand{\bQ}{\boldsymbol{Q}}

\newcommand{\br}{\boldsymbol{r}}
\newcommand{\bs}{\boldsymbol{s}}

\newcommand{\bT}{\boldsymbol{T}}
\newcommand{\bF}{\boldsymbol{F}}
\newcommand{\bG}{\boldsymbol{G}}
\newcommand{\bH}{\boldsymbol{H}}
\newcommand{\Ga}{\boldsymbol{\gamma}}

\newcommand{\Ph}{\boldsymbol{\Phi}}
\newcommand{\Og}{\boldsymbol{\Omega}}
\newcommand{\st}{\hbox{\tiny\it{T}}}

\begin{document}

\title{Cauchy matrix solutions to some local and nonlocal complex equations}

\author{Hai-jing Xu, Song-lin Zhao$^{*}$\\
\\\lowercase{\scshape{
Department of Applied Mathematics, Zhejiang University of Technology,
Hangzhou 310023, P.R. China}}}
\email{Corresponding Author: songlinzhao@zjut.edu.cn}

\begin{abstract}

In this paper, we develop a Cauchy matrix reduction technique
that enables us to obtain solutions for the reduced local and nonlocal complex equations from the Cauchy matrix
solutions of the original before-reduction systems. Specifically,
by imposing local and nonlocal complex reductions on some Ablowitz-Kaup-Newell-Segur-type
equations, we study some local and nonlocal complex equations, involving
the local and nonlocal complex modified Korteweg-de Vries equation, the local and nonlocal complex sine-Gordon equation,
the local and nonlocal potential nonlinear Schr\"{o}dinger equation and the local and nonlocal potential complex
modified Korteweg-de Vries equation. Cauchy matrix-type soliton solutions and
Jordan block solutions for the aforesaid local and nonlocal complex equations are presented.
The dynamical behaviors of some obtained solutions are analyzed with
graphical illustrations.

\end{abstract}

\maketitle

\section{Introduction}
\label{sec-1}

The complex integrable equations have been of considerable interest recently in
many areas of mathematical physics. Some famous examples include the nonlinear Schr\"{o}dinger (NLS)
equation, the complex modified Korteweg-de Vries (cmKdV) equation, the complex sine-Gordon (csG) equation,
and so on. The NLS equation
\begin{align}
\label{NLS}
2iu_t+u_{xx}+8|u|^2u=0
\end{align}
usually appears in the description of deep water
waves \cite{Zak}, vortex filaments \cite{Hasimoto} and the collapse of Langmuir waves in plasma
physics \cite{Zak-1}. The cmKdV equation
\begin{align}
\label{mKdV}
4 u_{t}+u_{xxx}+24|u|^2u_{x}=0
\end{align}
has been applied as a model for the
nonlinear evolution of plasma waves \cite{Char} and a molecular chain
model \cite{Gorb}, as well as short pulses in optical fibers \cite{Anco}. The csG equation reads
\begin{align}
\label{csG}
u_{xt}+4u+8u\partial^{-1}(|u|^2)_t=0,
\end{align}
which arose in general relativity \cite{Vega} and propagated optical pulses in a nonlinear medium \cite{Park}.
In equations \eqref{NLS}-\eqref{csG}, the dependent variable $u$ is defined on the continuous coordinates $(x,t)\in \mathbb{R}^2$;
$i$ is the imaginary unit and $|\cdot|$ means module. Besides,
$\partial^{-1}$ stands for an inverse operator of $\partial=\frac{\partial}{\partial x}$
with the condition $\partial\partial^{-1}=\partial^{-1}\partial=1$. It is well known that
all the three above equations can be gained from the
Ablowitz-Kaup-Newell-Segur (AKNS) hierarchy by complex reductions.

In recent years, the study of parity-time symmetric integrable models
has become a focus of attention within the theory of integrable systems.
The parity-time symmetric integrable models play important roles in quantum
physics and other areas of physics, such as the
quantum chromodynamics, electric circuits, optics, Bose-instein condensates, and so on.
One of the famous examples is the reverse space NLS equation
\begin{align}
\label{n-NLS}
iu_t+u_{xx}+u^2u^*(-x)=0,
\end{align}
which was proposed firstly by Ablowitz and Musslimani \cite{AM-nNLS}. Here and hereafter, asterisk denotes the
complex conjugate. This equation is parity-time symmetric because it is invariant under the action of the parity-time operator,
i.e., the joint transformations $x\rightarrow -x,~t\rightarrow -t$ and complex conjugation $^*$.
Equation \eqref{n-NLS} possesses Lax integrability and admits
infinite number of conservation laws. Besides, it can be solved by the inverse scattering
transform \cite{AM-nNLS-IST}. Since the equation \eqref{n-NLS} includes two different
places $\{x, t\}$ and $\{x'=-x,~t'=t\}$, this model can be used to describe two-place physics phenomena \cite{Lou-JMP}.
Up to now, several methods have been adopted to
search for solutions to the reverse space NLS equation \eqref{n-NLS}, such as
the Riemann-Hilbert approach, the Hirota's bilinear method, the Darboux transformation, etc.
\cite{Sinha,Valchev,AbLuoMu,Yan,ChenZ-AML-2017,LX}. Besides the study of nonlocal NLS equation,
various results for nonlocal cmKdV equation have been also investigated. The inverse scattering transform for the reverse time-space nonlocal
cmKdV equation with nonzero boundary conditions at infinity was presented \cite{Luo}.
Ma {\it et al.} \cite{MaShenZhu} showed that nonlocal cmKdV equation
was gauge equivalent to a spin-like model. Consequently, dark soliton, W-type soliton, M-type soliton, and periodic solutions
for the nonlocal cmKdV equation were obtained with the help of the Darboux transformation.
In \cite{YY-SAPM}, Yang and Yang introduced a simple variable transformation to convert the
nonlocal cmKdV equation to the local cmKdV equation. As a result, multisoliton and quasi-periodic solutions
for the nonlocal cmKdV equation were constructed.
By using the Ablowitz-Musslimani reduction formulas, G\"{u}rses and Pekcan
found 1-, 2-, and 3-soliton solutions of the nonlocal cmKdV equations \cite{GuPe}.
Based on the double Wronskian solutions for the AKNS hierarchy,
Zhang {\it et al.} developed a reduction technique by imposing a constraint on the two
basic vectors in double Wronskians so that two potential functions in the AKNS
hierarchy obey some nonlocal relations therefore solutions to the nonlocal cmKdV equation
were obtained \cite{SAPM-2018}. This method proves generally and has been applied to many nonlocal systems,
involving the nonlocal discrete soliton equations \cite{DLZ,FZS} and the nonlocal nonisospectral soliton equations \cite{LWZ,FZ-IJMPB,Xu-Zhao}.
With regards to nonlocal csG equation, its solutions were also studied in \cite{SAPM-2018}. Furthermore,
the covariant hodograph transformations between nonlocal short pulse models and nonlocal multi-component csG equation were discussed
\cite{Zhang-AML-SP}.

The Cauchy matrix approach, as a systematic method for constructing integrable equations together with their solutions,
was first proposed by Nijhoff and his collaborators \cite{NAH-2009-JPA,N-2004-math}
to investigate the soliton solutions of the Adler-Bobenko-Suris lattice list. This method is actually a
by-product of the linearization approach which was first proposed by Fokas
and Ablowitz \cite{FA-DL} and developed to discrete integrable
systems by Nijhoff, Quispel {\it et al.}, one can refer to \cite{NQLC,FCW-BOOK,SAPM-2012,PRSA-2017}.
In \cite{ZZ-SAPM-2013}, Zhang and Zhao proposed
a generalized Cauchy matrix scheme to construct more kinds of exact solutions
for the Adler-Bobenko-Suris lattice list beyond the soliton solutions. The (generalized) Cauchy matrix approach
actually arose from the well known Sylvester equation in matrix theory \cite{S-1884}. In \cite{XZZ}, they used this method
to discuss the relations between the Sylvester equation and some continuous integrable equations,
including the KdV equation, the mKdV equation, the Schwarzian KdV equation and the sG equation.
They found that all these equations arose from the same Sylvester equation and could be expressed by some discrete equations of
$S^{(i,j)}$ defined on certain points.


Recently, motivated by the Cauchy matrix approach and the understanding of dispersion relations,
the first author of the present paper
studied the connections between the Sylvester equation and some AKNS equations \cite{Zhao-ROMP}, involving
the second order AKNS equation, the third order AKNS equation, the negative order AKNS equation,
the second order potential AKNS equation and the third order potential AKNS equation.
Consequently, Cauchy matrix-type soliton solutions, Jordan block solutions and generic
soliton-Jordan block mixed solutions for these AKNS-type equations were revealed.
In terms of that work, Cauchy matrix-type soliton solutions and Cauchy matrix-type Jordan block
solutions for the nonlocal NLS equation were investigated \cite{FZ-ROMP}.

In this paper we plan to discuss Cauchy matrix solutions to some local and nonlocal complex integrable equations
which can be obtained by imposing reductions on the above AKNS-type equations except for the second order AKNS equation.
The models include local and nonlocal cmKdV equation, local and nonlocal csG equation,
local and nonlocal potential NLS equation and local and nonlocal potential cmKdV equation.
In what follows, we call the AKNS-type equations mentioned above except for the second order AKNS equation
in sequence by AKNS(3) equation, AKNS(-1) equation, pAKNS(2) equation and pAKNS(3) equation for short.

The paper is organized as follows.
In Sec. 2, we briefly review some AKNS-type equations and their local and nonlocal complex reductions.
The corresponding Cauchy matrix solutions are also listed.
In Sec. 3, Cauchy matrix solutions to the local and nonlocal cmKdV
equation and the local and nonlocal csG equation are given.
Dynamics of some obtained solutions are analyzed and illustrated by asymptotic analysis.
In Sec. 4, we present the Cauchy matrix solutions and the dynamics of the
local and nonlocal potential NLS equation and the local and nonlocal potential cmKdV equation.
Sec. 5 is for conclusions and some remarks.

\section{AKNS-type equations and Cauchy matrix solutions}
\label{sec-2}

As a preliminary part let us present an account of some
AKNS-type equations and their local and nonlocal complex reductions.
The equations involve AKNS(3) equation, AKNS(-1) equation, pAKNS(2) equation and
pAKNS(3) equation. Moreover, we recall the Cauchy matrix solutions,
including Cauchy matrix-type soliton solutions and
Cauchy matrix-type Jordan block solutions, to these equations.
For more details one can refer to \cite{Zhao-ROMP}.

\subsection{AKNS-type equations}

In \cite{AKNS-1973}, Ablowitz, Kaup, Newell and Segur proposed a spectral problem, named AKNS
spectral problem,
\begin{align}
\label{AKNS-sp}
&\Ph_x=\bP\Ph,\quad
\bP=\left(\begin{array}{cc}
-\lambda & 2u \\
-2v & \lambda
\end{array}
\right),\quad \Ph=\left(\begin{array}{c}
\Phi_{1}\\
\Phi_{2}
\end{array}\right),
\end{align}
with spectral parameter $\lambda$ and potentials $u=u(x,t)$ and $v=v(x,t)$,
which provides integrable backgrounds for the NLS equation, the mKdV equation and
the sG equation. By imposing time evolution
\begin{align}
\label{AKNS-time}
\Ph_t=\bQ \Ph, \quad \bQ=\left( \begin{array}{cc}
A & B \\
C & -A
\end{array}
\right),
\end{align}
the compatibility condition  $\Ph_{xt}=\Ph_{tx}$ or moreover
zero curvature equation $\bP_{t}-\bQ_{x}+[\bP,\bQ]=0$ leads to
\begin{subequations}
\label{uv-A}
\begin{align}
& A=2\partial^{-1}(-v,u)\left(\begin{array}{c}
-B \\ C
\end{array}\right)-\lambda_tx+A_0,\\
& \label{uv-t}
2\left(\begin{array}{c}
u \\ -v
\end{array}\right)_t=L\left(\begin{array}{c}
-B \\ C
\end{array}\right)-2\lambda\left(\begin{array}{c}
-B \\ C
\end{array}\right)+4A_0\left(\begin{array}{c}
u\\
v
\end{array}\right)-4\lambda_t\left(\begin{array}{c}
xu\\
xv
\end{array}\right),
\end{align}
\end{subequations}
where $A_0$ is a constant and
\begin{align}
\label{L}
L=\left(\begin{array}{cc}
-\partial & 0\\
0 & \partial
\end{array}\right)+8\left(\begin{array}{c}
u \\ v
\end{array}\right)\partial^{-1}(-v,u).
\end{align}

Setting $\lambda_t=0$, $A_0=\lambda^n$ and expanding $(B,C)^{\st}$ into
\begin{align}
\left(\begin{array}{c} B\\ C \end{array}
\right)=\sum\limits_{j=1}^{n} \left(\begin{array}{c} B_j\\ C_j
\end{array} \right) \lambda^{n-j},
\end{align}
then from \eqref{uv-t} one can derive
the positive order AKNS hierarchy
\begin{align}
\label{AKNS-hie}
2^{n-1}\left(\begin{array}{c} u\\ -v\end{array}
\right)_t
=L^n \left( \begin{array}{c} u\\ v \end{array} \right), \quad
(n=1,2,\ldots),
\end{align}
where $L$ given by \eqref{L} is the recursion operator.
Now we still take $\lambda_t=0$ but $A_0=\lambda^{-n}$.
Substituting the expansion
\begin{align}
\left(\begin{array}{c} B\\ C
\end{array} \right)=\sum\limits_{j=1}^{n} \left(\begin{array}{c}
B_j\\ C_j \end{array} \right) \lambda^{j-n-1}
\end{align}
into \eqref{uv-t} and taking $(B_1, C_1)^{\st}$ such that
$L(-B_1, C_1)^{\st}=-4(u, v )^{\st}$,
we reach the following negative order AKNS hierarchy
\begin{align}
\label{NAKNSH}
L^n \left(\begin{array}{c} u\\ -v \end{array}
\right)_t=2^{n+1}\left(\begin{array}{c} u\\ v \end{array} \right), \quad
(n=1,2,\cdots).
\end{align}


In the positive order AKNS hierarchy \eqref{AKNS-hie}, one of the most salient or characteristic equations is the AKNS(3) equation, which is of form
\begin{subequations}
\label{3th-AKNS}
\begin{align}
& \label{3th-AKNS-b}
4 u_{t}+u_{xxx}+24 uvu_{x}=0, \\
& \label{3th-AKNS-c}
4 v_{t}+v_{xxx}+24 uvv_{x}=0.
\end{align}
\end{subequations}
Under complex reduction
\begin{align}
\label{3AKNS-red}
v(x,t)=\delta u^*(\sigma x,\sigma t), \quad \delta,\sigma=\pm 1, \quad x,t \in \mathbb{R},
\end{align}
one can get a local and nonlocal cmKdV equation
\begin{align}
\label{cmKdV}
4 u_{t}+u_{xxx}+24\delta uu^*(\sigma x,\sigma t)u_{x}=0.
\end{align}
It is obvious that equation \eqref{cmKdV} is
preserved under transformations $u\rightarrow -u$ and $u\rightarrow \pm iu$.
When $\sigma=-1$, equation \eqref{cmKdV} is the reverse time-space cmKdV equation.

The first AKNS-type equation in the negative order AKNS hierarchy \eqref{NAKNSH} is the AKNS(-1) equation (see \cite{ZJZ}), i.e.,
\begin{subequations}
\label{-1th-AKNS}
\begin{align}
& u_{xt}+4 u+8 u\partial^{-1}(uv)_t=0, \\
& v_{xt}+4 v+8 v\partial^{-1}(uv)_t=0.
\end{align}
\end{subequations}
Imposing reduction
\begin{align}
\label{-1AKNS-red}
v(x,t)=\delta u^*(\sigma x,\sigma t), \quad \delta, \sigma=\pm 1, \quad x,t \in \mathbb{R},
\end{align}
on equation \eqref{-1th-AKNS} gives rise to a local and nonlocal (non-potential) csG equation
\begin{align}
\label{n-sG}
u_{xt}+4u+8\delta u\partial^{-1}(uu^*(\sigma x,\sigma t))_t=0.
\end{align}
Equation \eqref{n-sG} is also preserved under transformations $u\rightarrow -u$ and $u\rightarrow \pm iu$.
When $\sigma=-1$, equation \eqref{n-sG} is the reverse time-space csG equation.

Since both $u$ and $v$ tend to zero as $|x|\rightarrow \infty$ and $uv$ is a conserved density,
one can alternatively write \eqref{-1th-AKNS} as
\begin{align}
\label{-1th-AKNS-ano}
u_{xt}+8uw=0, \quad v_{xt}+8vw=0, \quad w_x=(uv)_t,
\end{align}
where $w=\partial^{-1}(uv)_t+\frac{1}{2}$ is an auxiliary function. By complex reduction \eqref{-1AKNS-red}
together with relation $w(x,t)=\sigma w^*(\sigma x,\sigma t)$,
one get another form of local and nonlocal csG equation
\begin{align}
\label{n-sG-1}
u_{xt}+8uw=0, \quad w_x=\delta (uu^*(\sigma x,\sigma t))_t.
\end{align}

Besides the AKNS hierarchy, there exists another hierarchy which can be also viewed as AKNS
type. We name this hierarchy as a potential AKNS hierarchy. The first non-trivial
member is the pAKNS(2) equation, namely,
\begin{subequations}
\label{2th-pAKNS}
\begin{align}
& \label{2th-pAKNS-a}
2 q_t-q_{xx}+2\frac{q_xr_x}{r}=0, \\
& \label{2th-pAKNS-b}
2 r_t+r_{xx}-2\frac{q_xr_x}{q}=0.
\end{align}
\end{subequations}
With complex reduction
\begin{align}
\label{2pAKNS-red}
r(x,t)=q^*(\sigma x,t), \quad \sigma=\pm 1, \quad t\rightarrow -it, \quad x,t \in \mathbb{R},
\end{align}
equation \eqref{2th-pAKNS} leads to a local and nonlocal potential NLS equation
\begin{align}
\label{n-pNLS-1}
2iq_t-q_{xx}+2\frac{q_x q^*(\sigma x)}{q^*(\sigma x)}=0,
\end{align}
which is preserved under transformations $q\rightarrow -q$ and $q\rightarrow \pm iq$.
When $\sigma=1$, \eqref{n-pNLS-1} is the local potential NLS equation, which
was firstly introduced by Nijhoff {\it et al.} \cite{Schwar,Schwar-1} and
related to the equation of motion of the Heisenberg ferromagnet with uniaxial anisotropy \cite{Quispel-1982-3}.
When $\sigma=-1$, \eqref{n-pNLS-1} is the reverse space potential NLS equation.

The second non-trivial member in the potential AKNS hierarchy is the pAKNS(3) equation, which is described as
\begin{subequations}
\label{3th-pAKNS}
\begin{align}
& 4 q_{t}+q_{xxx}-3\frac{q_{xx}r_x}{r}-3\frac{q_xr_x}{qr^2}(q_x r-qr_x)=0, \\
& 4 r_{t}+r_{xxx}-3\frac{q_x r_{xx}}{q}+3\frac{q_xr_x}{q^2 r}(q_xr-qr_x)=0.
\end{align}
\end{subequations}
Through complex reduction
\begin{align}
\label{3pAKNS-red}
r(x,t)=q^*(\sigma x,\sigma t), \quad \sigma=\pm 1, \quad x,t \in \mathbb{R},
\end{align}
equation \eqref{3th-pAKNS} yields a local and nonlocal potential cmKdV equation
\begin{align}
\label{n-pmKdV-1}
4 q_{t}+q_{xxx}-3\frac{q_{xx}q_x^*(\sigma x,\sigma t)}{q^*(\sigma x,\sigma t)}-3\frac{q_xq_x^*(\sigma x,\sigma t)}
{qq^{*^2}(\sigma x,\sigma t)}(q_x q^*(\sigma x,\sigma t)-q q_x^*(\sigma x,\sigma t))=0,
\end{align}
which is also preserved under transformations $q\rightarrow -q$ and $q\rightarrow \pm iq$.
When $\sigma=-1$, equation \eqref{n-pmKdV-1} corresponds to the reverse time-space potential cmKdV equation.

It is really noteworthy that the pAKNS(2) equation \eqref{2th-pAKNS} and the pAKNS(3) equation \eqref{3th-pAKNS}, respectively, are related to the
AKNS(2) equation
\begin{subequations}
\label{2th-AKNS}
\begin{align}
& \label{2th-AKNS-a}
2 u_{t}-u_{xx}-8u^2v=0, \\
& \label{2th-AKNS-b}
2 v_t+v_{xx}+8uv^2=0,
\end{align}
\end{subequations}
and the AKNS(3) equation \eqref{3th-AKNS} by transformations
\begin{align}
\label{MT}
2u=q_x/r,\quad 2v=-r_x/q.
\end{align}

\subsection{Cauchy matrix solutions}

The present part is dedicated to Cauchy matrix solutions
of the AKNS-type equations mentioned above, including Cauchy matrix-type soliton solutions
and Cauchy matrix-type Jordan block solutions.

Cauchy matrix solutions of the AKNS(3) equation \eqref{3th-AKNS} and the AKNS(-1) equation \eqref{-1th-AKNS}
can be summarized in the following Theorem.
\begin{Thm}
\label{Thm-1}
The functions
\begin{subequations}
\label{AKNS-solu}
\begin{align}
& u=\bs^{\st}_2(\bI-\bM_2\bM_1)^{-1}\br_2, \\
& v=\bs^{\st}_1(\bI-\bM_1\bM_2)^{-1}\br_1,
\end{align}
\end{subequations}
solve the AKNS(3) equation \eqref{3th-AKNS} and the AKNS(-1) equation \eqref{-1th-AKNS},
provided that the components $\bM_1\in \mathbb{C}_{N_1\times N_2}$, $\bM_2\in \mathbb{C}_{N_2\times N_1}$,
$\br_j,~\bs_j \in \mathbb{C}_{N_j\times 1}$, $(j=1,2)$ satisfy the following determining
equation set
\begin{subequations}
\label{SE-sys}
\begin{align}
& \label{SE}
\bK_1 \bM_1-\bM_1\bK_2=\br_1\, \bs_2^{\st}, \quad  \bK_2 \bM_2-\bM_2\bK_1=\br_2\, \bs_1^{\st}, \\
& \label{rs-x}
\br_{j,x}=(-1)^{j-1}\bK_j \br_j,\qquad \bs_{j,x}=(-1)^{j-1}\bK_j^{\st}\bs_j, \\
& \label{rs-t}
\br_{j,t}=(-1)^{j}\bK_j^n\br_j,\qquad \bs_{j,t}=(-1)^{j}(\bK_j^{\st})^n\bs_j,
\end{align}
\end{subequations}
with $n=3$ and $n=-1$, respectively, where $N_1+N_2=2N$ and $^{\st}$ means transpose.
\end{Thm}

For notational brevity, here and in what follows, we omit the index
of each unit matrix $\bI$ to indicate its size.
Cauchy matrix solutions of the pAKNS(2) equation \eqref{2th-pAKNS} and the pAKNS(3) equation
\eqref{3th-pAKNS} are presented by the following Theorem.
\begin{Thm}
\label{Thm-2}
The functions
\begin{subequations}
\label{2th-pAKNS-so}
\begin{align}
& q=\bs^{\st}_2(\bI-\bM_2\bM_1)^{-1}\bK_2^{-1}\br_2-\bs^{\st}_2\bM_2(\bI-\bM_1\bM_2)^{-1}\bK_1^{-1}\br_1-1, \\
& r=\bs^{\st}_1(\bI-\bM_1\bM_2)^{-1}\bK_1^{-1}\br_1-\bs^{\st}_1\bM_1(\bI-\bM_2\bM_1)^{-1}\bK_2^{-1}\br_2-1,
\end{align}
\end{subequations}
solve the pAKNS(2) equation \eqref{2th-pAKNS} and the pAKNS(3) equation \eqref{3th-pAKNS},
provided that the components $\bK_j \in \mathbb{C}_{N_j\times N_j}$,
$\bM_1\in \mathbb{C}_{N_1\times N_2}$, $\bM_2\in \mathbb{C}_{N_2\times N_1}$,
$\br_j,~\bs_j \in \mathbb{C}_{N_j\times 1}$, $(j=1,2)$ satisfy the determining
equation set \eqref{SE-sys} with $n=2$ and $n=3$, respectively, where $N_1+N_2=2N$.
\end{Thm}

According to the analysis above, we know that Cauchy matrix solutions of the
AKNS-type equations \eqref{3th-AKNS}, \eqref{-1th-AKNS},
\eqref{2th-pAKNS} and \eqref{3th-pAKNS} are determined by the determining equation set \eqref{SE-sys}.
Equations \eqref{rs-x} and \eqref{rs-t} are used to determine the dispersion relations
$\br_j,~\bs_j,~(j=1,2)$ and the two equations in \eqref{SE} are applied to determine $\bM_1$ and $\bM_2$.
Both of equations in \eqref{SE} are the
Sylvester equations and have a unique solution for $\bM_1$ and $\bM_2$
provided that $\mathcal{E}(\bK_1)\bigcap \mathcal{E}(\bK_2)=\varnothing$,
where $\mathcal{E}(\bK_1)$ and $\mathcal{E}(\bK_2)$, respectively,
denote eigenvalue sets of $\bK_1$ and $\bK_2$ (see \cite{S-1884}). We assume that
$\bK_1$ and $\bK_2$ satisfy such a condition. Additionally, we suppose
$0 \notin \mathcal{E}(\bK_1\cup\bK_2)$ and $1 \notin \mathcal{E}(\bM_1\bM_2)$
to guarantee the reversibility of matrices
$\bK_1$, $\bK_2$ and $\bI-\bM_1\bM_2$. Because $\bM_1\bM_2$ and $\bM_2\bM_1$
have the same non-zero eigenvalues, matrix $\bI-\bM_2\bM_1$ is also invertible.
It is worth noting that variables $u$, $v$, $q$ and $r$ are invariant
and system \eqref{SE-sys} is covariant under similarity transformations
\begin{align}
\bK_j=\Ga_j^{-1}\wb{\bK}_j\Ga_j,\quad \br_j=\Ga_j^{-1}\wb{\br}_j,\quad \bs_j=\Ga^{\st}_j\wb{\bs}_j,\quad \bM_1=\Ga_1^{-1}\wb{\bM}_1\Ga_2,\quad
\bM_2=\Ga_2^{-1}\wb{\bM}_2\Ga_1,
\end{align}
for $j=1,2$, where $\Ga_1$ and $\Ga_2$ are transform matrices. For deriving the explicit solutions one just need to
solve the following canonical equations
\begin{subequations}
\label{DES-C}
\begin{align}
& \label{DES-M12-C}\Og_1 \bM_1-\bM_1\Og_2=\br_1\, \bs^{\st}_2, \quad \Og_2 \bM_2-\bM_2\Og_1=\br_2\, \bs^{\st}_1, \\
& \label{DES-x-C}\br_{j,x}=(-1)^{j-1}\Og_j \br_j,\quad \bs_{j,x}=(-1)^{j-1}\Og^{\st}_j\bs_j, \\
& \label{DES-t-C}\br_{j,t}=(-1)^{j}\Og_j^n\br_j,\quad \bs_{j,t}=(-1)^{j}(\Og_j^{\st})^n\bs_j,
\end{align}
\end{subequations}
for $j=1,2$, where $\Og_1$ and $\Og_2$ are the canonical forms of the matrices $\bK_1$ and $\bK_2$, respectively.
In terms of the constraints on $\bK_1$ and $\bK_2$, we know that
\begin{align}
\label{con-Og}
\mathcal{E}(\Og_1)\bigcap \mathcal{E}(\Og_2)=\varnothing, \quad 0 \notin \mathcal{E}(\Og_1\cup\Og_2).
\end{align}

Equations \eqref{DES-x-C} and \eqref{DES-t-C} are linear, which implies the expressions for $\br_j$ and $\bs_j$ as
\begin{align}
\label{rsj}
\br_{j}=\mbox{exp}((-1)^{j-1}(\Og_{j}x-\Og^n_{j}t))\bC^+_{j}, \quad \bs_{j}=\mbox{exp}((-1)^{j-1}(\Og^{\st}_{j}x-(\Og_j^{\st})^nt)\bD^+_{j},\quad j=1,2,
\end{align}
where $\{\bC^+_j,~\bD^{+}_j \}$ are constant column vectors. The key point of
the solving procedure of \eqref{DES-M12-C} is to factorize $\bM_1$ and $\bM_2$ into triplets, i.e. $\bM_1=\bF_1\bG_1\bH_2$ and
$\bM_2=\bF_2\bG_2\bH_1$, where $\{\bF_j,~\bH_j\} \subset \mathbb{C}_{N_j\times N_j}$, $\bG_1
\in \mathbb{C}_{N_1\times N_2}$ and $\bG_2\in \mathbb{C}_{N_2\times N_1}$. For the detailed calculations, one can refer to \cite{XZZ}.
In what follows, we just list two kinds of solutions involving  Cauchy matrix-type soliton solutions
and Cauchy matrix-type Jordan block solutions, where usually the subscripts $_D$ and $_J$ correspond to the cases of
$\{\Og_j\}$ being diagonal and Jordan block, respectively.

When
\begin{align}
\Og_j=\Og_{\ty{D},j}=\mbox{Diag}(k_{j,1},k_{j,2},\ldots,k_{j,N_j}),\quad j=1,2,
\end{align}
we have
\begin{subequations}
\label{MDD-solu}
\begin{align}
& \label{rD-solu} \br_{j}=\mbox{exp}((-1)^{j-1}(\Og_{\ty{D},j}x-\Og^n_{\ty{D},j}t))\bC^+_{\ty{D},j},\quad j=1,2,\\
& \label{sD-solu} \bs_{j}=\mbox{exp}((-1)^{j-1}(\Og_{\ty{D},j}x-\Og^n_{\ty{D},j}t))\bD^+_{\ty{D},j},\quad j=1,2,\\
& \bM_1=\mbox{exp}(\Og_{\ty{D},1}x-\Og^n_{\ty{D},1}t)\bC^-_{\ty{D},1}\cdot\bG_{\ty{D}}\cdot\bD^-_{\ty{D},2}\mbox{exp}(-\Og_{\ty{D},2}x+\Og^n_{\ty{D},2}t), \\
& \bM_2=\mbox{exp}(-\Og_{\ty{D},2}x+\Og^n_{\ty{D},2}t)\bC^-_{\ty{D},2}\cdot(-\bG_{\ty{D}}^{\st})\cdot\bD^-_{\ty{D},1}\mbox{exp}(\Og_{\ty{D},1}x-\Og^n_{\ty{D},1}t),
\end{align}
\end{subequations}
where
\begin{subequations}
\label{CD-GDD}
\begin{align}
& \bC^+_{\ty{D},j}=(c_{j,1},c_{j,2},\ldots,c_{j,N_j})^{\st},\quad
\bD^+_{\ty{D},j}=(d_{j,1},d_{j,2},\ldots,d_{j,N_j})^{\st},\quad j=1,2, \\
& \label{CD-GDD-soli} \bC^-_{\ty{D},j}=\mbox{Diag}(c_{j,1},c_{j,2},\ldots,c_{j,N_j}),\quad
\bD^-_{\ty{D},j}=\mbox{Diag}(d_{j,1},d_{j,2},\ldots,d_{j,N_j}),\quad j=1,2,\\
& \bG_{\ty{D}}=(g_{i,j})_{N_1\times N_2},\quad g_{i,j}=\frac{1}{k_{1,i}-k_{2,j}}.
\end{align}
\end{subequations}
In this case, one can derive Cauchy matrix-type soliton solutions.

When
\begin{align}
\Og_j=\Og_{\ty{J},j}=\left(\begin{array}{cccccc}
k_{j,1} & 0    & 0   & \cdots & 0   & 0 \\
1   & k_{j,1}  & 0   & \cdots & 0   & 0 \\
\vdots &\vdots &\vdots &\vdots &\vdots &\vdots \\
0   & 0    & 0   & \cdots & 1   & k_{j,1}
\end{array}\right)_{N_j \times N_j},\quad j=1,2,
\end{align}
we have
\begin{subequations}
\label{MJJ-solu}
\begin{align}
& \label{rJ-solu} \br_{j}=\mbox{exp}((-1)^{j-1}(\Og_{\ty{J},j}x-\Og^n_{\ty{J},j}t))\bC^+_{\ty{J},j},\quad j=1,2,\\
& \label{sJ-solu} \bs_{j}=\mbox{exp}((-1)^{j-1}(\Og^{\st}_{\ty{J},j}x-(\Og^{\st}_{\ty{J},j})^nt))\bD^+_{\ty{J},j},\quad j=1,2,\\
& \bM_1=\mbox{exp}(\Og_{\ty{J},1}x-\Og^n_{\ty{J},1}t)\bC^-_{\ty{J},1}\cdot\bG_{\ty{J}}\cdot\bD^-_{\ty{J},2}\mbox{exp}(-\Og^{\st}_{\ty{J},2}x+(\Og^{\st}_{\ty{J},2})^nt), \\
& \bM_2=\mbox{exp}(-\Og_{\ty{J},2}x+\Og^n_{\ty{J},2}t)\bC^-_{\ty{J},2}\cdot(-\bG_{\ty{J}}^{\st})\cdot\bD^-_{\ty{J},1}\mbox{exp}(\Og^{\st}_{\ty{J},1}x-(\Og^{\st}_{\ty{J},1})^nt),
\end{align}
\end{subequations}
where
\begin{subequations}
\label{CD-GJJ}
\begin{align}
\label{CD-J}
& \bC^+_{\ty{J},j}=c_{j,1}\bar{\bI},\quad
\bD^+_{\ty{J},j}=d_{j,1}\bar{\bI},\quad \bC^-_{\ty{J},j}=c_{j,1}\bI,\quad
\bD^-_{\ty{J},j}=d_{j,1}\bI, \\
& \bG_{\ty{J}}=(g_{i,j})_{N_1\times N_2},
\qquad g_{i,j}=\mathrm{C}^{i-1}_{i+j-2}\frac{(-1)^{i+1}}{(k_{1,1}-k_{2,1})^{i+j-1}},
\end{align}
\end{subequations}
with $\bar{\bI}=(1,0,\ldots,0)^{\st}$ and $\mathrm{C}^{i}_{j}=\frac{j!}{i!(j-i)!},~~(j\geq i).$
In this case, one can derive Cauchy matrix-type Jordan block solutions.

In the following two sections, we show the Cauchy matrix reduction technique
to consider the Cauchy matrix solutions for the local and nonlocal
equations \eqref{cmKdV}, \eqref{n-sG}, \eqref{n-pNLS-1} and \eqref{n-pmKdV-1}.
For this purpose, we take $N_1=N_2=N$. In addition, for convenience, we denote
\begin{align}
k_i=k_{2,i}, \quad c_i=c_{2,i}, \quad d_i=d_{2,i}, \quad \vartheta_i=c_id_i,
\end{align}
with $i=1,2,\ldots,N$.

\section{Cauchy matrix solutions for the equations \eqref{cmKdV} and \eqref{n-sG}}

In this section, we are going to discuss Cauchy matrix solutions to the
local and nonlocal cmKdV equation \eqref{cmKdV} and the local and nonlocal csG equation \eqref{n-sG}.
We will design suitable constraints on the pairs $(\Og_1,\Og_2)$, $(\br_1,\br_2)$, $(\bs_1,\bs_2)$
and $(\bM_1, \bM_2)$ in the determining equation set \eqref{DES-C} so that \eqref{AKNS-solu}
coincides with the reductions \eqref{3AKNS-red} and \eqref{-1AKNS-red}.

\subsection{Local and nonlocal cmKdV equation}

To derive solutions for the local and nonlocal cmKdV equation
\eqref{cmKdV}, we restrict ourselves in Theorem \ref{Thm-1} with $n=3$.
For the Cauchy matrix solution of equation \eqref{cmKdV}, we show the results in Theorem \ref{so-nlcmKdV}.
\begin{Thm}
\label{so-nlcmKdV}
The function
\begin{align}\label{nlcmKdV-u-solu}
u(x,t)=\bs^{\st}_2(\bI-\bM_2\bM_1)^{-1}\br_2,
\end{align}
solves the local and nonlocal cmKdV equation \eqref{cmKdV},
where the entities satisfy \eqref{DES-C} $(n=3)$ and simultaneously obey the constraints
\begin{align}\label{nlcmKdV-M1M12}
\br_1=\varepsilon \bT \br^*_2(\sigma x,\sigma t),\quad \bs_1=\varepsilon \bT^{\st^{-1}}\bs_2^*(\sigma x,\sigma t),\quad
\bM_1=-\delta\sigma\bT \bM^*_2(\sigma x,\sigma t)\bT^*,
\end{align}
in which $\bT\in \mathbb{C}_{N\times N}$ is a constant matrix satisfying
\begin{align}\label{nlcmKdV-at-eq}
\Og_1\bT+\sigma\bT\Og^*_2=0,\quad \bC^+_1=\varepsilon\bT\bC_2^{+^*},\quad
\bD^+_1=\varepsilon\bT^{\st^{-1}}\bD_2^{+^*},\quad \varepsilon^2=\varepsilon^{*^2}=\delta.
\end{align}
\end{Thm}

\begin{proof}

This is by direct computation. Under the assumption \eqref{nlcmKdV-at-eq}, starting from \eqref{rsj} we have
\begin{align}
\br_1(x,t)=& \exp{\Bigl(\Og_1 x-\Og_1^{3}t \Bigr)}\bC_1^+ \nn \\
=& \exp{\Bigl(\bT(-\Og^*_2 \sigma x+\Og^{*^3}_2\sigma t)\bT^{-1}\Bigr)}\bC_1^+ \nn \\
=& \bT\exp{\Bigl(-\Og^*_2 \sigma x+\Og^{*^3}_2\sigma t\Bigr)}\bT^{-1}\bC_1^+ \nn \\
=& \varepsilon \bT \br^*_2(\sigma x,\sigma t).
\end{align}
A similar computation as above yields $\bs_1(x,t)=\varepsilon \bT^{\st^{-1}} \bs^*_2(\sigma x,\sigma t)$.
Substituting $\Og_1=-\sigma\bT\Og^*_2\bT^{-1}$ into the first equation in \eqref{DES-M12-C},
applying the second equation on \eqref{DES-M12-C} and noticing the solvability of the Sylvester equation,
we arrive at the relation $\bM_1(x,t)=-\delta\sigma\bT \bM^*_2(\sigma x,\sigma t)\bT^*$.

In terms of the relations \eqref{nlcmKdV-M1M12} and \eqref{nlcmKdV-at-eq}, we immediately get
\begin{align*}
v(x,t)=& \bs^{\st}_1(\bI-\bM_1\bM_2)^{-1}\br_1 \nn \\
=& \varepsilon^2 \bs^{*^{\st}}_2(\sigma x,\sigma t)(\bI-\bM_2^*(\sigma x,\sigma t)\bM^*_1(\sigma x,\sigma t))^{-1}\br^*_2(\sigma x,\sigma t) \nn \\
=& \delta u^*(\sigma x,\sigma t),
\end{align*}
which is nothing but the reduction relation \eqref{3AKNS-red}. Therefore, we complete the verification.
\end{proof}

\subsubsection{Exact solutions}

According to the Theorem \ref{so-nlcmKdV}, we know that solution to the local and nonlocal cmKdV equation \eqref{cmKdV} can be
expressed as
\begin{subequations}
\begin{align}\label{cmKdV-u-solu-1}
u(x,t)=\bs^{\st}_2(\bI+\delta\sigma\bM_2\bT \bM^*_2(\sigma x,\sigma t)\bT^*)^{-1}\br_2,
\end{align}
where the components $\br_2$ and $\bs_2$ read
\begin{align}
\label{rsj-mKdV}
\br_{2}=\mbox{exp}(-\Og_2x+\Og^3_2t)\bC^+_2, \quad \bs_{2}=\mbox{exp}(-\Og^{\st}_2x+(\Og_2^{\st})^3t)\bD^+_2,
\end{align}
and $\bM_2$ and $\bT$ are determined by
\begin{align}
\label{DES-M12-C-mKdV}
\Og_2 \bM_2\bT+\sigma\bM_2\bT\Og^*_2=\varepsilon\br_2\, \bs^{*^{\st}}_2(\sigma x,\sigma t).
\end{align}
\end{subequations}

\noindent {\bf Remark 1}: \textit{Denoting $\bM_2\bT\rightarrow\bM_2$, we can simplify solution \eqref{cmKdV-u-solu-1} together with
\eqref{DES-M12-C-mKdV} as
\begin{subequations}
\label{cmKdV-u-solu-2}
\begin{align}
& u(x,t)=\bs^{\st}_2(\bI+\delta\sigma\bM_2 \bM^*_2(\sigma x,\sigma t))^{-1}\br_2, \\
& \label{DES-M12-C-mKdV-nt} \Og_2 \bM_2+\sigma\bM_2\Og^*_2=\varepsilon\br_2\, \bs^{*^{\st}}_2(\sigma x,\sigma t).
\end{align}
\end{subequations}}

We now pay attention to the equations \eqref{rsj-mKdV} and \eqref{cmKdV-u-solu-2}
to give the expressions of solutions for the local and nonlocal cmKdV equation \eqref{cmKdV} with different $(\varepsilon,\sigma)$.
To obtain soliton solutions, we take $\Og_2=\Og_{\ty{D},2}$. We
list the soliton solution formulae with different $(\varepsilon,\sigma)$ in Table 1.
\begin{center}
\footnotesize \setlength{\tabcolsep}{8pt}
\renewcommand{\arraystretch}{1.5}
\begin{tabular}[htbp]{|l|l|l|l|}
\hline
$(\varepsilon,\sigma)$
& Soliton solutions & Sylvester equation
\\
\hline
(1,1)
& $u(x,t)=\bs^{\st}_2(\bI+\bM^{(1)}_{2}\bM^{(1)^*}_{2})^{-1}\br_2 $ & $\Og_{\ty{D},2} \bM^{(1)}_{2}+\bM^{(1)}_{2}\Og^*_{\ty{D},2}=\br_2\, \bs^{*^{\st}}_2$ \\
\hline
(i,1)
& $u(x,t)=\bs^{\st}_2(\bI-\bM^{(2)}_{2}\bM^{{(2)}^*}_{2})^{-1}\br_2 $  & $\Og_{\ty{D},2} \bM^{(2)}_{2}+\bM^{(2)}_{2}\Og^*_{\ty{D},2}=i\br_2\, \bs^{*^{\st}}_2$\\
\hline
(1,-1)
& $u(x,t)=\bs^{\st}_2(\bI-\bM^{(3)}_{2}\bM^{(3)^*}_{2}(-x,-t))^{-1}\br_2 $  & $\Og_{\ty{D},2} \bM^{(3)}_{2}-\bM^{(3)}_{2}\Og^*_{\ty{D},2}=\br_2\, \bs^{*^{\st}}_2(-x,-t)$\\
\hline
(i,-1)
& $u(x,t)=\bs^{\st}_{2}(\bI+\bM^{(4)}_{2}\bM^{(4)^*}_{2}(-x,-t))^{-1}\br_2 $  & $\Og_{\ty{D},2} \bM^{(4)}_{2}-\bM^{(4)}_{2}\Og^*_{\ty{D},2}=i\br_2\, \bs^{*^{\st}}_2(-x,-t)$\\
\hline
\end{tabular}
\end{center}
\begin{center}
\begin{minipage}{11cm}{\footnotesize
{Table 1. \emph{Soliton solutions for local and nonlocal cmKdV equation \eqref{cmKdV}}}}
\end{minipage}
\end{center}
In Table 1, $\br_{2},~\bs_{2}$ and $\bM^{(j)}_2=\bM^{(j)}_2(x,t) ~(j=1,2,3,4)$ are given by
\begin{subequations}
\label{rsM-mKdV-D}
\begin{align}
& \label{rsj-mKdV-D}
\br_{2}=\mbox{exp}(-\Og_{\ty{D},2}x+\Og^3_{\ty{D},2}t)\bC^+_{\ty{D},2}, \quad \bs_{2}=\mbox{exp}(-\Og^{\st}_{\ty{D},2}x+(\Og_{\ty{D},2}^{\st})^3t)\bD^+_{\ty{D},2}, \\
& \bM^{(1)}_{2}=-i\bM^{(2)}_{2}=\mbox{exp}(-\Og_{\ty{D},2}x+\Og^3_{\ty{D},2}t)\bC^-_{\ty{D},2}\cdot
\bG_{\ty{D}}^{(12)}\cdot\bD^{-^*}_{\ty{D},2}\mbox{exp}(-\Og^*_{\ty{D},2}x+\Og^{*^3}_{\ty{D},2}t), \\
& \bM^{(3)}_{2}=-i\bM^{(4)}_{2}=\mbox{exp}(-\Og_{\ty{D},2}x+\Og^3_{\ty{D},2}t)\bC^-_{\ty{D},2}\cdot
\bG_{\ty{D}}^{(34)}\cdot\bD^{-^*}_{\ty{D},2}\mbox{exp}(\Og^*_{\ty{D},2}x-\Og^{*^3}_{\ty{D},2}t),
\end{align}
\end{subequations}
where $\bC^{\pm}_{\ty{D},2},~\bD^{\pm}_{\ty{D},2}$ are defined in \eqref{CD-GDD} and
\begin{subequations}
\label{G1234}
\begin{align}
\label{G12} \bG^{(12)}_{\ty{D}}=(g^{(12)}_{i,j})_{N\times N},\quad g^{(12)}_{i,j}=\frac{1}{k_i+k_j^*}, \\
\label{G34} \bG^{(34)}_{\ty{D}}=(g^{(34)}_{i,j})_{N\times N},\quad g^{(34)}_{i,j}=\frac{1}{k_i-k_j^*}.
\end{align}
\end{subequations}

For deriving the Cauchy matrix-type Jordan block solutions, we take $\Og_2=\Og_{\ty{J},2}$ and summarize the Jordan block solutions in Table 2,
\begin{center}
\footnotesize \setlength{\tabcolsep}{8pt}
\renewcommand{\arraystretch}{1.5}
\begin{tabular}[htbp]{|l|l|l|l|}
\hline
$(\varepsilon,\sigma)$
& Jordan block solutions & Sylvester equation
\\
\hline
(1,1)
 & $u(x,t)=\bs^{\st}_2(\bI+\wh{\bM}^{(1)}_2\wh{\bM}^{(1)^*}_2)^{-1}\br_2$ &
$\Og_{\ty{J},2} \wh{\bM}^{(1)}_2+\wh{\bM}^{(1)}_2\Og^*_{\ty{J},2}=\br_2\, \bs^{*^{\st}}_2$ \\
\hline
(i,1)
& $u(x,t)=\bs^{\st}_2(\bI-\wh{\bM}^{(2)}_2\wh{\bM}^{(2)^*}_2)^{-1}\br_2$ &
$\Og_{\ty{J},2} \wh{\bM}^{(2)}_2+\wh{\bM}^{(2)}_2\Og^*_{\ty{J},2}=i\br_2\, \bs^{*^{\st}}_2$ \\
\hline
(1,-1)
 &$u(x,t)=\bs^{\st}_2(\bI-\wt{\bM}^{(1)}_2\wt{\bM}^{(1)^*}_2(-x,-t))^{-1}\br_2$ &
$\Og_{\ty{J},2} \wt{\bM}^{(1)}_2-\wt{\bM}^{(1)}_2\Og^*_{\ty{J},2}=\br_2\, \bs^{*^{\st}}_2(-x,-t)$ \\
\hline
(i,-1)
 &$u(x,t)=\bs^{\st}_2(\bI+\wt{\bM}^{(2)}_2\wt{\bM}^{(2)^*}_2(-x,-t))^{-1}\br_2 $ &
$\Og_{\ty{J},2} \wt{\bM}^{(2)}_2-\wt{\bM}^{(2)}_2\Og^*_{\ty{J},2}=i\br_2\, \bs^{*^{\st}}_2(-x,-t)$ \\
\hline
\end{tabular}
\end{center}
\begin{center}
\begin{minipage}{11cm}{\footnotesize
{Table 2. \emph{Jordan block solutions for local and nonlocal cmKdV equation \eqref{cmKdV}}}}
\end{minipage}
\end{center}
in which
\begin{subequations}
\label{rsj-mKdV-OM}
\begin{align}
& \label{rsj-mKdV-J}
\br_{2}=\mbox{exp}(-\Og_{\ty{J},2}x+\Og^3_{\ty{J},2}t)\bC^+_{\ty{J},2}, \quad \bs_{2}=\mbox{exp}(-\Og^{\st}_{\ty{J},2}x+(\Og_{\ty{J},2}^{\st})^3t)\bD^+_{\ty{J},2}, \\
& \wh{\bM}^{(1)}_2=-i\wh{\bM}^{(2)}_2=\mbox{exp}(-\Og_{\ty{J},2} x+
\Og_{\ty{J},2}^3 t)\bC_{\ty{J},2}^- \cdot \wh{\bG}_{\ty{J}}^{\st}\cdot\bD_{\ty{J},2}^{-^*}\mbox{exp}(-\Og_{\ty{J},2}^{*^{\st}}x
+(\Og_{\ty{J},2}^{*^{\st}})^3t), \\
& \label{OM-mKdV-J} \wt{\bM}^{(1)}_2=-i\wt{\bM}^{(2)}_2=\mbox{exp}(-\Og_{\ty{J},2} x
+\Og_{\ty{J},2}^3 t)\bC_{\ty{J},2}^- \cdot\wt{\bG}_{\ty{J}}^{\st}\cdot\bD_{\ty{J},2}^
{-^*}\mbox{exp}(\Og_{\ty{J},2}^{*^{\st}}x-(\Og_{\ty{J},2}^{*^{\st}})^3 t),
\end{align}
\end{subequations}
where $\bC^{\pm}_{\ty{J},2},~\bD^{\pm}_{\ty{J},2}$ are defined in \eqref{CD-J} and
\begin{subequations}
\begin{align}
& \wh{\bG}_{\ty{J}}=(\wh{g}_{i,j})_{N\times N}, \quad \wh{g}_{i,j}=\left(\mathrm{C}^{i-1}_{i+j-2}\frac{(-1)^{j+1}}{(k_1+k^*_1)^{i+j-1}}\right), \\
& \wt{\bG}_{\ty{J}}=(\wt{g}_{i,j})_{N\times N}, \quad
\wt{g}_{i,j}=\left(\mathrm{C}^{i-1}_{i+j-2}\frac{(-1)^{j+1}}{(k_1-k^*_1)^{i+j-1}}\right).
\end{align}
\end{subequations}

\subsubsection{Dynamics}

We proceed now with presenting some explicit solutions for the local and nonlocal cmKdV equation \eqref{cmKdV}.
Moreover, we shall identify their dynamic properties.
For the sake of brevity, we introduce some notations
\begin{align}
\label{FHftheta}
\xi_{i}=2(k_i^3t-k_ix), \quad e^{\theta_{ij}}=\frac{1}{(k^*_i+k_j)^2}, \quad e^{\epsilon_{ij}}=\frac{1}{(k^*_i-k_j)^2},
\end{align}
with $i,j=1,2,\ldots,N$.

When $N=1$, 1-soliton solutions for equation \eqref{cmKdV} can be described as
\begin{subequations}
\label{cmKdV-u-soli}
\begin{align}
& \label{cmKdV-soli-11} u_{11,(\varepsilon=1,\sigma=1)}=\frac{\vartheta_1}{
 e^{-\xi_{1}}+|\vartheta_1|^{2}e^{\xi_{1}^*+\theta_{11}}}, \\
& \label{cmKdV-soli-12} u_{11,(\varepsilon=i,\sigma=1)}=\frac{\vartheta_1}{
 e^{-\xi_{1}}-|\vartheta_1|^{2}e^{\xi_{1}^*+\theta_{11}}}, \\
& \label{cmKdV-soli-21} u_{11,(\varepsilon=1,\sigma=-1)}=\frac{\vartheta_1}{
e^{-\xi_{1}}+|\vartheta_1|^{2}e^{-\xi_{1}^*+\epsilon_{11}}}, \\
& \label{cmKdV-soli-22} u_{11,(\varepsilon=i,\sigma=-1)}=\frac{\vartheta_1}{
e^{-\xi_{1}}-|\vartheta_1|^{2}e^{-\xi_{1}^*+\epsilon_{11}}}.
\end{align}
\end{subequations}
These solutions can recover to the ones obtained in \cite{SAPM-2018} by taking
simple transformations.

Because these solutions are
complex form, we turn to present analysis of the dynamics of $|u_{11}|^2$. With no loss of generality
we just discuss solutions \eqref{cmKdV-soli-11} and \eqref{cmKdV-soli-21}.
We start with the 1-soliton solution \eqref{cmKdV-soli-11}. Taking
\begin{align}
\label{k1-com}
k_1=\alpha+i\beta
\end{align}
into solution \eqref{cmKdV-soli-11} and by direct calculation, we get
\begin{align}
\label{var1si1-mKdV}
|u_{11}|^2_{(\varepsilon=1,\sigma=1)}=\alpha^2\mbox{sech}^2\left[2\alpha\big(x-(\alpha^2-3\beta^2)t+\hbar\big)\right],
\end{align}
here and hereafter $\hbar=\frac{1}{2\alpha} \ln\frac{2\alpha}{|\vartheta_1|}$.
This solution describes a stable uni-directional travelling wave, which travels with a fixed amplitude $\alpha^2$, constant velocity $\alpha^2-3\beta^2$,
and top trajectory $x=(\alpha^2-3\beta^2)t-\hbar$.
When $\alpha^2=3\beta^2$, \eqref{var1si1-mKdV} implies a stationary
soliton wave. Fig. 1 depicts a stationary soliton wave and two moving soliton waves.


\begin{center}
\begin{picture}(120,100)
\put(-150,-23){\resizebox{!}{3.5cm}{\includegraphics{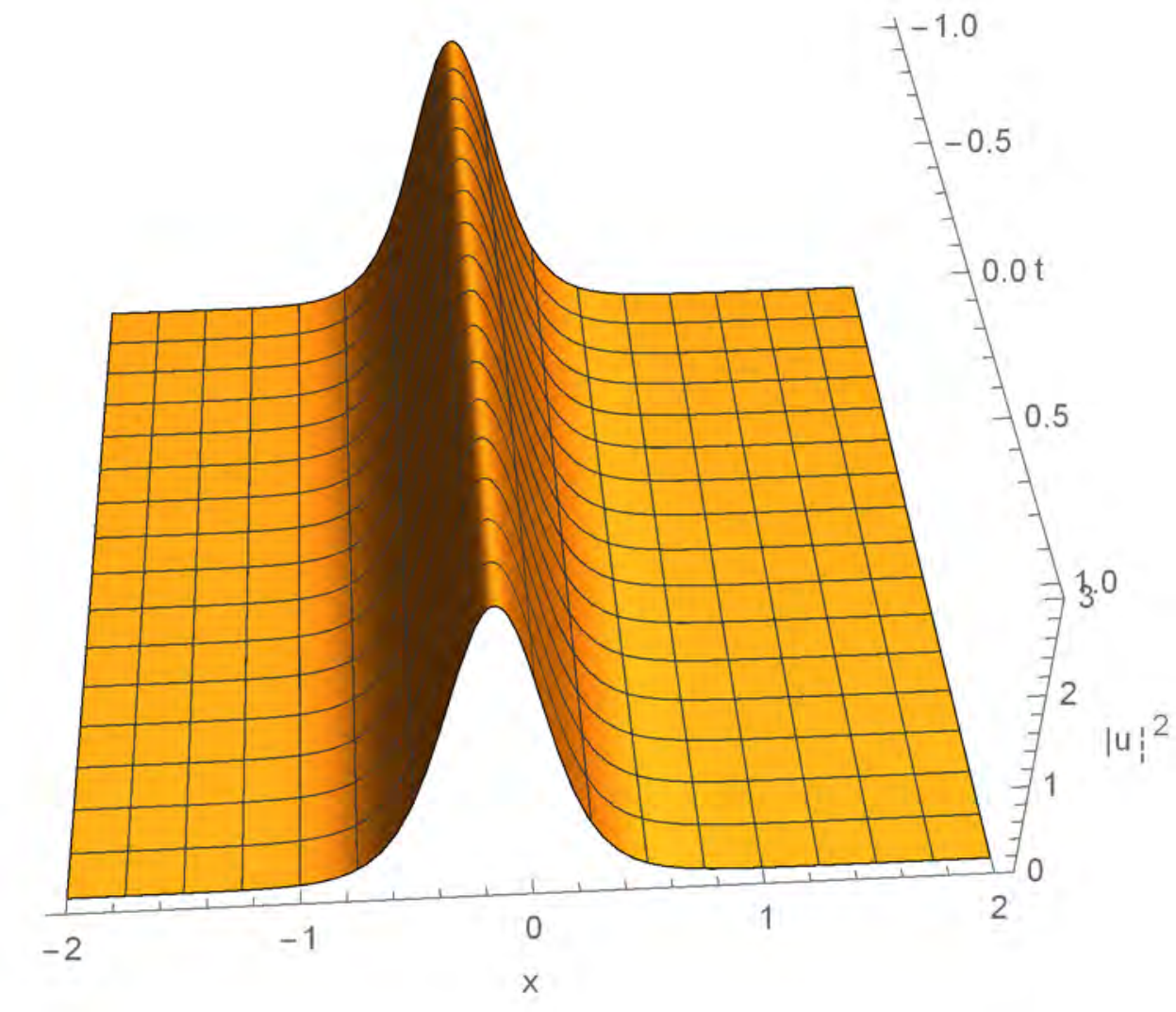}}}
\put(10,-23){\resizebox{!}{3.5cm}{\includegraphics{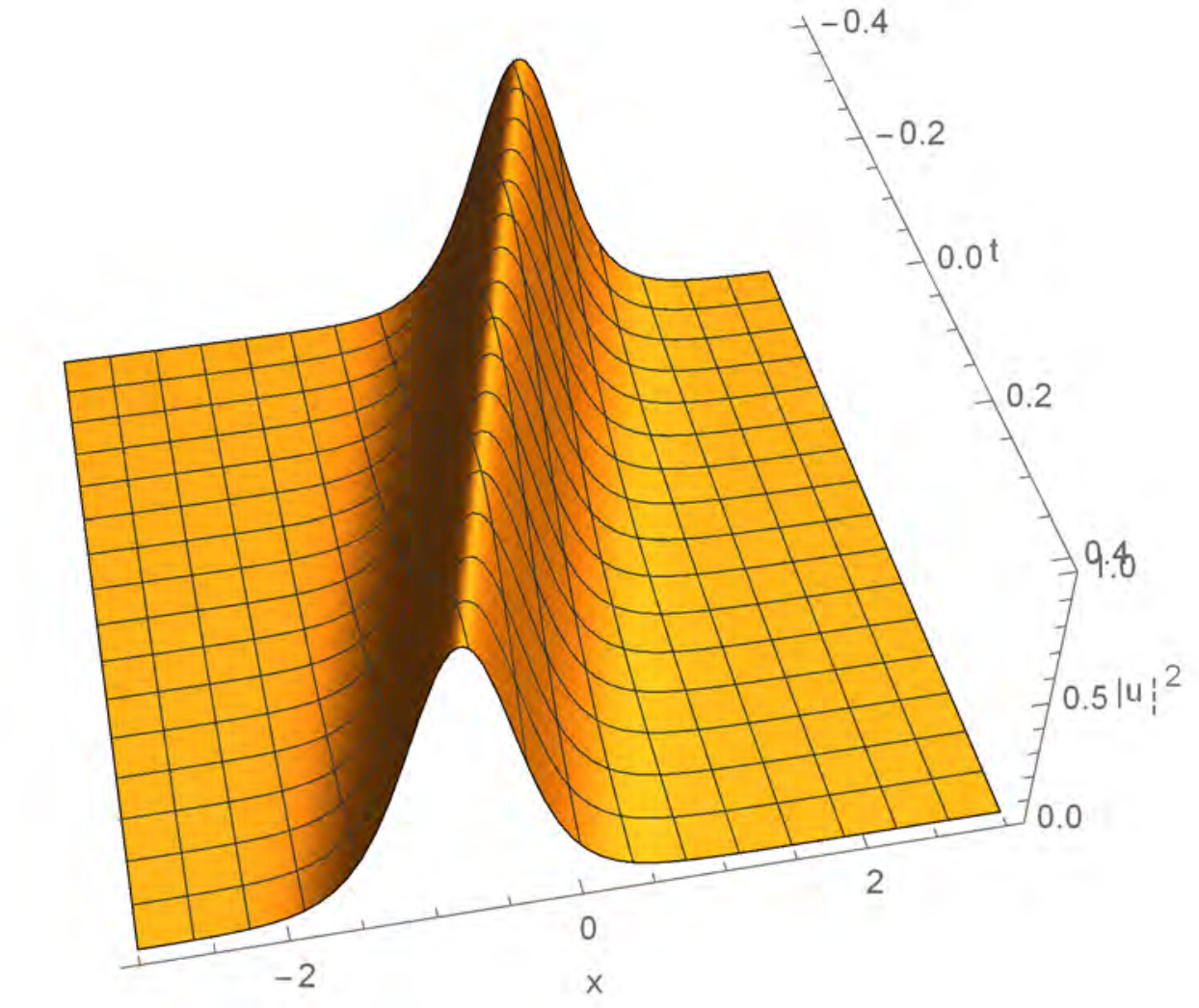}}}
\put(150,-23){\resizebox{!}{3.5cm}{\includegraphics{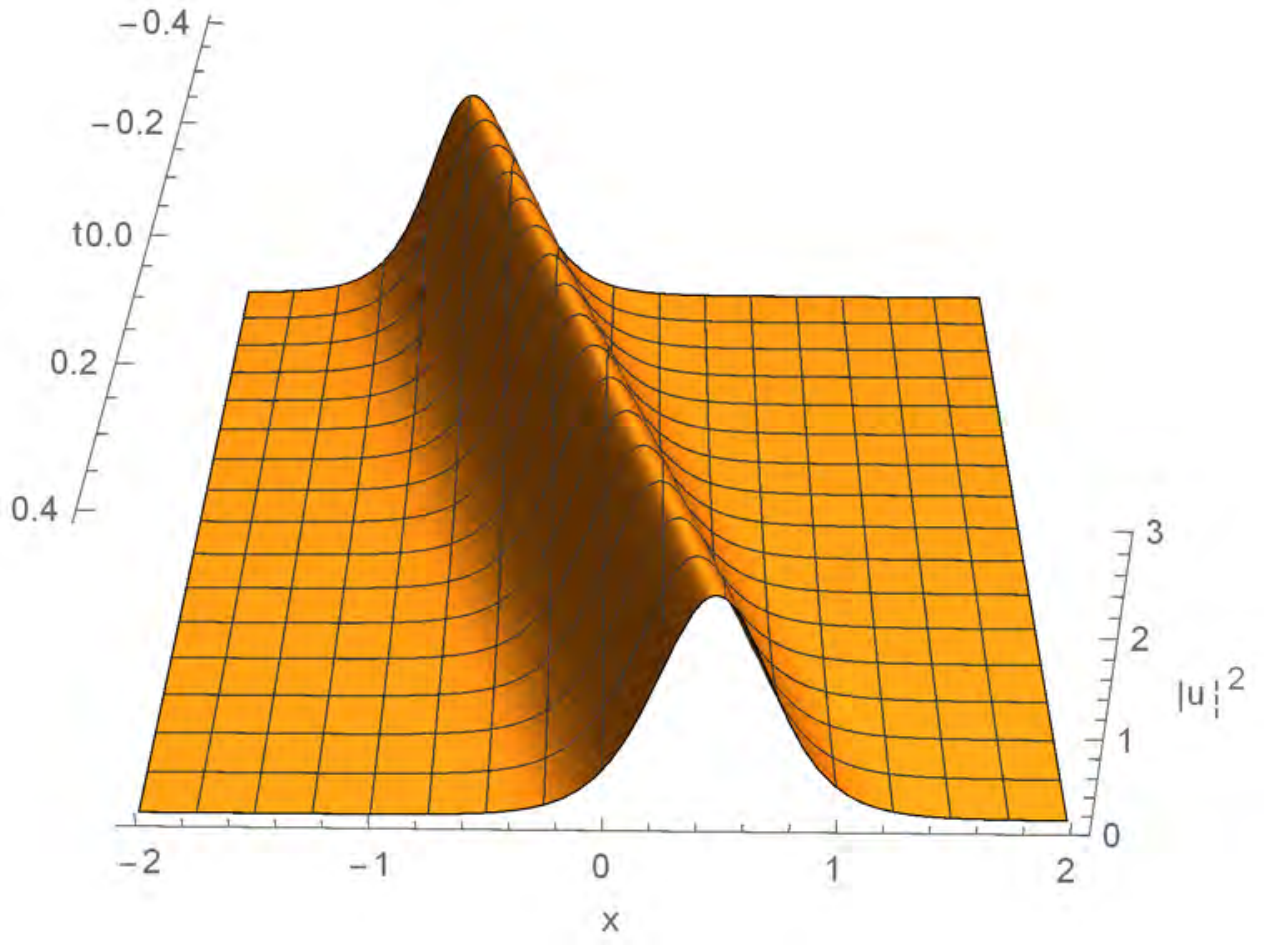}}}
\end{picture}
\end{center}
\vskip 20pt
\begin{center}
\begin{minipage}{16cm}{\footnotesize
\quad\qquad\qquad\qquad(a)\qquad\qquad\qquad\qquad \qquad\quad \qquad \qquad \qquad (b) \qquad\qquad \qquad \qquad\qquad\qquad\qquad  (c)\\
{\bf Fig. 1} shape and motion of 1-soliton solution \eqref{var1si1-mKdV} for $c_1=d_1=1+i$: (a)
a stationary wave for $k_1=\sqrt{3}+i$. (b) a moving wave for $k_1=1+i$.
(c) a moving wave for $k_1=1.5+0.5i$.}
\end{minipage}
\end{center}


To continue, we next consider solution \eqref{cmKdV-soli-21}.
We substitute \eqref{k1-com} into solution \eqref{cmKdV-soli-21} and have
\begin{align}
\label{var1si-1-mKdV}
|u_{11}|^2_{(\varepsilon=1,\sigma=-1)}=\frac{16|\beta^2\vartheta_1|^2e^{4\alpha(\alpha^2t-3\beta^2t-x)}}
{|\vartheta_1|^4+16\beta^4-8|\beta \vartheta_1|^2\cos[4\beta(x-3\alpha^2t+\beta^2t)]}.
\end{align}
Since the involvement of cosine function in denominator, there is a quasi-periodic phenomenon.
When $|\vartheta_1|^2=4\beta^2$, solution \eqref{var1si-1-mKdV} has singularities along with straight lines
\begin{align}
x(t)=(3\alpha^2-\beta^2)t+\frac{\kappa\pi}{2\beta}, \quad \kappa\in \mathbb{Z}.
\end{align}
While when $|\vartheta_1|^2\neq 4\beta^2$, \eqref{var1si-1-mKdV} is nonsingular and reaches its extrema along with straight lines
\begin{align}
\label{xt-ext}
x(t)=(3\alpha^2-\beta^2)t+\frac{1}{4\beta}\left(\gamma+2\kappa\pi-
\arcsin\frac{\alpha(|\vartheta_1|^4+16\beta^4)}{8|\beta\vartheta_1|^2\sqrt{\alpha^2+\beta^2}}\right), \quad \kappa\in \mathbb{Z},
\end{align}
where $\sin \gamma=\frac{\alpha}{\sqrt{\alpha^2+\beta^2}}$.
The velocity is $3\alpha^2-\beta^2$.
When $3\alpha^2=\beta^2$, \eqref{var1si-1-mKdV} is a stationary solution. For a given $t$, $|u_{11}|^2_{(\varepsilon=1,\sigma=-1)}$
tends to zero for $(\alpha>0, x\rightarrow \infty)$ and  $(\alpha<0, x\rightarrow -\infty)$, respectively.
We depict this solution in Fig. 2.


\begin{center}
\begin{picture}(120,100)
\put(-150,-23){\resizebox{!}{3.5cm}{\includegraphics{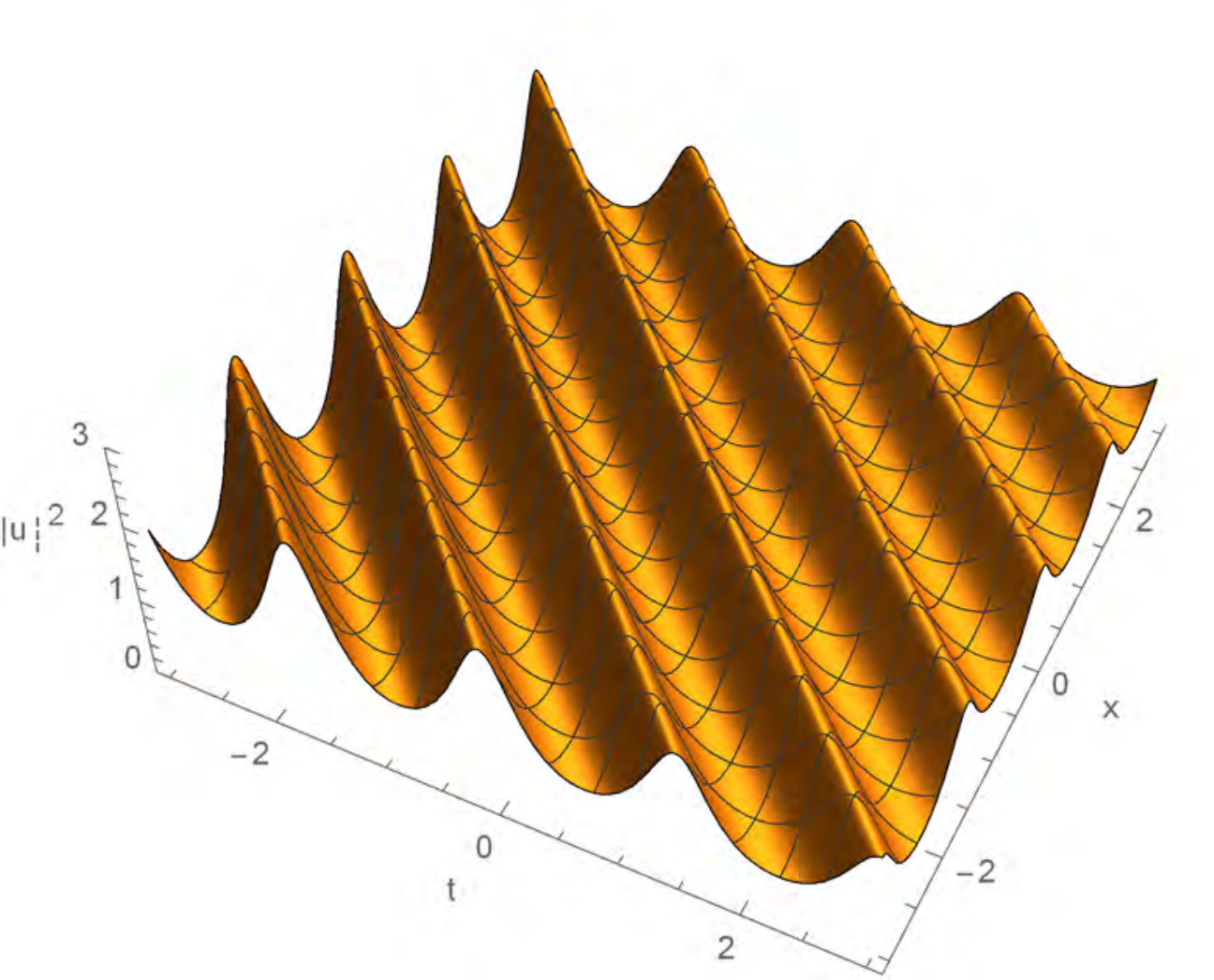}}}
\put(10,-23){\resizebox{!}{3.5cm}{\includegraphics{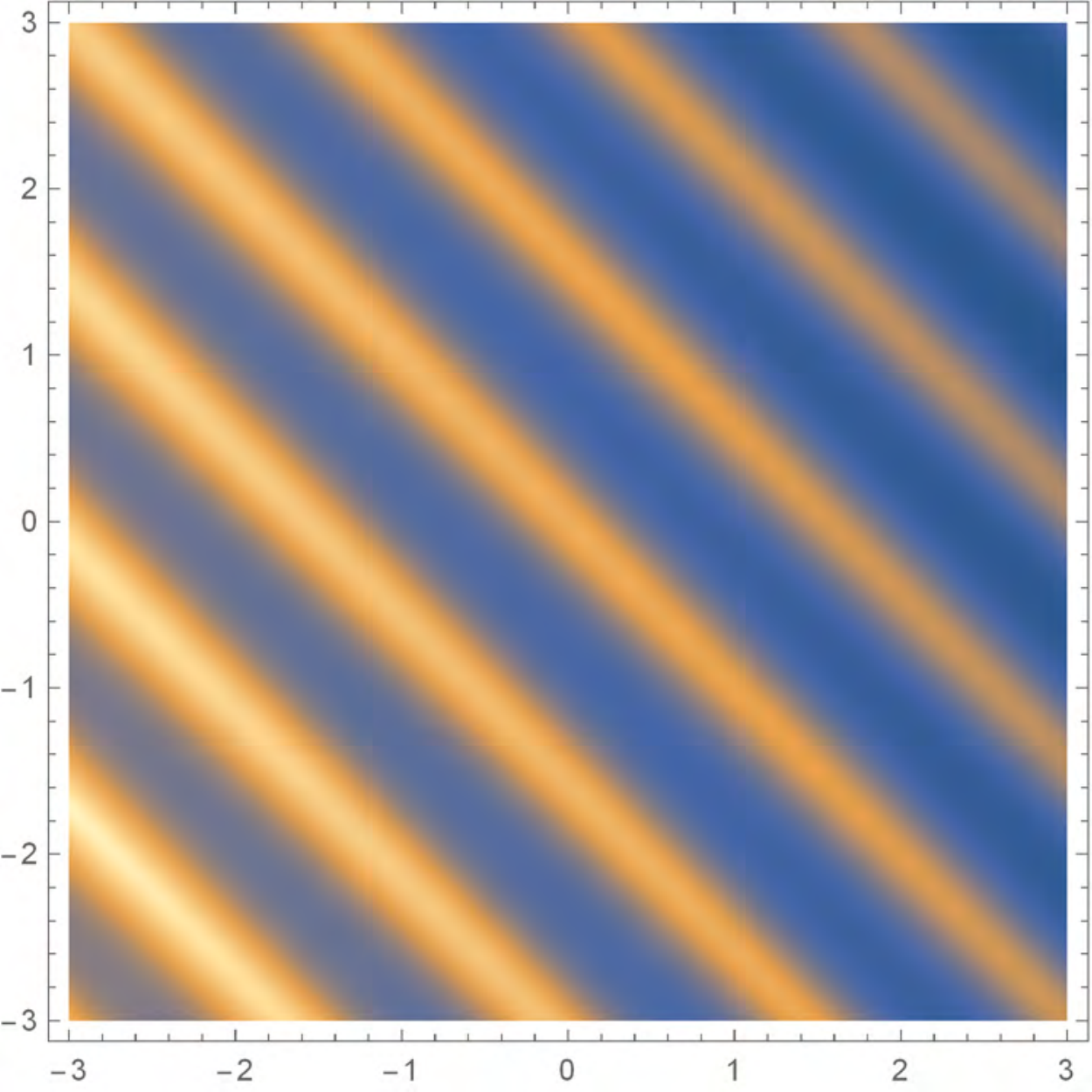}}}
\put(150,-23){\resizebox{!}{3.5cm}{\includegraphics{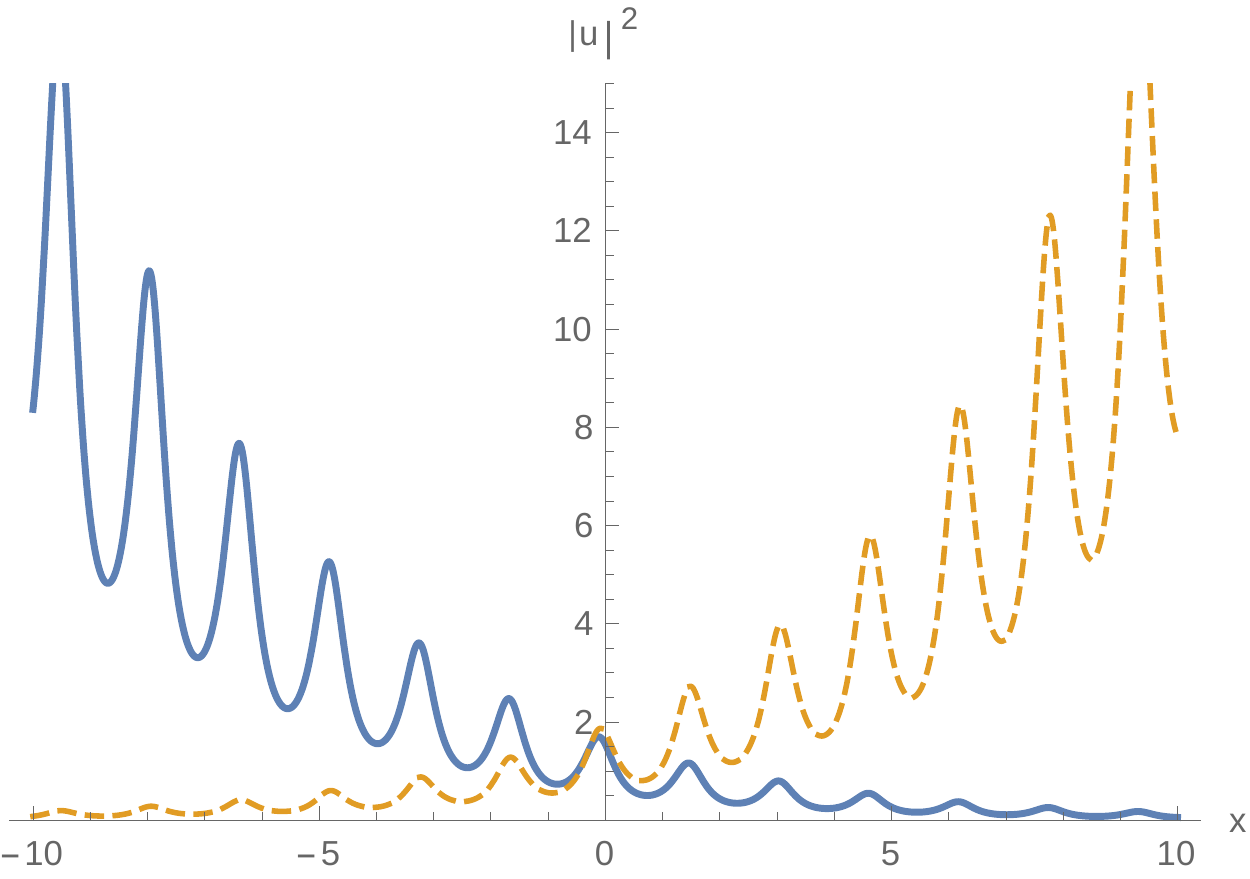}}}
\end{picture}
\end{center}
\vskip 20pt
\begin{center}
\begin{minipage}{16cm}{\footnotesize
\quad\qquad\qquad\qquad(a)\qquad\qquad\qquad\qquad \qquad\qquad\quad \qquad \quad  (b) \qquad\qquad \qquad\qquad\qquad\qquad\qquad\quad  (c)\\
{\bf Fig. 2} shape and motion of 1-soliton solution \eqref{var1si-1-mKdV} for $c_1=d_1=1$: (a)
3D-plot for $k_1=0.01+i$. (b) a contour plot of (a) with range
$x\in [-3, 3]$ and $t\in [-3,3]$. (c) waves in solid line and dotted line stand for curves at $t=0.1$ with $k_1=0.06+i$ and $k_1=-0.06+i$, respectively.}
\end{minipage}
\end{center}


When $N=2$, with no loss of generality we just list 2-soliton solutions of \eqref{cmKdV} in the case of $(\varepsilon=1,\sigma=1)$ and $(\varepsilon=1,\sigma=-1)$.
They are $u_{12,(\varepsilon=1,\sigma=1)}=\frac{u_{\ty{D},2}}{u_{\ty{D},1}}$ with
\begin{subequations}
\label{q12-mKdV-2ss}
\begin{align}
& u_{\ty{D},1}=1+\sum_{i=1}^2\sum_{j=1}^2\vartheta^*_i\vartheta_je^{\xi^*_{i}+\xi_{j}+\theta_{ij}}+
|\vartheta_1\vartheta_2|^2|k_1-k_2|^4e^{\xi_{1}+\xi_{1}^*+\xi_{2}+\xi_{2}^*+\theta_{11}+\theta_{12}+\theta_{21}+\theta_{22}}, \\
& u_{\ty{D},2}=\sum_{i=1}^2\vartheta_ie^{\xi_{i}}+\vartheta_1\vartheta_2(k_1-k_2)^2e^{\xi_{1}+\xi_{2}}\left(\sum_{i=1}^2\vartheta_i^*e^{\xi_{i}^*+\theta_{i1}+\theta_{i2}}\right),
\end{align}
\end{subequations}
and $u_{12,(\varepsilon=1,\sigma=-1)}=\frac{u_{\ty{D},4}}{u_{\ty{D},3}}$ with
\begin{subequations}
\label{q34-mKdV-2ss}
\begin{align}
& u_{\ty{D},3}=1+\sum_{i=1}^2\sum_{j=1}^2\vartheta^*_i\vartheta_je^{-\xi^*_{i}+\xi_{j}+\epsilon_{ij}}+
|\vartheta_1\vartheta_2|^2|k_1-k_2|^4e^{\xi_{1}-\xi_{1}^*+\xi_{2}-\xi_{2}^*+\epsilon_{11}+\epsilon_{12}+\epsilon_{21}+\epsilon_{22}}, \\
& u_{\ty{D},4}=\sum_{i=1}^2\vartheta_ie^{\xi_{i}}+\vartheta_1\vartheta_2(k_1-k_2)^2e^{\xi_{1}+\xi_{2}}\left(\sum_{i=1}^2\vartheta_i^*e^{-\xi_{i}^*+\epsilon_{i1}+\epsilon_{i2}}\right).
\end{align}
\end{subequations}
We depict, respectively, these two solitary waves in Fig. 3 and Fig. 4.


\begin{center}
\begin{picture}(120,100)
\put(-80,-23){\resizebox{!}{4.0cm}{\includegraphics{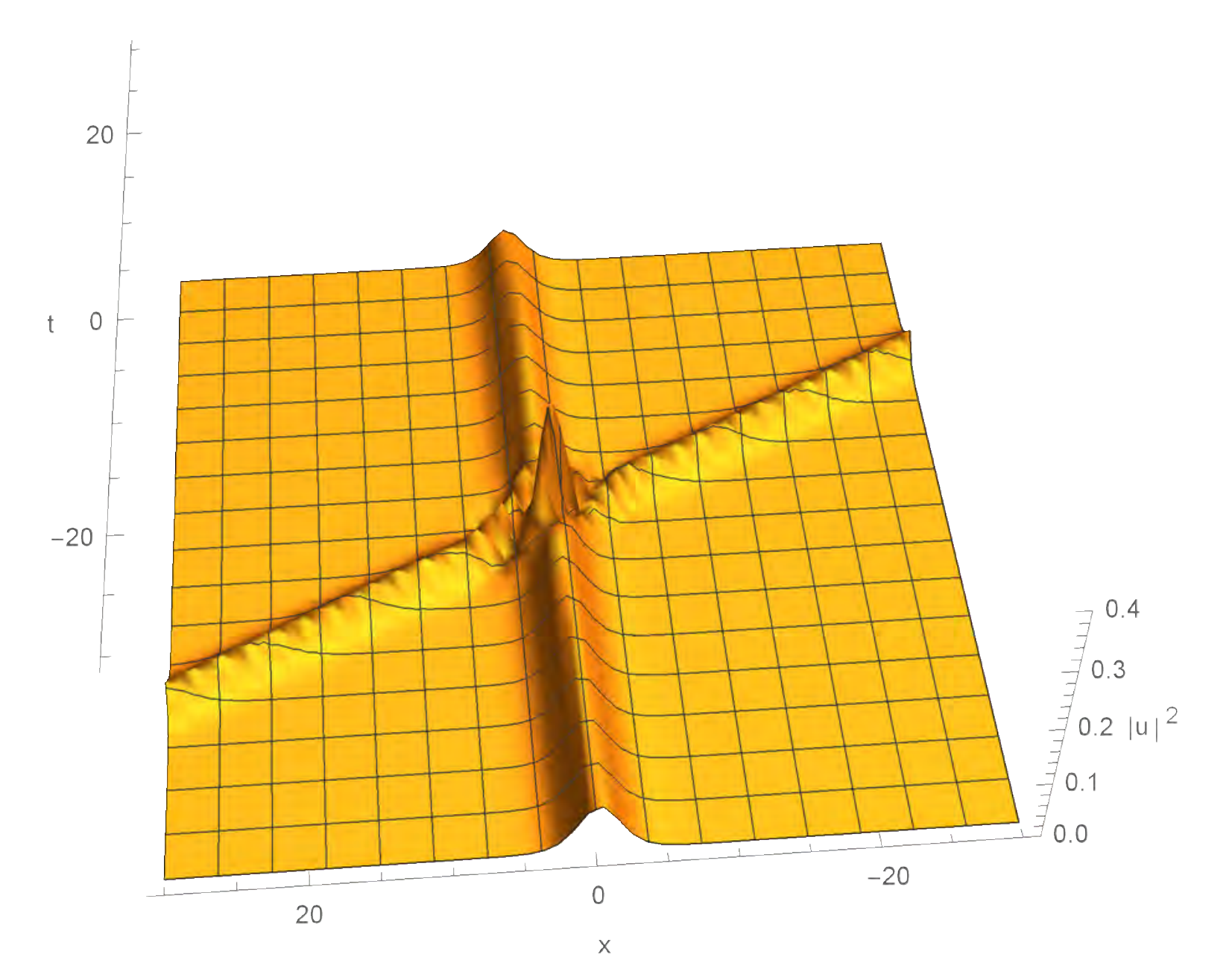}}}
\put(100,-23){\resizebox{!}{3.5cm}{\includegraphics{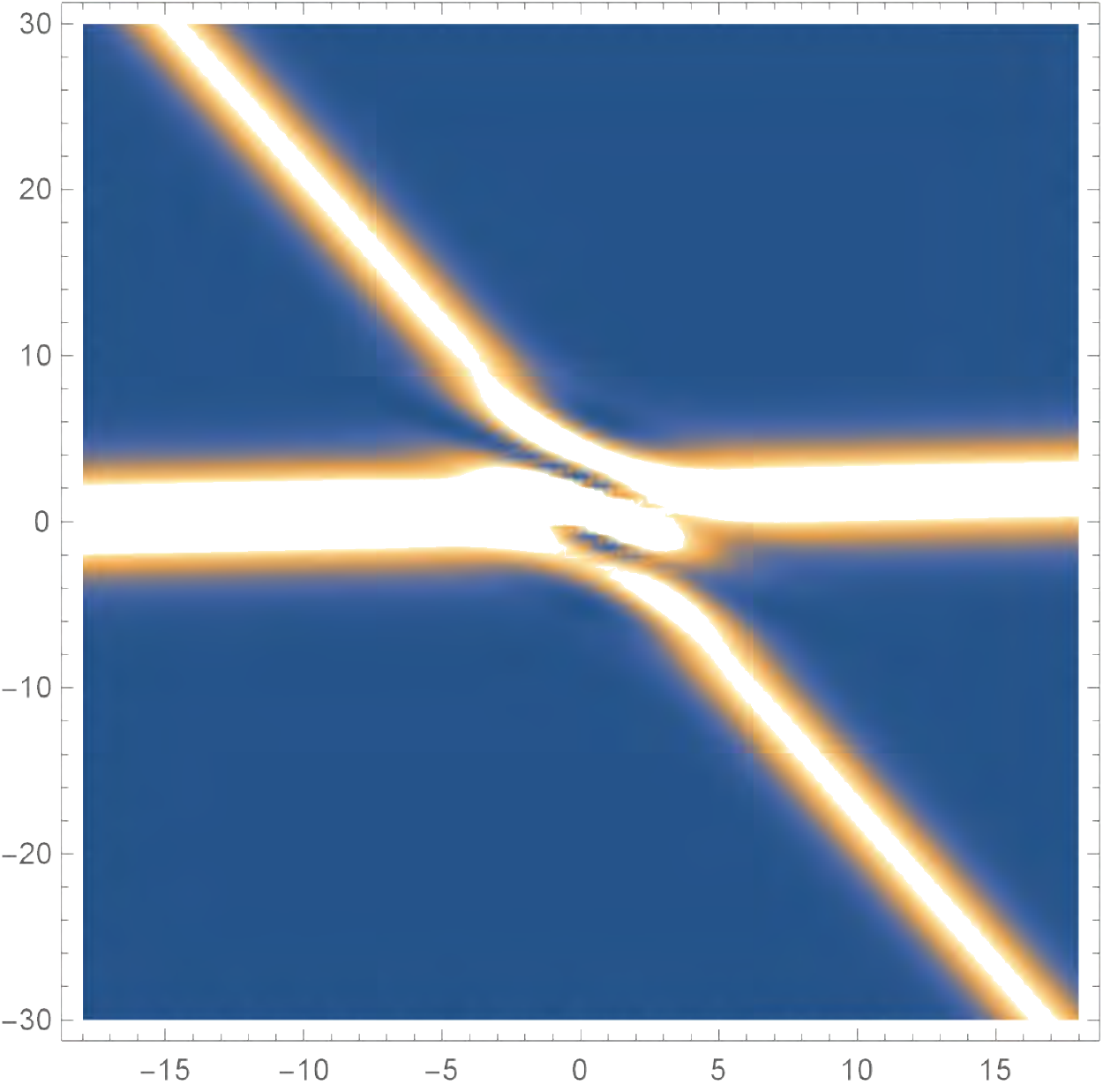}}}
\end{picture}
\end{center}
\vskip 20pt
\begin{center}
\begin{minipage}{16cm}{\footnotesize
\quad\qquad\qquad\qquad\qquad\qquad\qquad\qquad \qquad(a)\quad \qquad \qquad\qquad\qquad \qquad \qquad\qquad\qquad(b)\qquad \qquad}
\end{minipage}
\end{center}
\vskip 10pt
\begin{center}
\begin{picture}(120,80)
\put(-150,-23){\resizebox{!}{3.5cm}{\includegraphics{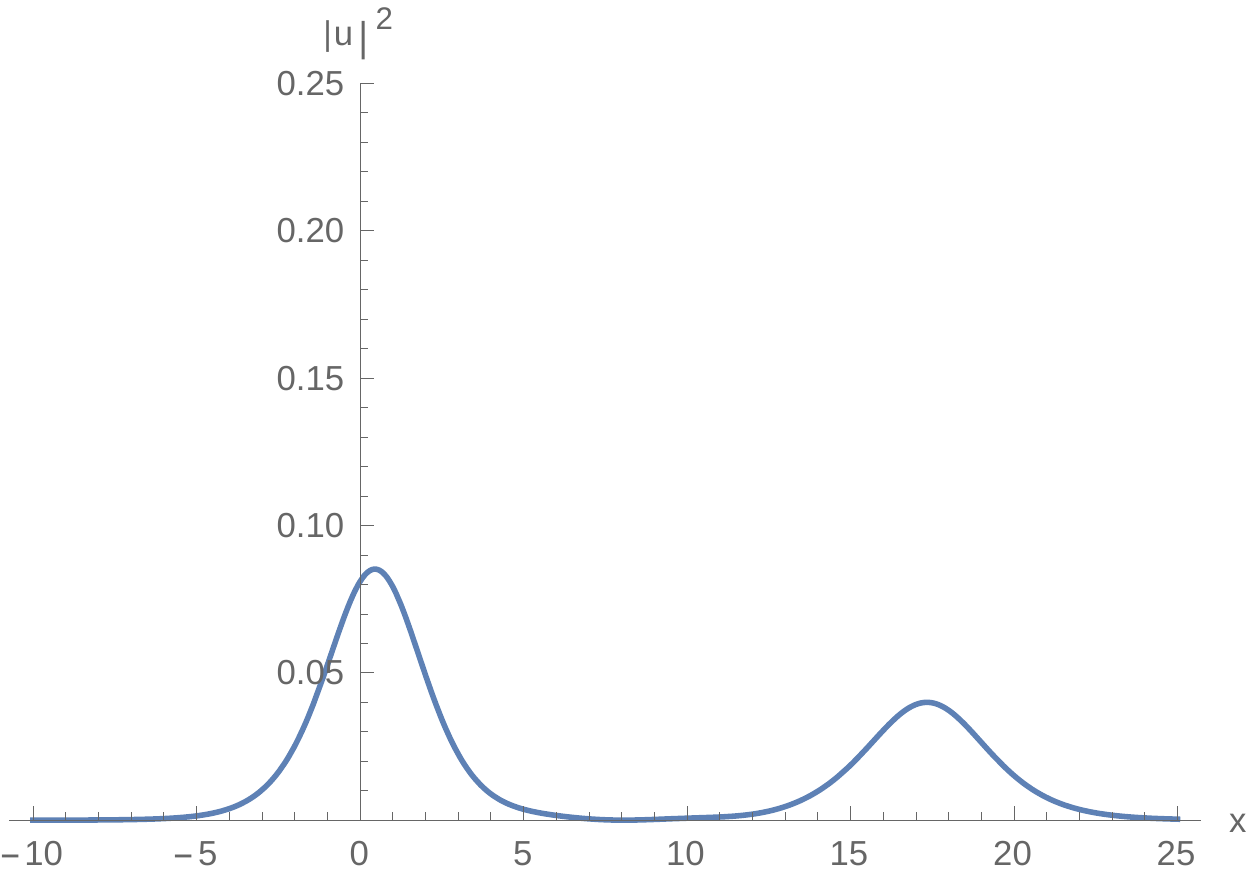}}}
\put(10,-23){\resizebox{!}{3cm}{\includegraphics{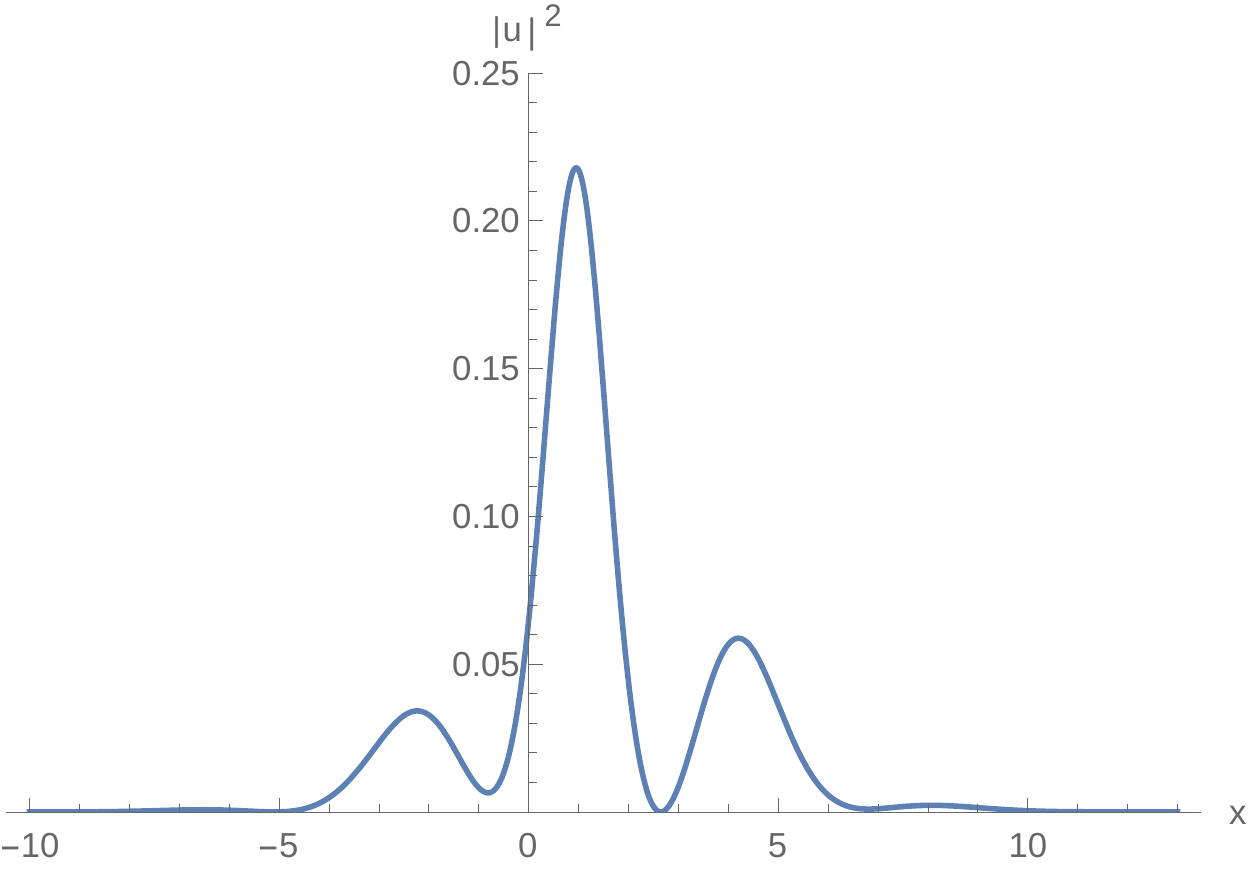}}}
\put(150,-23){\resizebox{!}{3cm}{\includegraphics{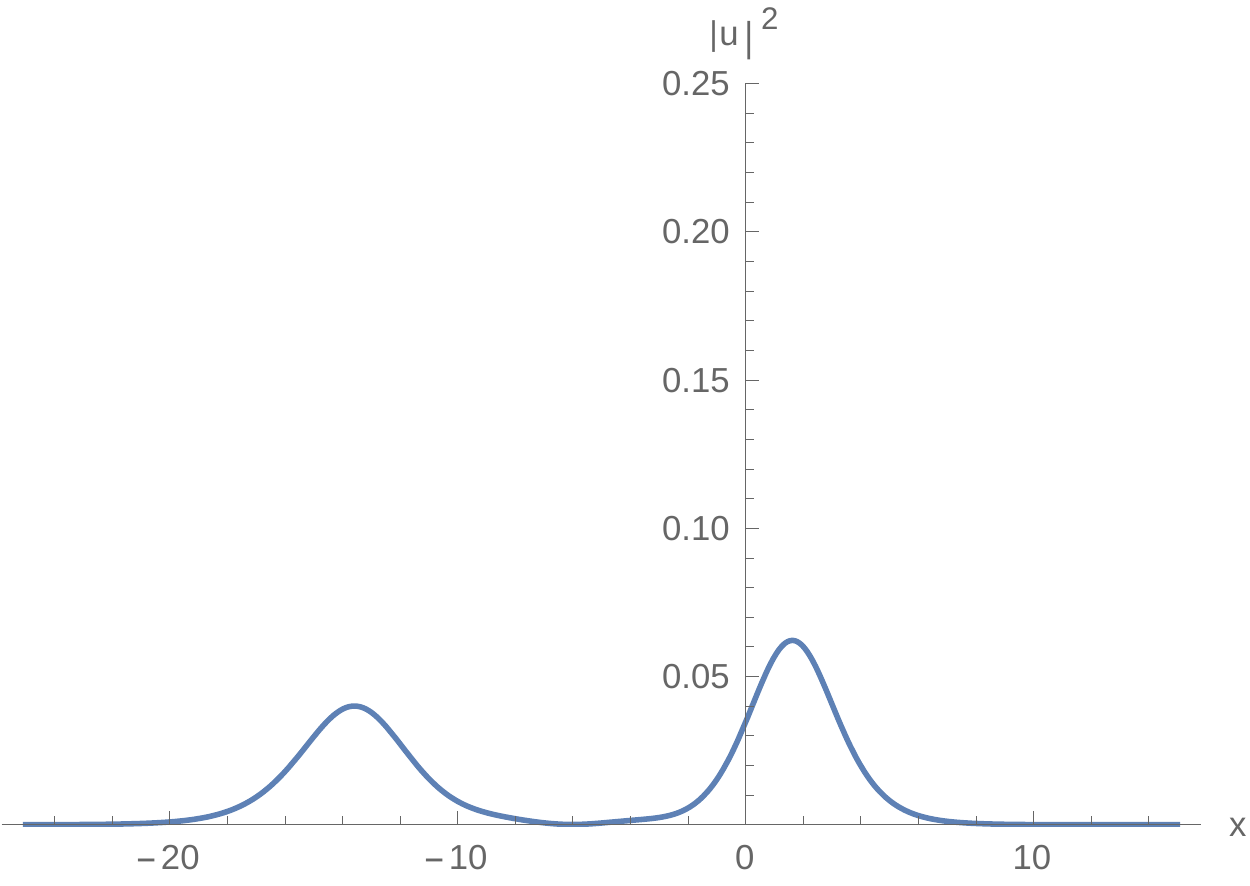}}}
\end{picture}
\end{center}
\vskip 20pt
\begin{center}
\begin{minipage}{16cm}{\footnotesize
\quad\qquad\qquad\qquad(c)\qquad\qquad\qquad\qquad \qquad\quad \qquad \qquad \quad (d) \qquad\qquad \qquad \qquad\qquad\qquad\qquad \quad (e)\\
{\bf Fig. 3} (a) shape and motion of 2-soliton solutions $|u_{12}|^2_{(\varepsilon=1,\sigma=1)}$ given by \eqref{q12-mKdV-2ss}
for $k_1=0.2+0.8i,~k_2=0.25+0.1i$ and $c_1=c_2=d_1=d_2=1$. (b) a contour plot of (a) with range $x\in [-30, 30]$ and $t\in [-18,18]$.
(c) 2D-plot of (a) at $t=-8$. (d) 2D-plot of (a) at $t=0$. (e) 2D-plot of (a) at $t=8$.}
\end{minipage}
\end{center}



\begin{center}
\begin{picture}(120,100)
\put(-150,-23){\resizebox{!}{4.5cm}{\includegraphics{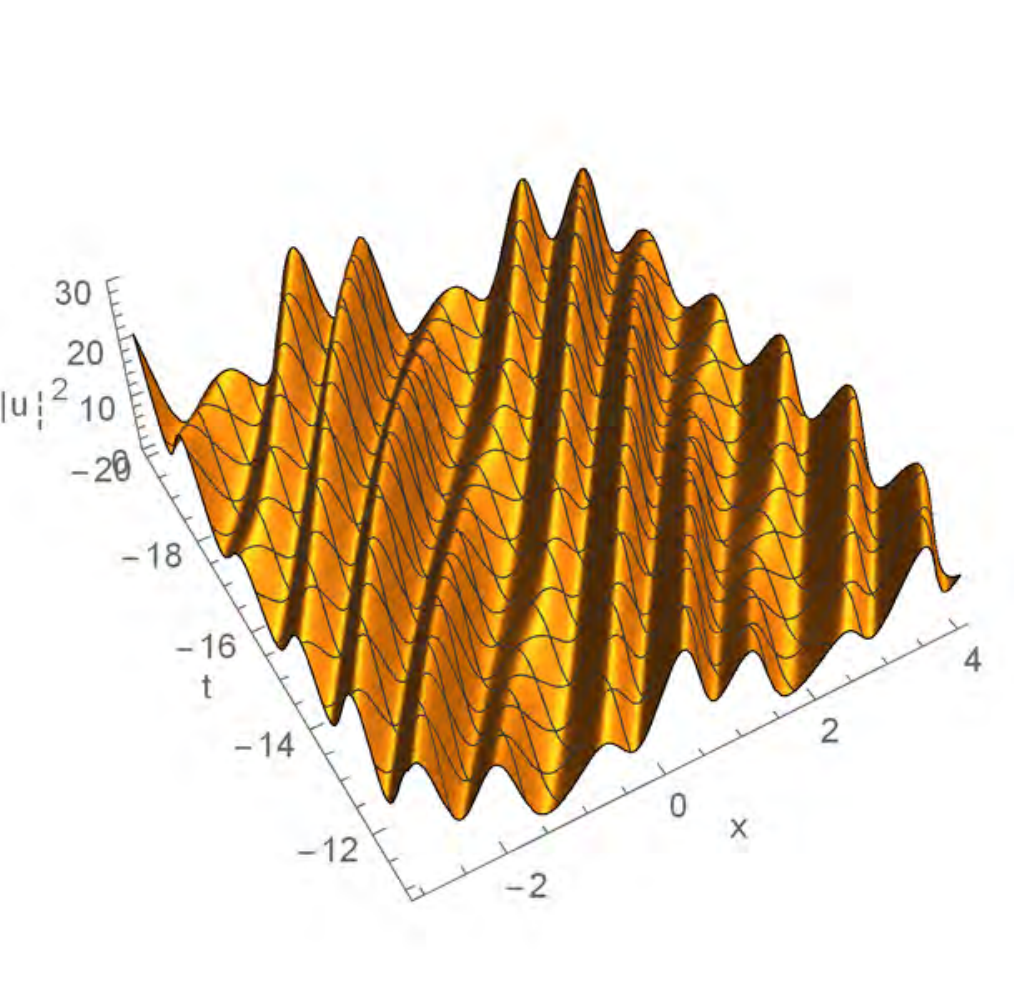}}}
\put(10,-23){\resizebox{!}{3.5cm}{\includegraphics{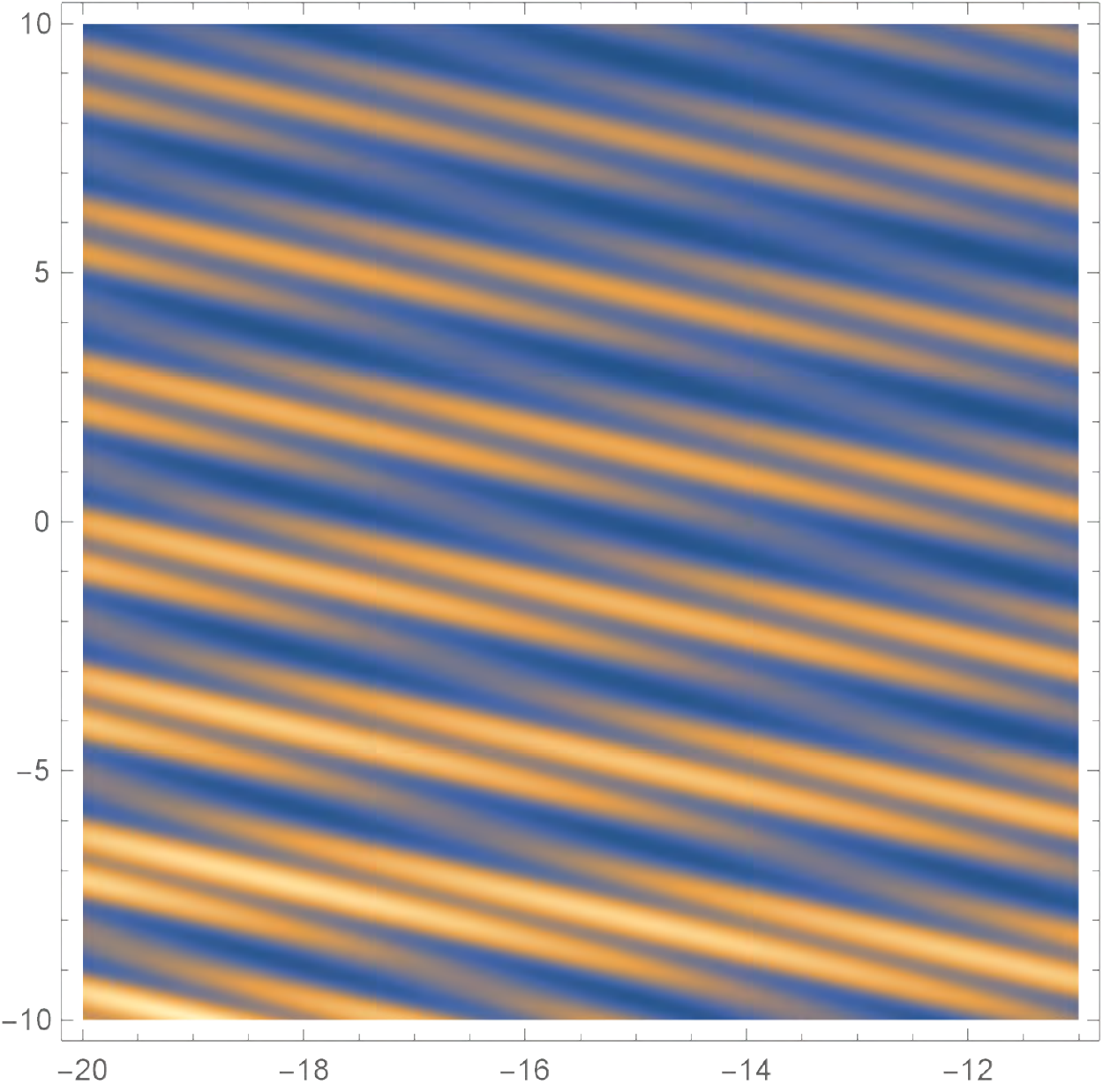}}}
\put(150,-23){\resizebox{!}{3.5cm}{\includegraphics{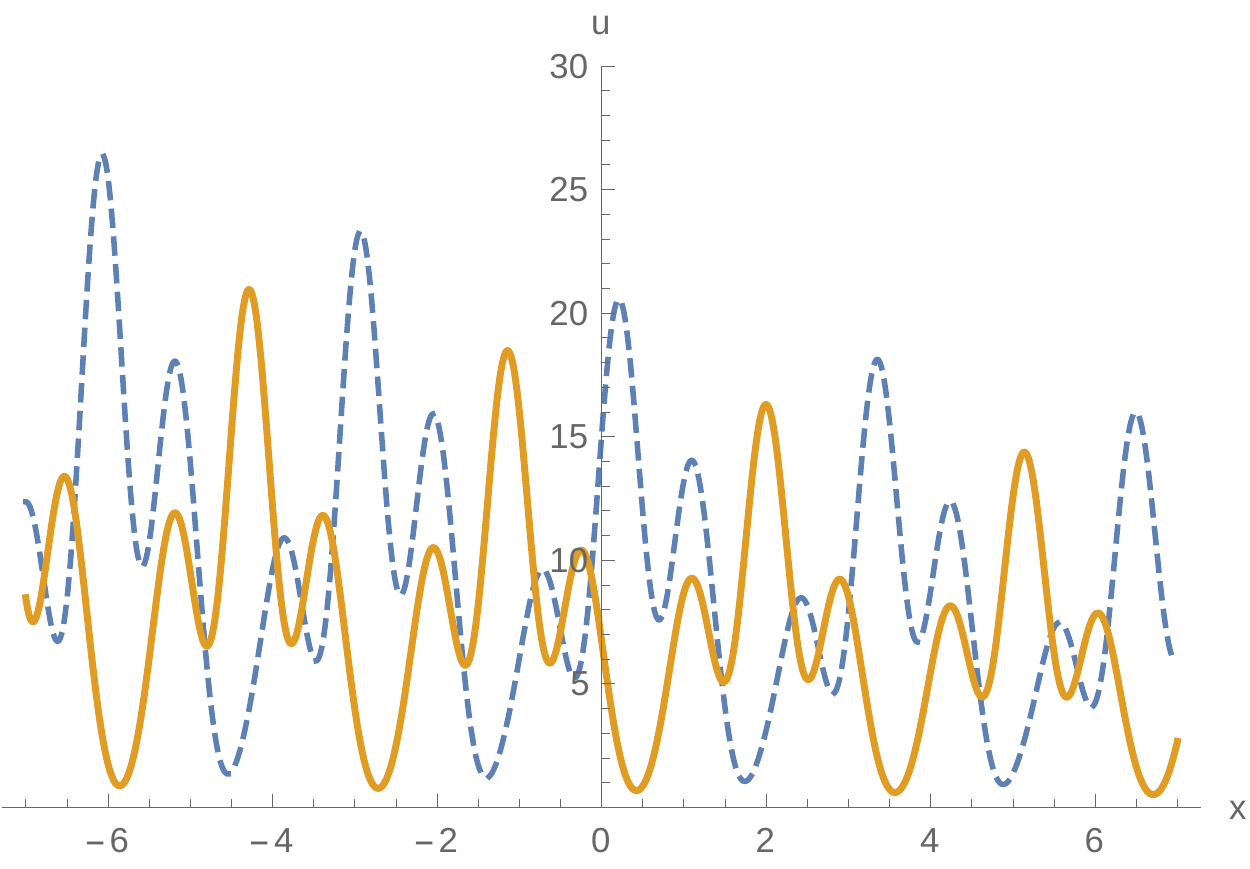}}}
\end{picture}
\end{center}
\vskip 20pt
\begin{center}
\begin{minipage}{16cm}{\footnotesize
\quad\qquad\qquad\qquad(a)\qquad\qquad\qquad\qquad \qquad\quad \qquad \qquad \quad (b) \qquad \qquad \qquad\qquad\qquad\qquad \qquad\quad (c)\\
{\bf Fig. 4} (a) shape and motion of 2-soliton solutions $|u_{12}|^2_{(\varepsilon=1,\sigma=-1)}$ given by \eqref{q34-mKdV-2ss} for $k_1=0.3+i,~k_2=0.01+0.5i,~
c_1=2$ and $c_2=d_1=d_2=1$. (b) a contour plot of (a) with range
$x\in [-10, 10]$ and $t\in [-20,-11]$. (c) waves in soild and dotted line stand for plot (a) at $t=-5$ and $t=-11$, respectively.}
\end{minipage}
\end{center}


We finally pay attention to the simplest Jordan block solutions ($N=2$). We still take $(\varepsilon=1,\sigma=1)$ and $(\varepsilon=1,\sigma=-1)$ as two examples.
The Jordan block
solution to equation \eqref{cmKdV} for $(\varepsilon=1,\sigma=1)$ reads $u_{13,(\varepsilon=1,\sigma=1)}=\frac{u_{\ty{J},2}}{u_{\ty{J},1}}$, in which
\begin{subequations}
\label{u12-cmKdV}
\begin{align}
& u_{\ty{J},1}=(|\vartheta_1|^2e^{\xi_1+\xi_1^*}+e^{-2\theta_{11}})^2+|\vartheta_1|^2e^{\xi_1+\xi_1^*-3\theta_{11}}
\big((1-2e^{\theta_{11}})^2\nn \\
&\qquad\quad+(1-e^{\theta_{11}})(\tau^2+\tau^{*^2})-2e^{\theta_{11}}|\tau|^2+|\tau|^4\big), \\
&u_{\ty{J},2}=\vartheta_1e^{\xi_1-4\theta_{11}}\big(1-|\vartheta_1|^2e^{\xi_1+\xi_1^*+2\theta_{11}}\big(1-4e^{\theta_{11}}+\tau^{*}(\tau+\tau^{*}
+|\tau|^2e^{-\frac{1}{2}\theta_{11}})\nn\\
&\qquad\quad-2\sqrt{2}\tau\sinh((\theta_{11}+\ln 2)/2)\big)\big),
\end{align}
\end{subequations}	
where $\tau=3k_1^2t-x$. While, in the case of $(\varepsilon=1,\sigma=-1)$, the Jordan block
solution to equation \eqref{cmKdV} reads $u_{13,(\varepsilon=1,\sigma=-1)}=\frac{u_{\ty{J},4}}{u_{\ty{J},3}}$ with
\begin{subequations}
\label{u34-cmKdV}
\begin{align}
& u_{\ty{J},3}=(|\vartheta_1|^2e^{\xi_{1}-\xi^*_1}+e^{-2\epsilon_{11}})^2+|\vartheta_1|^2e^{\xi_{1}-\xi^*_1-3\epsilon_{11}}
\big(4e^{2\epsilon_{11}}+e^{\epsilon_{11}}((\tau+\tau^*)^2+4|\tau|^2)\nn \\
&\qquad\quad+|1+\tau^2|^2+2e^{\frac{1}{2}\epsilon_{11}}(1+2e^{\epsilon_{11}}+|\tau|^2)(\tau+\tau^*)\big), \\
&u_{\ty{J},4}=\vartheta_1e^{\xi_1-4\epsilon_{11}}\big(1+|\vartheta_1|^2e^{\xi_{1}-\xi^*_1+2\epsilon_{11}}
\big(1+4e^{\epsilon_{11}}+2e^{\frac{1}{2}\epsilon_{11}}(\tau+2\tau^*)+3|\tau|^2\nn \\
&\qquad\quad+\tau^{*^2}+\tau e^{-\frac{1}{2}\epsilon_{11}}(1+\tau^{*^2})\big)\big).
\end{align}
\end{subequations}
We depict these two solutions in Fig. 5 and Fig. 6, respectively.


\begin{center}
\begin{picture}(120,100)
\put(-80,-23){\resizebox{!}{4.0cm}{\includegraphics{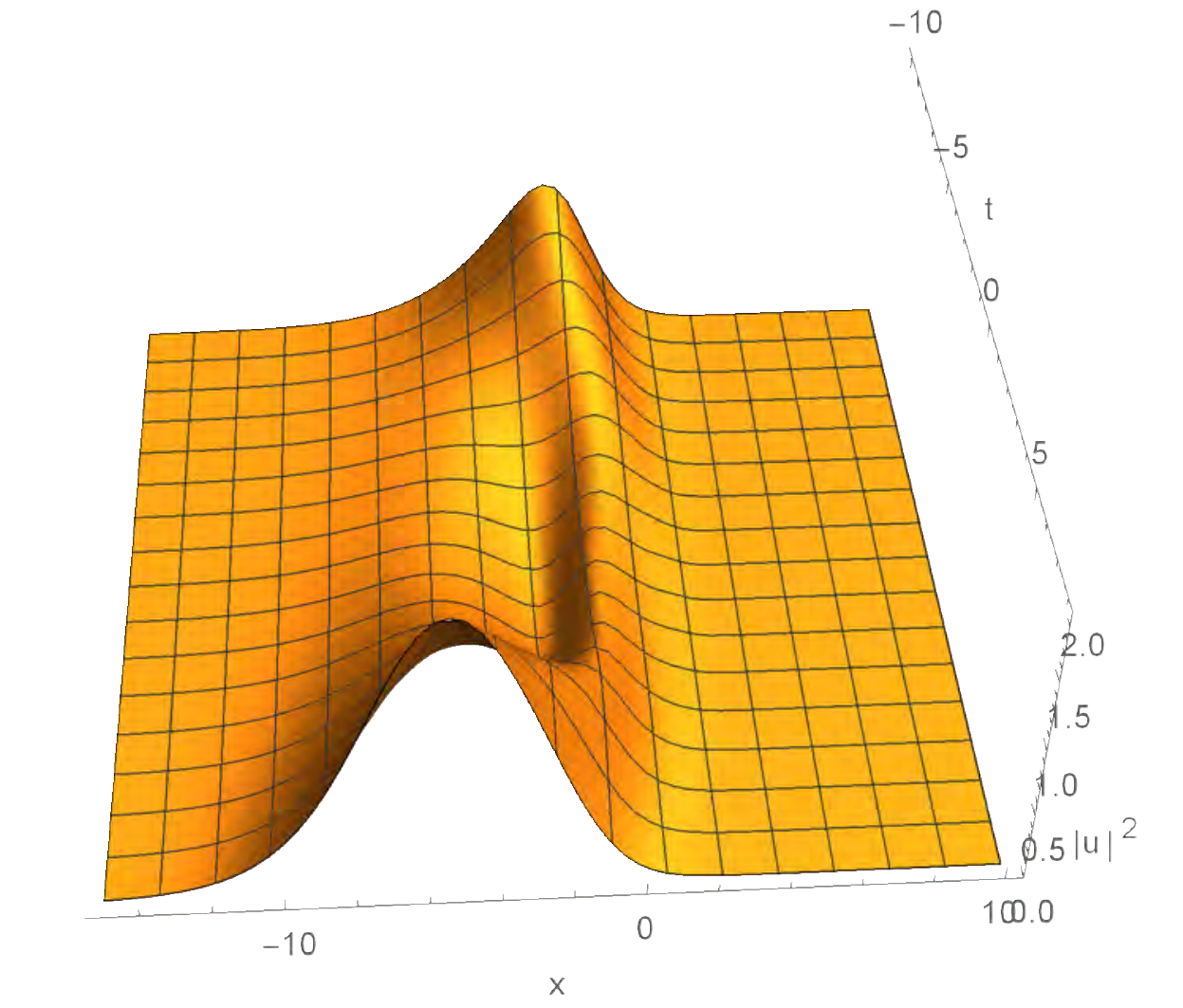}}}
\put(100,-23){\resizebox{!}{3.5cm}{\includegraphics{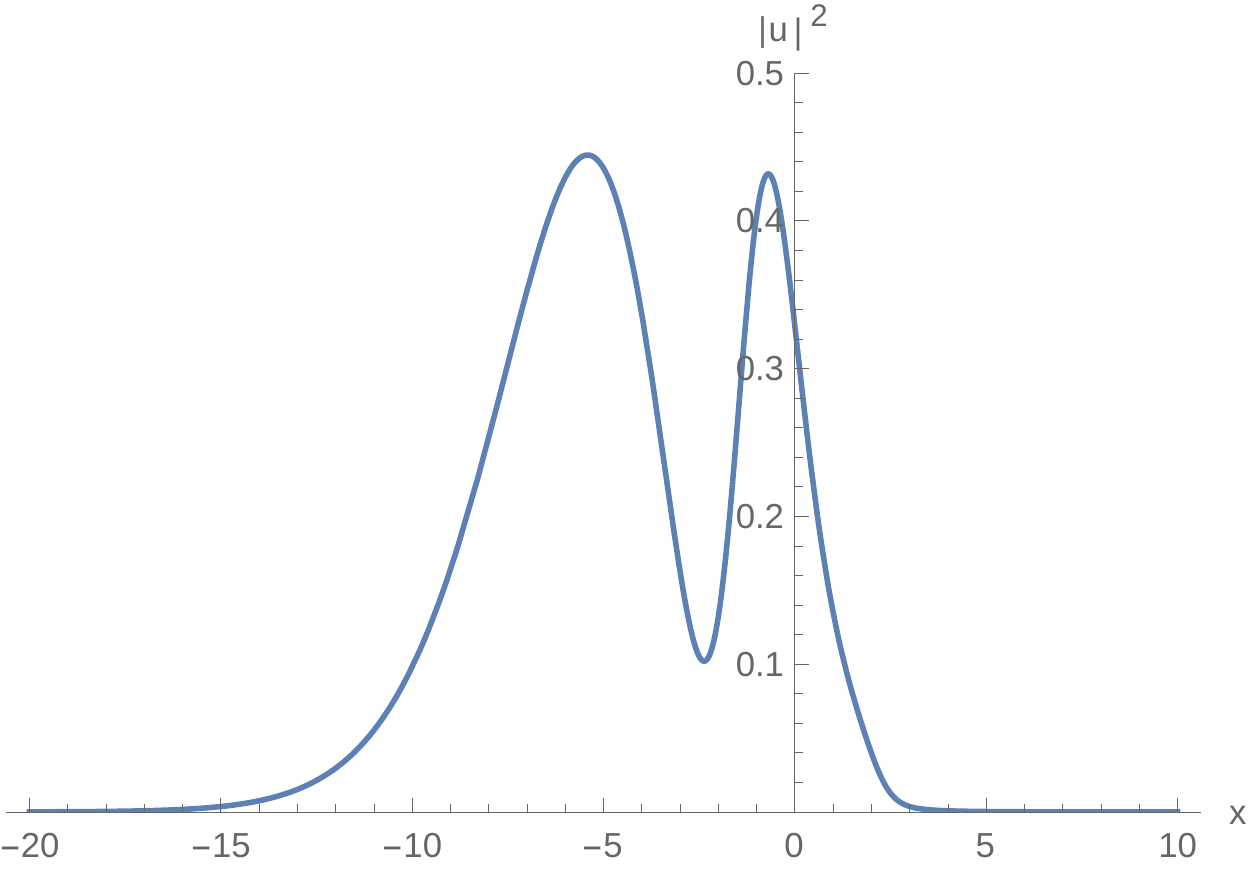}}}
\end{picture}
\end{center}
\vskip 20pt
\begin{center}
\begin{minipage}{16cm}{\footnotesize\quad\qquad\qquad\qquad\qquad\qquad\qquad\qquad (a)\quad \qquad \qquad\qquad\qquad \qquad \qquad\qquad\qquad\qquad \qquad (b)\\
{\bf Fig. 5} (a) shape and motion of Jordan block solution $|u_{13}|^2_{(\varepsilon=1,\sigma=1)}$ given by \eqref{u12-cmKdV} for $k_1=0.28-0.28i$
and $c_1=d_1=1$. (b) 2D-plot of (a) at $t=2$. }
\end{minipage}
\end{center}



\begin{center}
\begin{picture}(120,100)
\put(-80,-23){\resizebox{!}{4.0cm}{\includegraphics{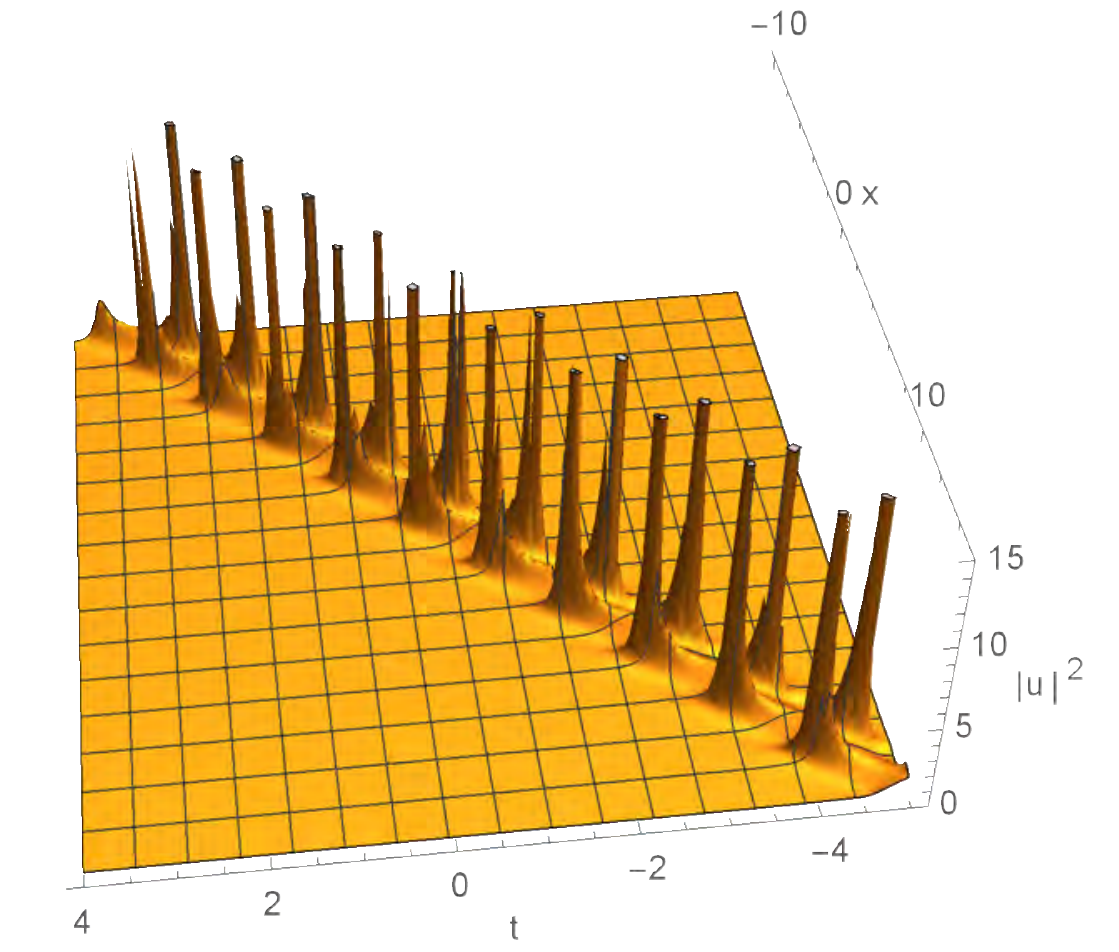}}}
\put(100,-23){\resizebox{!}{3.5cm}{\includegraphics{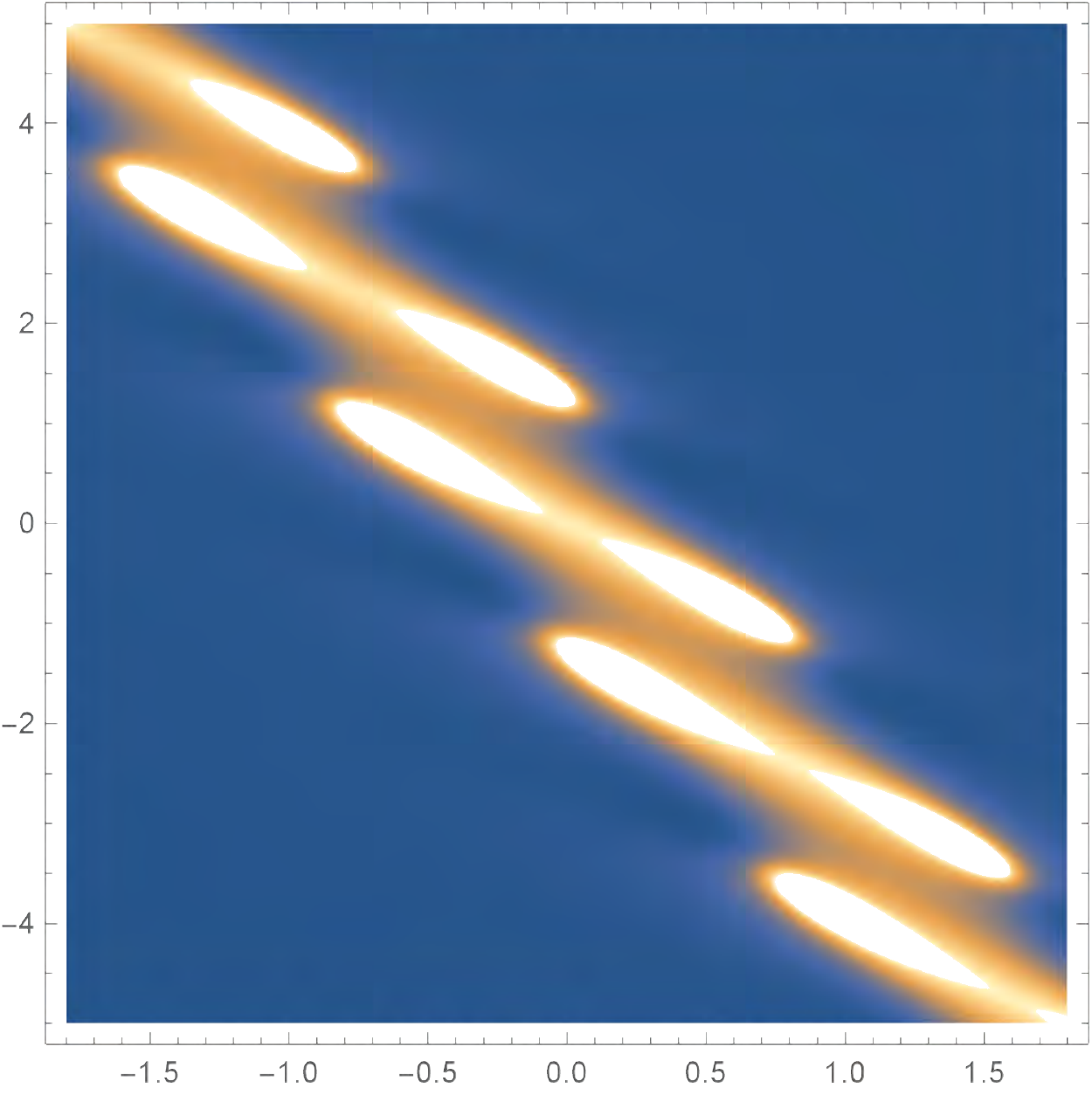}}}
\end{picture}
\end{center}
\vskip 20pt
\begin{center}
\begin{minipage}{16cm}{\footnotesize
\quad\qquad\qquad\quad\qquad\qquad\qquad\qquad \qquad(a)\quad \qquad \qquad\qquad\qquad \qquad \qquad\qquad\qquad(b)\qquad \qquad}
\end{minipage}
\end{center}
\vskip 10pt
\begin{center}
\begin{picture}(120,80)
\put(-150,-23){\resizebox{!}{3.5cm}{\includegraphics{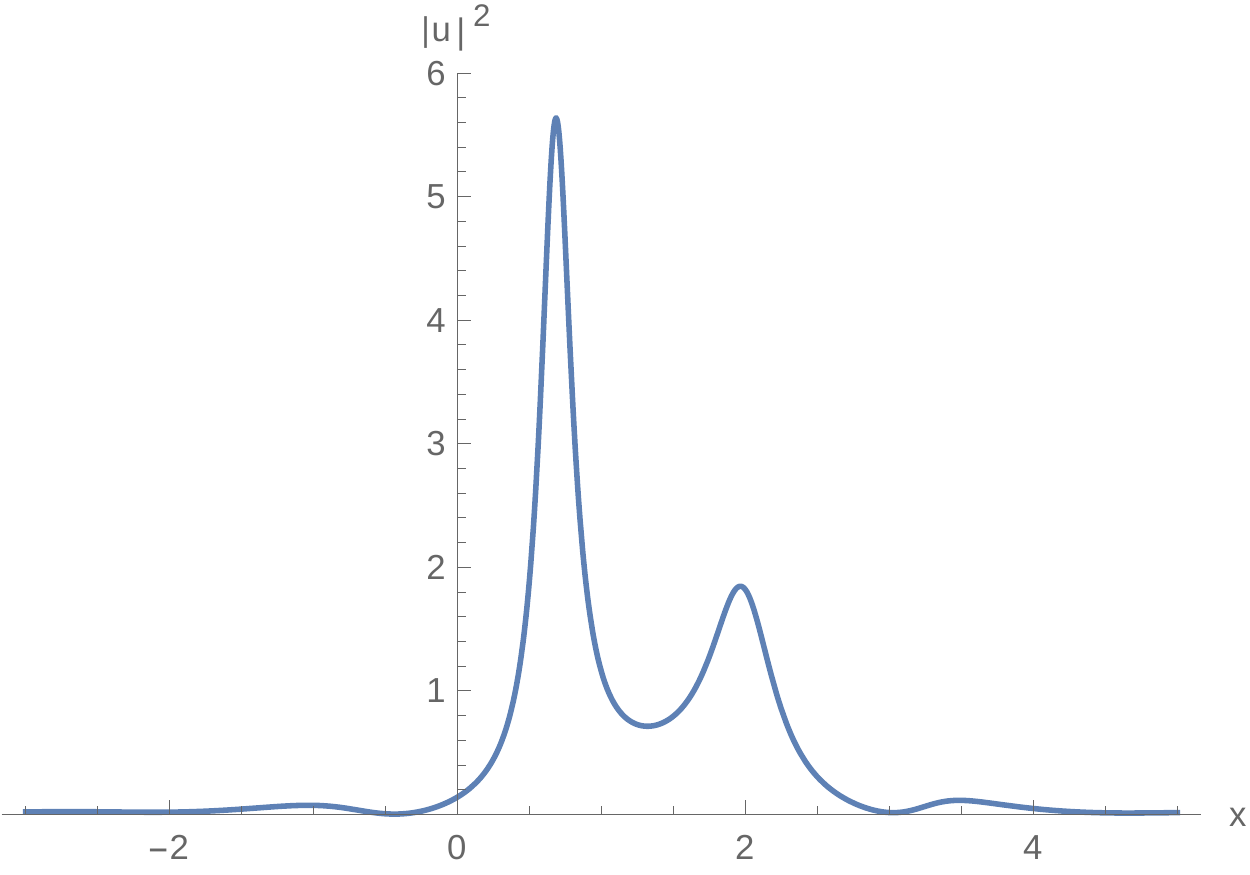}}}
\put(10,-23){\resizebox{!}{3.5cm}{\includegraphics{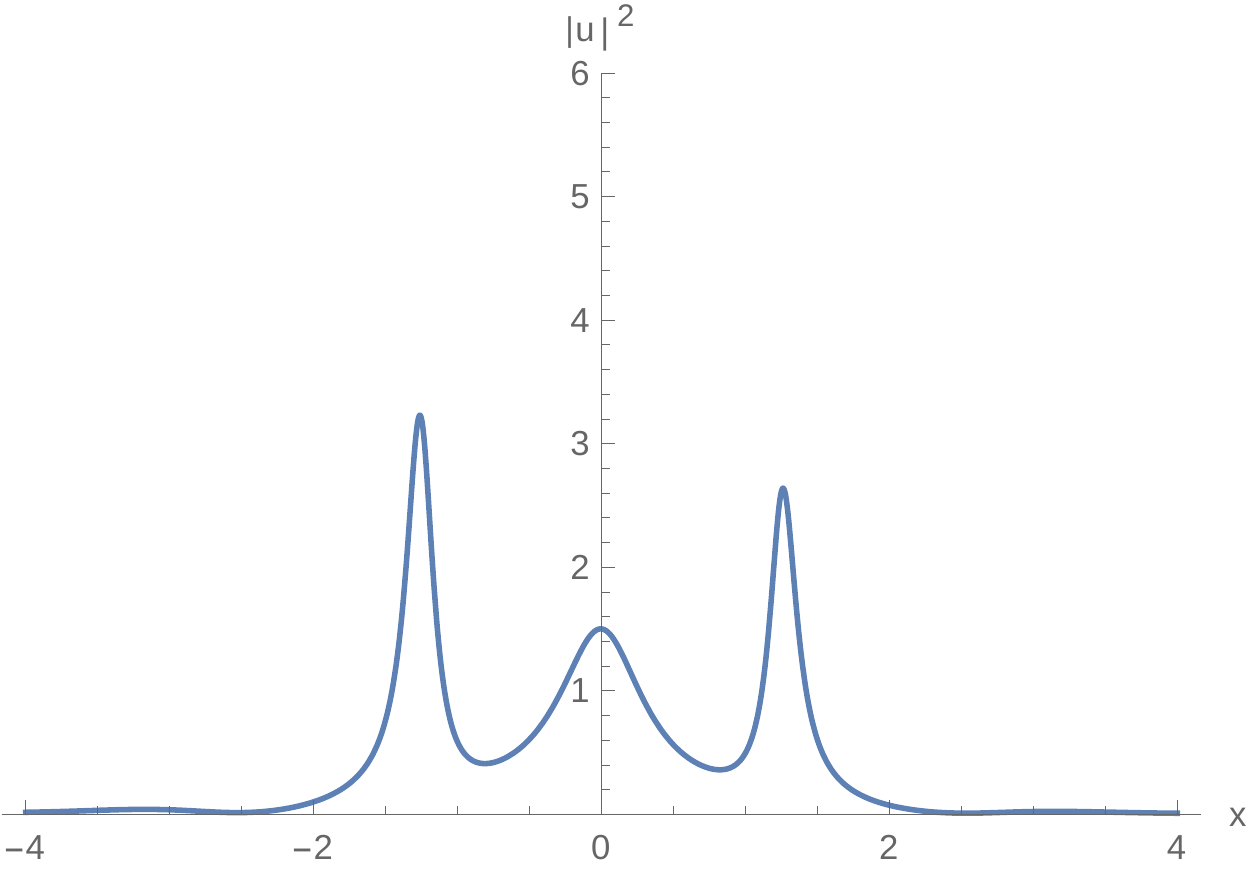}}}
\put(150,-23){\resizebox{!}{3.5cm}{\includegraphics{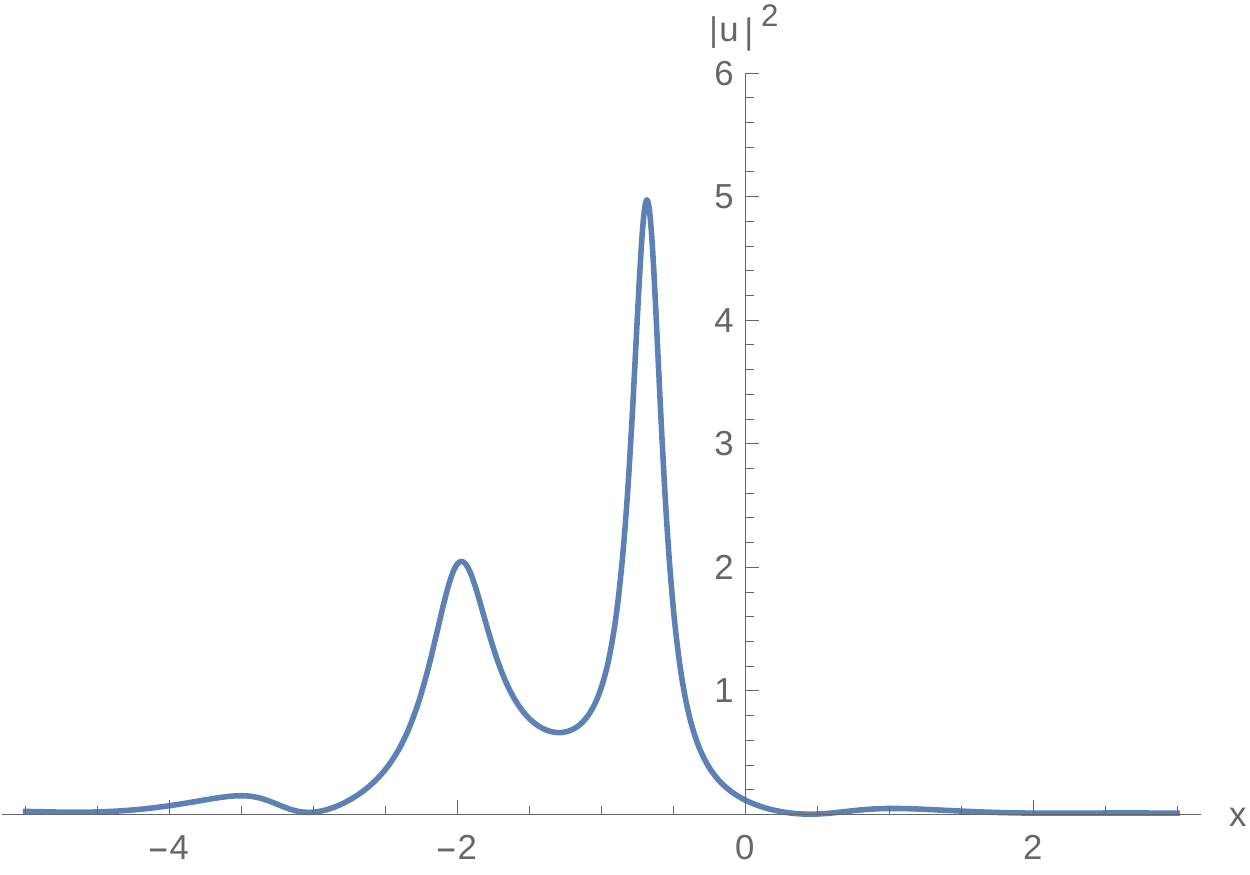}}}
\end{picture}
\end{center}
\vskip 20pt
\begin{center}
\begin{minipage}{16cm}{\footnotesize\quad\qquad\qquad\qquad(c)\qquad\qquad\qquad\qquad \qquad\quad \qquad \qquad \qquad\quad
(d)\qquad \qquad \qquad\qquad\qquad\qquad\qquad\quad (e)\\
{\bf Fig. 6} (a) shape and motion of Jordan block solution $|u_{13}|^2_{(\varepsilon=1,\sigma=-1)}$
given by \eqref{u34-cmKdV} for $k_1=0.02+i$ and $c_1=d_1=1$. (b) a contour plot of (a) with range $x\in [-5, 5]$ and $t\in [-1.8,1.8]$.
(c) 2D-plot of (a) at $t=-0.5$. (d) 2D-plot of (a) at $t=0$. (e) 2D-plot of (a) at $t=0.5$.}
\end{minipage}
\end{center}

\subsection{Local and nonlocal csG equation}

Let us now consider the Cauchy matrix solutions to local and nonlocal csG equation \eqref{n-sG}. We set $n=-1$ in \eqref{SE-sys}.
Analogous to the previous analysis, the Cauchy matrix solution for equation \eqref{n-sG} can be summerized by the following theorem.
Here we skip the proof since it is similar to the one for Theorem \ref{so-nlcmKdV}.
\begin{Thm}
\label{so-nlcsG}
The function
\begin{equation}\label{nlcsG-u-solu}
u(x,t)=\bs^{\st}_2(\bI-\bM_2\bM_1)^{-1}\br_2,
\end{equation}
solves the local and nonlocal csG equation \eqref{n-sG},
where the entities satisfy \eqref{DES-C}  $(n=-1)$ and the constraints
\eqref{nlcmKdV-M1M12}, in which $\bT\in \mathbb{C}_{N\times N}$ is a constant matrix satisfying
\eqref{nlcmKdV-at-eq}.
\end{Thm}

\subsubsection{Exact solutions}

From the Theorem \ref{so-nlcsG}, one knows that solution to the local and nonlocal csG equation \eqref{n-sG} is still
given by \eqref{cmKdV-u-solu-1}, where the components $\br_2(x,t)$ and $\bs_2(x,t)$ are of form
\begin{align}
\label{rsj-sG}
\br_{2}=\mbox{exp}(-\Og_2x+\Og^{-1}_2t)\bC^+_2, \quad \bs_{2}=\mbox{exp}(-\Og^{\st}_2x+(\Og_2^{\st})^{-1}t)\bD^+_2,
\end{align}
and $\bM_2$ and $\bT$ are determined by \eqref{DES-M12-C-mKdV}. Table 1 and Table
2 also, respectively, provide the soliton solutions and the Jordan block solutions for the
local and nonlocal csG equation \eqref{n-sG}, where in Table 1
\begin{subequations}
\begin{align}
& \br_{2}=\mbox{exp}(-\Og_{\ty{D},2} x+\Og^{-1}_{\ty{D},2}t)\bC^+_{\ty{D},2},\quad \bs_{2}=\mbox{exp}(-\Og_{\ty{D},2}x+\Og^{-1}_{\ty{D},2}t)\bD^+_{\ty{D},2},\\
& \bM^{(1)}_{2}=-i\bM^{(2)}_{2}=\mbox{exp}(-\Og_{\ty{D},2}x+\Og^{-1}_{\ty{D},2}t)\bC^-_{\ty{D},2}\cdot
\bG_{\ty{D}}^{(12)}\cdot\bD^{-^*}_{\ty{D},2}\mbox{exp}(-\Og^{*}_{\ty{D},2}x+\Og^{*^{-1}}_{\ty{D},2}t),\\
& \bM^{(3)}_{2}=-i\bM^{(4)}_{2}=\mbox{exp}(-\Og_{\ty{D},2}x+\Og^{-1}_{\ty{D},2}t)\bC^-_{\ty{D},2}\cdot
\bG_{\ty{D}}^{(34)}\cdot\bD^{-^*}_{\ty{D},2}\mbox{exp}(\Og^*_{\ty{D},2}x-\Og^{*^{-1}}_{\ty{D},2}t),
\end{align}
\end{subequations}
respectively, in Table 2
\begin{subequations}
\begin{align}
& \br_{2}=\mbox{exp}(-\Og_{\ty{J},2} x+\Og_{\ty{J},2}^{-1}t)\bC_{\ty{J},2}^+, \quad \bs_{2}=\mbox{exp}(-\Og_{\ty{J},2}^{\st} x+(\Og_{\ty{J},2}^{\st})^{-1}t)\bD_{\ty{J},2}^+,\\
& \wh{\bM}^{(1)}_2=-i\wh{\bM}^{(2)}_2=\mbox{exp}(-\Og_{\ty{J},2} x+\Og_{\ty{J},2}^{-1} t)\bC_{\ty{J},2}^-
\cdot\wh{\bG}_{\ty{J}}^{\st}\cdot\bD_{\ty{J},2}^{-^*}\mbox{exp}(-\Og_{\ty{J},2}^{*^{\st}} x+(\Og_{\ty{J},2}^{*^{\st}})^{-1} t) , \\
& \wt{\bM}^{(1)}_2=-i\wt{\bM}^{(2)}_2=\mbox{exp}(-\Og_{\ty{J},2} x+\Og_{\ty{J},2}^{-1} t)\bC_{\ty{J},2}^-
\cdot\wt{\bG}_{\ty{J}}^{\st}\cdot\bD_{\ty{J},2}^{-^*}\mbox{exp}(\Og_{\ty{J},2}^{*^{\st}}x-(\Og_{\ty{J},2}^{*^{\st}})^{-1} t).
\end{align}
\end{subequations}

We write down the 1-soliton solutions for equation \eqref{n-sG} as follows
\begin{subequations}
\label{sG-u-soli}
\begin{align}
& \label{sG-soli-11} u_{21,(\varepsilon=1,\sigma=1)}=\frac{\vartheta_1}{
 e^{-\zeta_1}+|\vartheta_1|^{2}e^{\zeta_1^*+\theta_{11}}}, \\
& \label{sG-soli-12} u_{21,(\varepsilon=i,\sigma=1)}=\frac{\vartheta_1}{
 e^{-\zeta_1}-|\vartheta_1|^{2}e^{\zeta_1^*+\theta_{11}}}, \\
& \label{sG-soli-21} u_{21,(\varepsilon=1,\sigma=-1)}=\frac{\vartheta_1}{
e^{-\zeta_1}+|\vartheta_1|^{2}e^{-\zeta_1^*+\epsilon_{11}}}, \\
& \label{sG-soli-22} u_{21,(\varepsilon=i,\sigma=-1)}=\frac{\vartheta_1}{
e^{-\zeta_1}-|\vartheta_1|^{2}e^{-\zeta_1^*+\epsilon_{11}}},
\end{align}
\end{subequations}
where $\zeta_1=2(k_1^{-1}t-k_1x)$.

\subsubsection{Dynamics}

We now identify the dynamics of 1-soliton solutions \eqref{sG-u-soli}.
As two examples, we consider \eqref{sG-soli-11} and \eqref{sG-soli-21}.
Substituting \eqref{k1-com} into \eqref{sG-soli-11} leads to
\begin{align}
\label{var1si1-sG}
|u_{21}|^2_{(\varepsilon=1,\sigma=1)}=\alpha^2\mbox{sech}^2\left[2\alpha\left(x-\frac{t}{\alpha^2+\beta^2}+\hbar\right)\right].
\end{align}
This wave travels with fixed amplitude $\alpha^2$ and
constant velocity $\frac{1}{\alpha^2+\beta^2}$. The top trajectory is $x=\frac{t}{\alpha^2+\beta^2}-\hbar$.
There is no stationary wave because of the positive velocity.
Fig. 7 depicts a moving soliton wave.


\begin{center}
\begin{picture}(120,100)
\put(-80,-23){\resizebox{!}{3.5cm}{\includegraphics{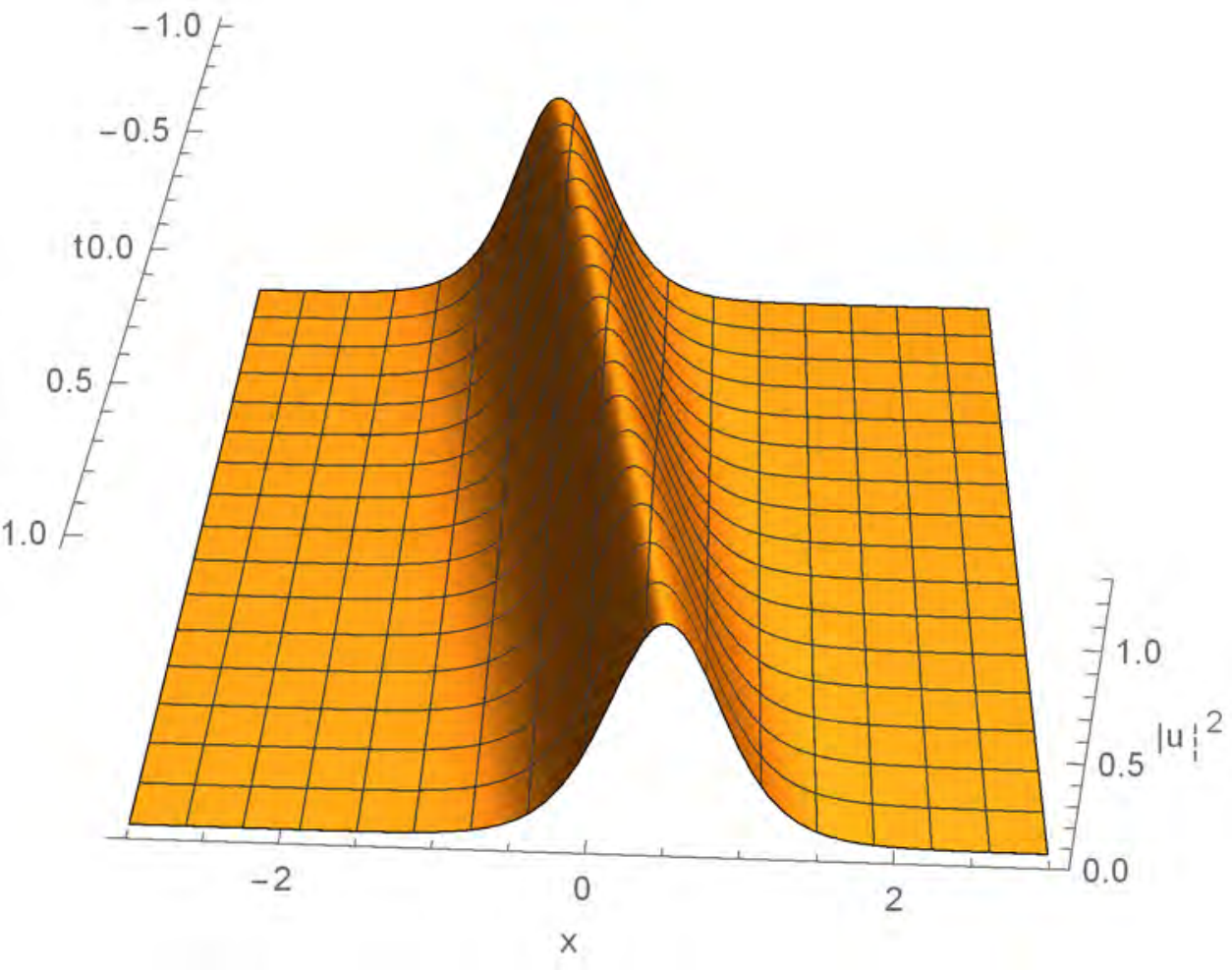}}}
\put(100,-23){\resizebox{!}{3.5cm}{\includegraphics{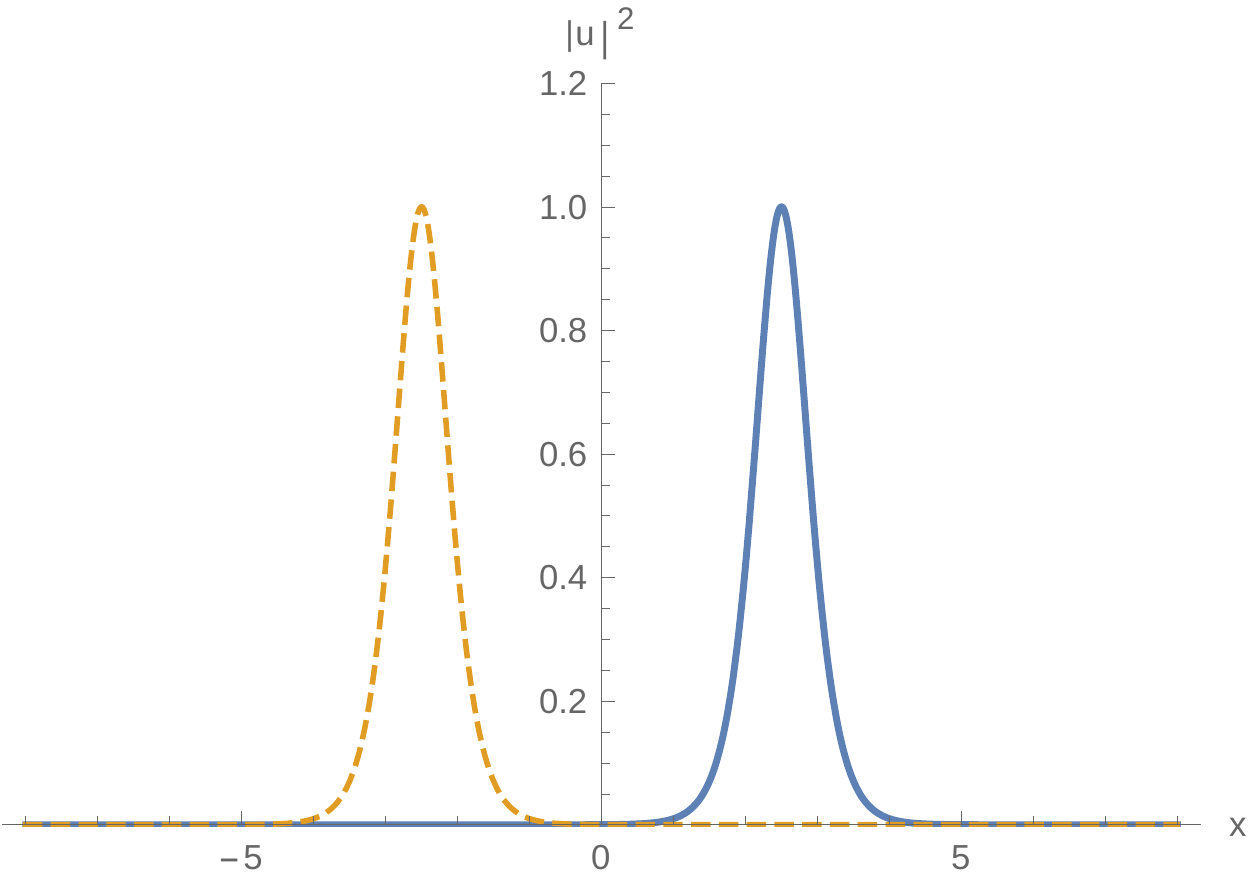}}}
\end{picture}
\end{center}
\vskip 20pt
\begin{center}
\begin{minipage}{16cm}{\footnotesize
\quad\qquad\qquad\qquad\qquad\qquad\qquad \qquad (a)\qquad\qquad\qquad\qquad\qquad\qquad\qquad\qquad \qquad\quad (b) \\
{\bf Fig. 7} (a) shape and motion of 1-soliton solution \eqref{var1si1-sG} for $k_1=c_1=d_1=1+i$.
(b) waves in solid line and dotted line stand for plot (a) at $t=5$ and $t=-5$, respectively.}
\end{minipage}
\end{center}


For the solution \eqref{sG-soli-21}, we have
\begin{align}
\label{var1si-1-sG}
|u_{21}|^2_{(\varepsilon=1,\sigma=-1)}=\frac{16|\beta^2\vartheta_1|^2e^{4\alpha\left(\frac{t}{\alpha^2+\beta^2}-x\right)}}
{|\vartheta_1|^4+16\beta^4-8|\beta\vartheta_1|^2\cos[4\beta(x+\frac{t}{\alpha^2+\beta^2})]}.
\end{align}
In quite a similar fashion, there is a quasi-periodic phenomenon for \eqref{var1si-1-sG} because of the involvement of cosine function.
When $|\vartheta_1|^2=4\beta^2$, solution \eqref{var1si-1-sG} has singularities along with straight lines
\begin{align}
x(t)=-\frac{t}{\alpha^2+\beta^2}+\frac{\kappa\pi}{2\beta}, \quad \kappa\in \mathbb{Z}.
\end{align}
While when $|\vartheta_1|^2\neq 4\beta^2$, \eqref{var1si-1-sG} is nonsingular and reaches its extrema along with straight lines
\begin{align}
\label{xt-ext-sG}
x(t)=-\frac{t}{\alpha^2+\beta^2}+\frac{1}{4\beta}\left(\gamma+2\kappa\pi-
\arcsin\frac{\alpha(|\vartheta_1|^4+16\beta^4)}{8|\beta \vartheta_1|^2\sqrt{\alpha^2+\beta^2}}\right), \quad \kappa\in \mathbb{Z},
\end{align}
where $\gamma$ is same as \eqref{xt-ext}.
The velocity is $-\frac{1}{\alpha^2+\beta^2}$. For a given $t$, $|u_{21}|^2_{(\varepsilon=1,\sigma=-1)}$
tends to zero for $(\alpha>0, x\rightarrow \infty)$ and  $(\alpha<0, x\rightarrow -\infty)$, respectively.
We depict this solution in Fig. 8.


\begin{center}
\begin{picture}(120,100)
\put(-150,-23){\resizebox{!}{3.5cm}{\includegraphics{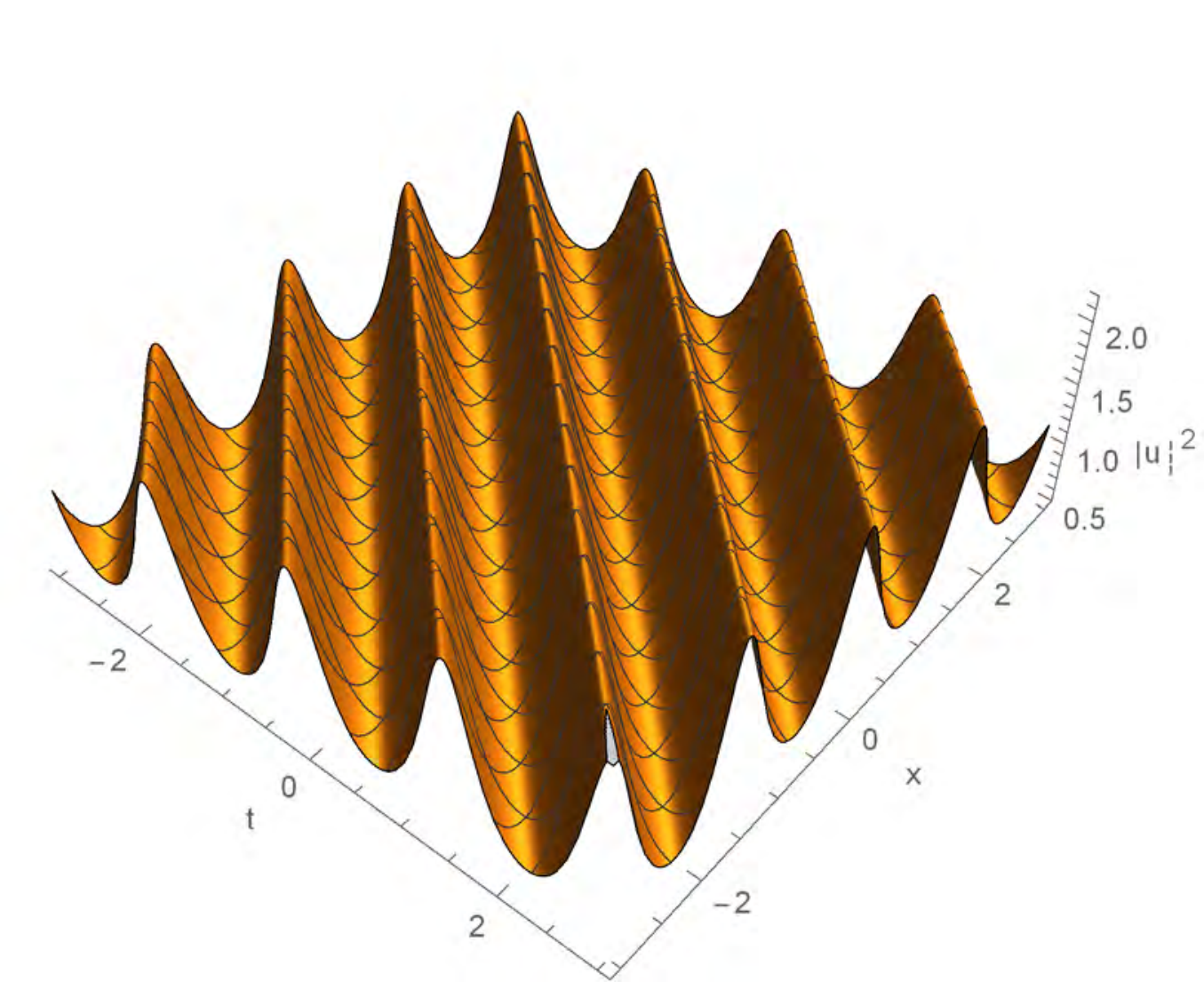}}}
\put(10,-23){\resizebox{!}{3.5cm}{\includegraphics{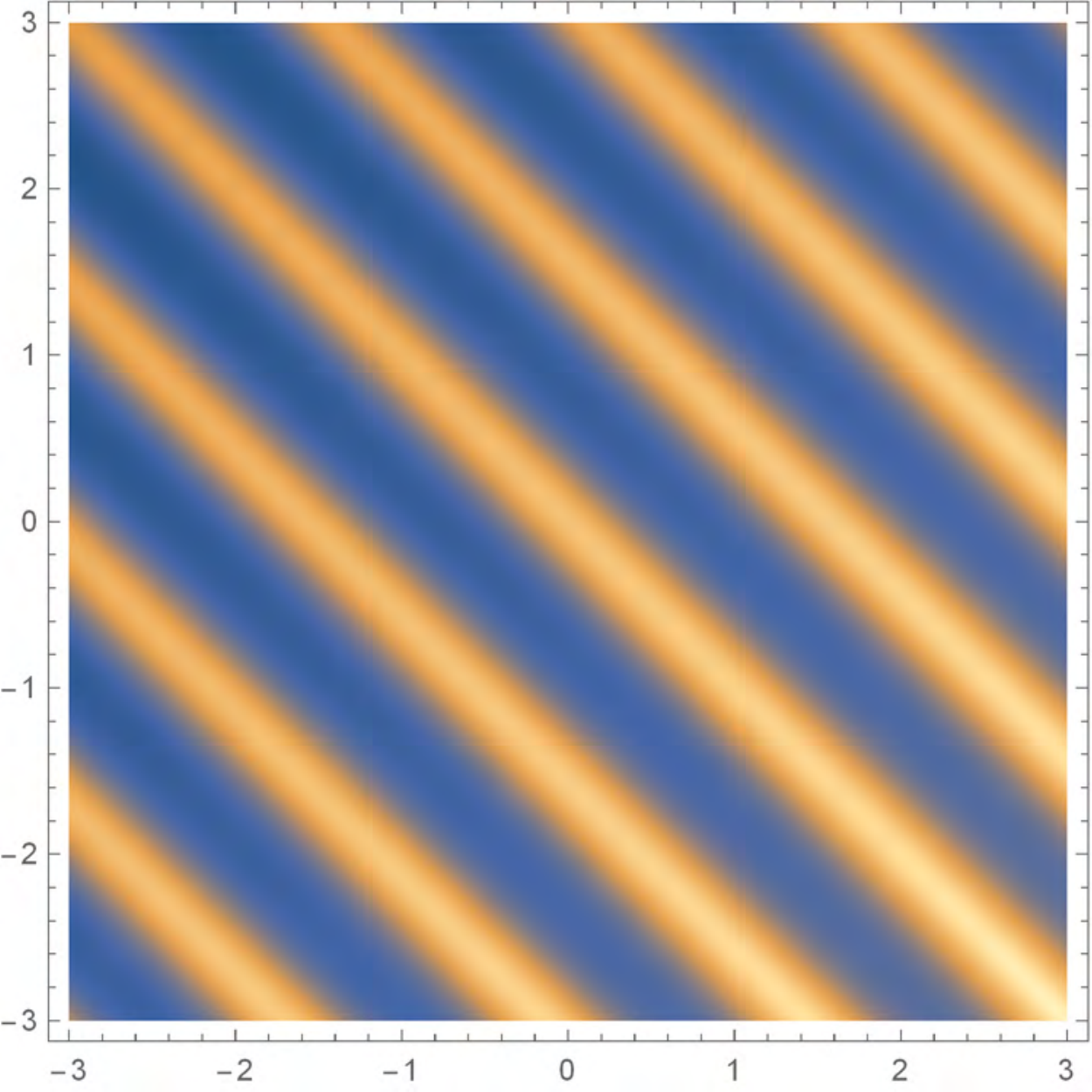}}}
\put(150,-23){\resizebox{!}{3.5cm}{\includegraphics{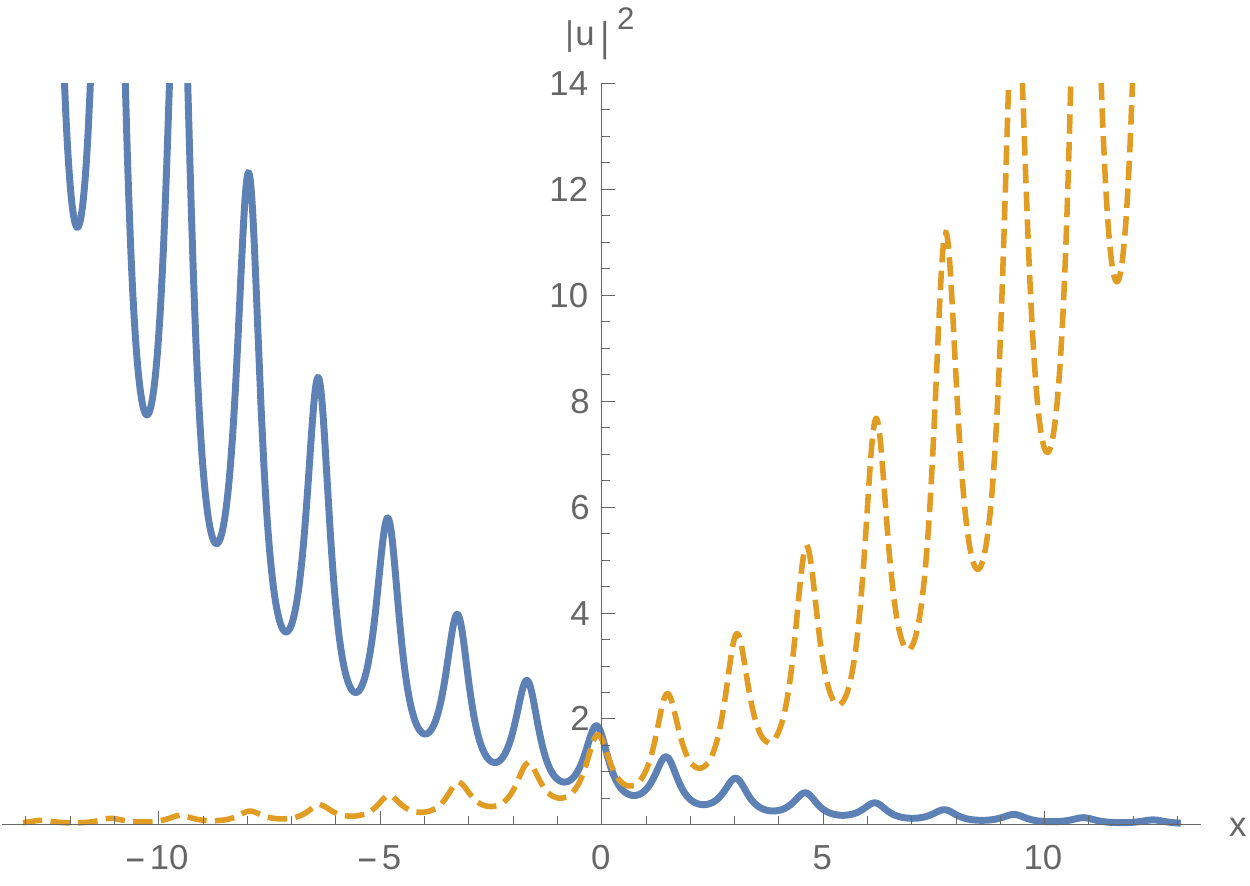}}}
\end{picture}
\end{center}
\vskip 20pt
\begin{center}
\begin{minipage}{16cm}{\footnotesize
\quad\qquad\qquad\qquad(a)\qquad\qquad\qquad\qquad \qquad\quad \qquad \qquad \quad (b) \qquad\qquad \qquad \qquad\qquad\qquad\qquad \quad (c)\\
{\bf Fig. 8} shape and motion of 1-soliton solution \eqref{var1si-1-sG} for $c_1=d_1=1$: (a) 3D-plot for
$k_1=0.01+i$. (b) a contour plot of (a) with range
$x\in [-3, 3]$ and $t\in [-3,3]$. (c) waves in solid line and dotted line stand for curves at $t=0.1$ with $k_1=0.06+i$ and $k_1=-0.06+i$, respectively.}
\end{minipage}
\end{center}


The 2-soliton solutions of \eqref{n-sG} in the case of $(\varepsilon=1,\sigma=1)$ and $(\varepsilon=1,\sigma=-1)$ are
$u_{22,(\varepsilon=1,\sigma=1)}=\frac{u'_{\ty{D},2}}{u'_{\ty{D},1}}$ with
\begin{subequations}
\label{q12-csG-2ss}
\begin{align}
& u'_{\ty{D},1}=1+\sum_{i=1}^2\sum_{j=1}^2\vartheta^*_i\vartheta_je^{\zeta^*_{i}+\zeta_{j}+\theta_{ij}}+
|\vartheta_1\vartheta_2|^2|k_1-k_2|^4e^{\zeta_{1}+\zeta_{1}^*+\zeta_{2}+\zeta_{2}^*+\theta_{11}+\theta_{12}+\theta_{21}+\theta_{22}}, \\
& u'_{\ty{D},2}=\sum_{i=1}^2\vartheta_ie^{\zeta_{i}}+\vartheta_1\vartheta_2(k_1-k_2)^2e^{\zeta_{1}+\zeta_{2}}
\left(\sum_{i=1}^2\vartheta_i^*e^{\zeta_{i}^*+\theta_{i1}+\theta_{i2}}\right),
\end{align}
\end{subequations}
and $u_{22,(\varepsilon=1,\sigma=-1)}=\frac{u'_{\ty{D},4}}{u'_{\ty{D},3}}$ with
\begin{subequations}
\label{q34-csG-2ss}
\begin{align}
& u'_{\ty{D},3}=1+\sum_{i=1}^2\sum_{j=1}^2\vartheta^*_i\vartheta_je^{-\zeta^*_{i}+\zeta_{j}+\epsilon_{ij}}+
|\vartheta_1\vartheta_2|^2|k_1-k_2|^4e^{\zeta_{1}-\zeta_{1}^*+\zeta_{2}-\zeta_{2}^*+\epsilon_{11}+\epsilon_{12}+\epsilon_{21}+\epsilon_{22}}, \\
& u'_{\ty{D},4}=\sum_{i=1}^2\vartheta_ie^{\zeta_{i}}+\vartheta_1\vartheta_2(k_1-k_2)^2e^{\zeta_{1}
+\zeta_{2}}\left(\sum_{i=1}^2\vartheta_i^*e^{-\zeta_{i}^*+\epsilon_{i1}+\epsilon_{i2}}\right).
\end{align}
\end{subequations}
We depict $|u_{22}|^2_{(\varepsilon=1,\sigma=1)}$ and $|u_{22}|^2_{(\varepsilon=1,\sigma=-1)}$ in Fig. 9 and Fig. 10, respectively.


\begin{center}
\begin{picture}(120,100)
\put(-80,-23){\resizebox{!}{4.0cm}{\includegraphics{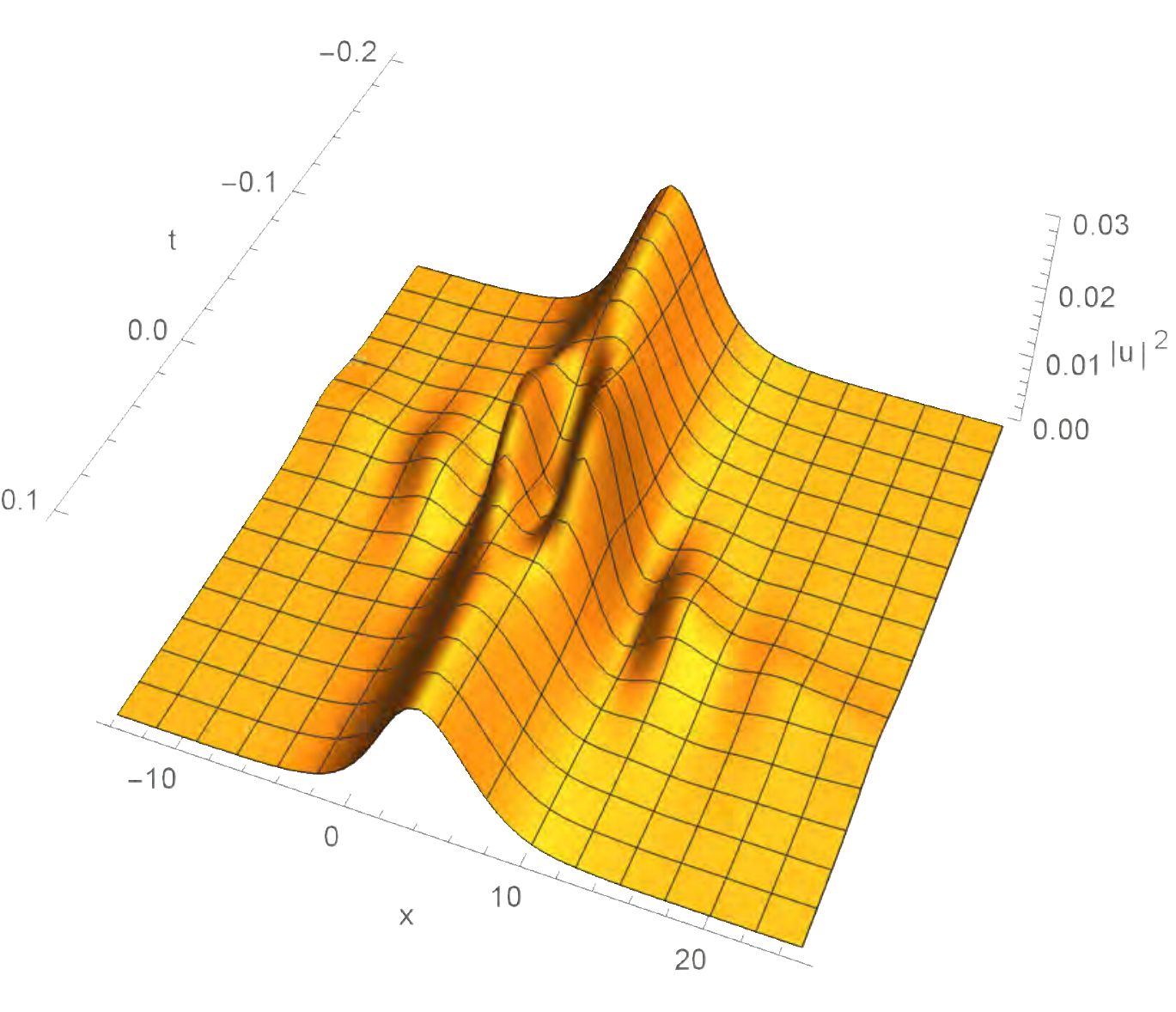}}}
\put(100,-23){\resizebox{!}{3.5cm}{\includegraphics{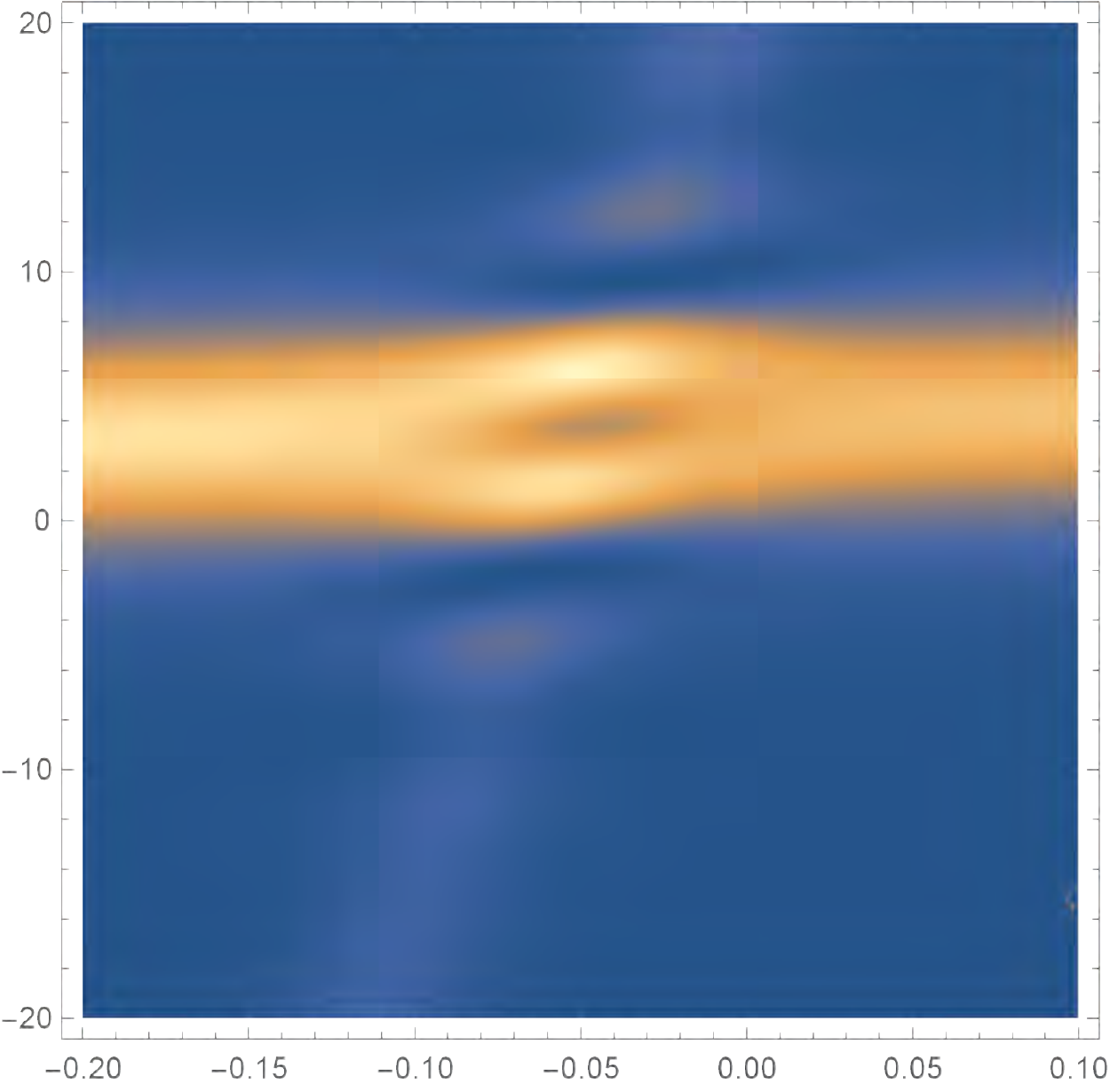}}}
\end{picture}
\end{center}
\vskip 20pt
\begin{center}
\begin{minipage}{16cm}{\footnotesize
\quad\qquad\qquad\qquad\qquad\qquad\qquad\qquad \qquad(a)\quad \qquad \qquad\qquad\qquad \qquad \qquad\qquad\quad(b)}
\end{minipage}
\end{center}
\vskip 10pt
\begin{center}
\begin{picture}(120,80)
\put(-150,-23){\resizebox{!}{3.5cm}{\includegraphics{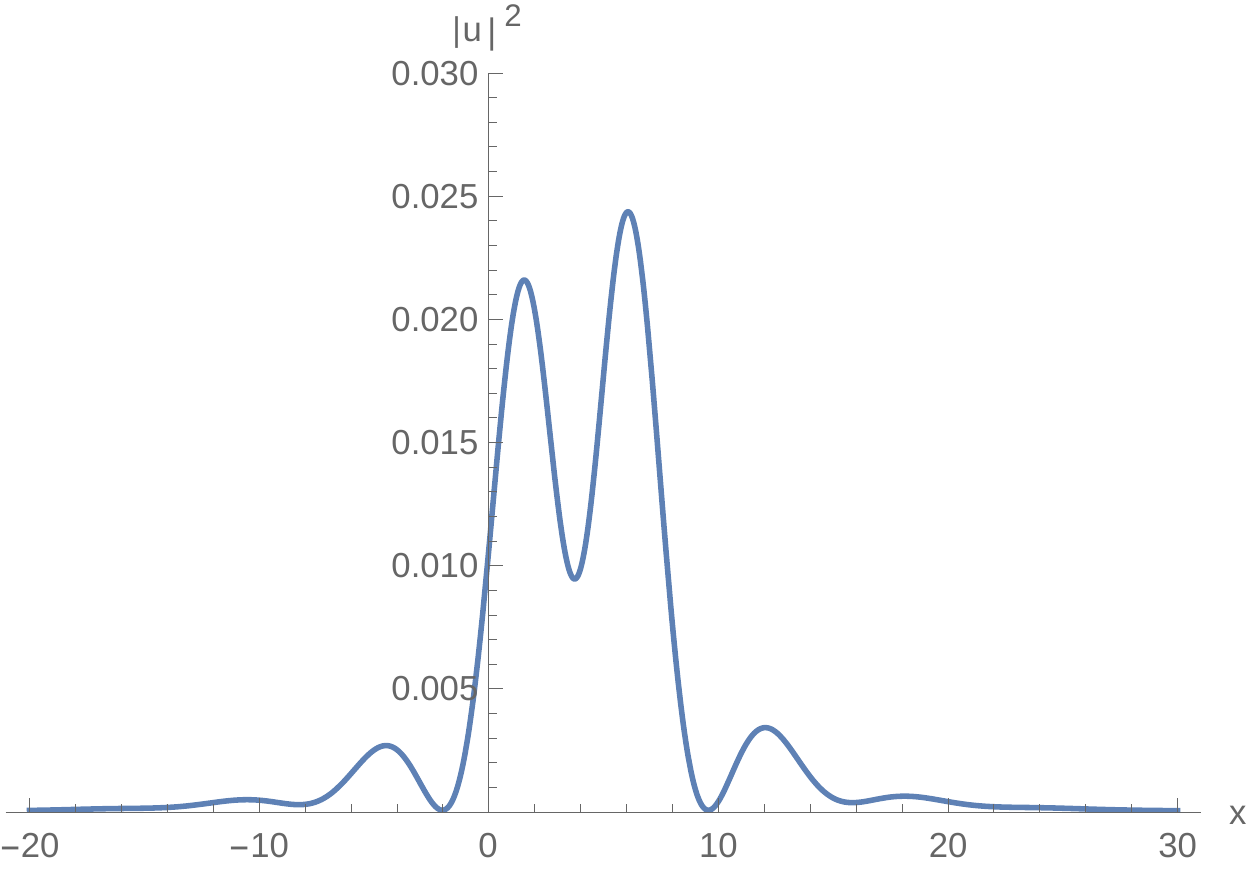}}}
\put(10,-23){\resizebox{!}{3cm}{\includegraphics{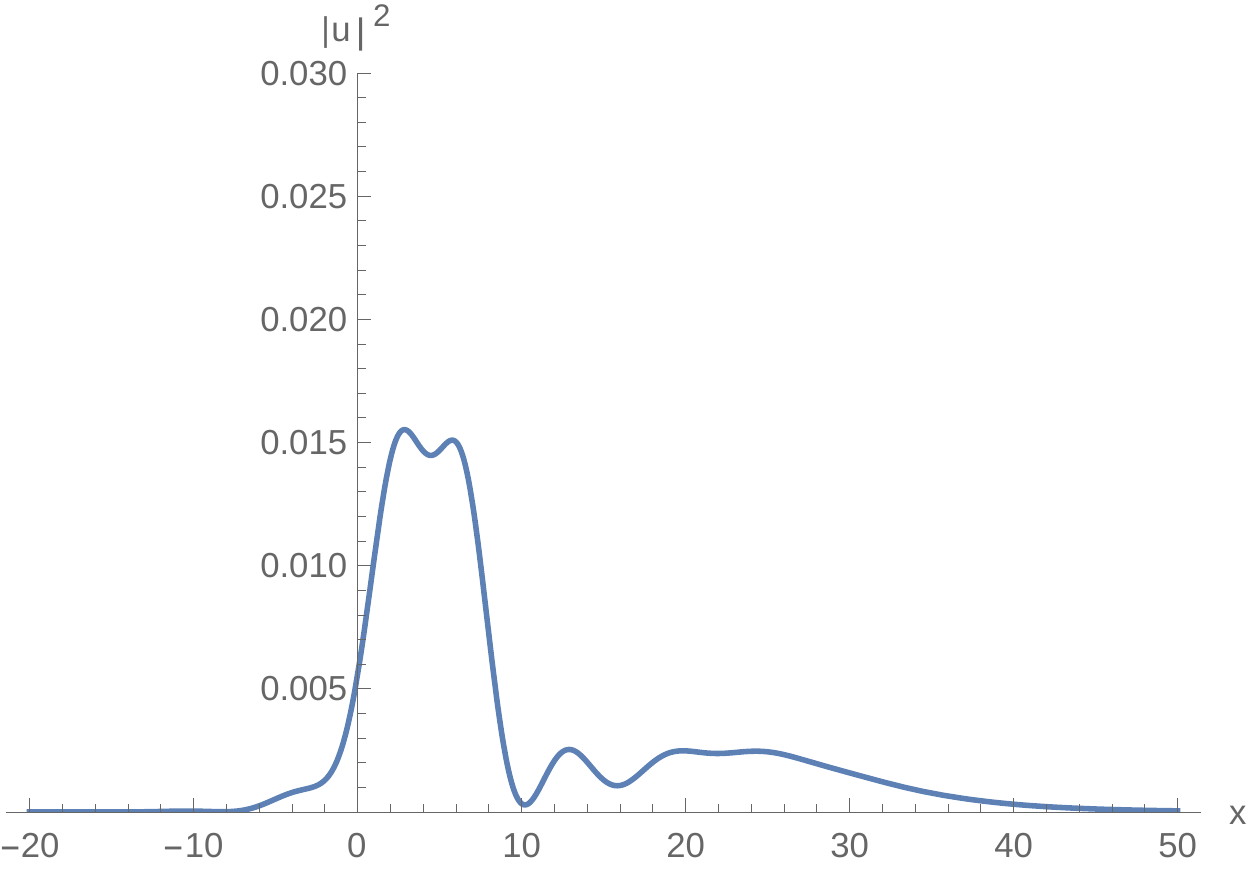}}}
\put(150,-23){\resizebox{!}{3cm}{\includegraphics{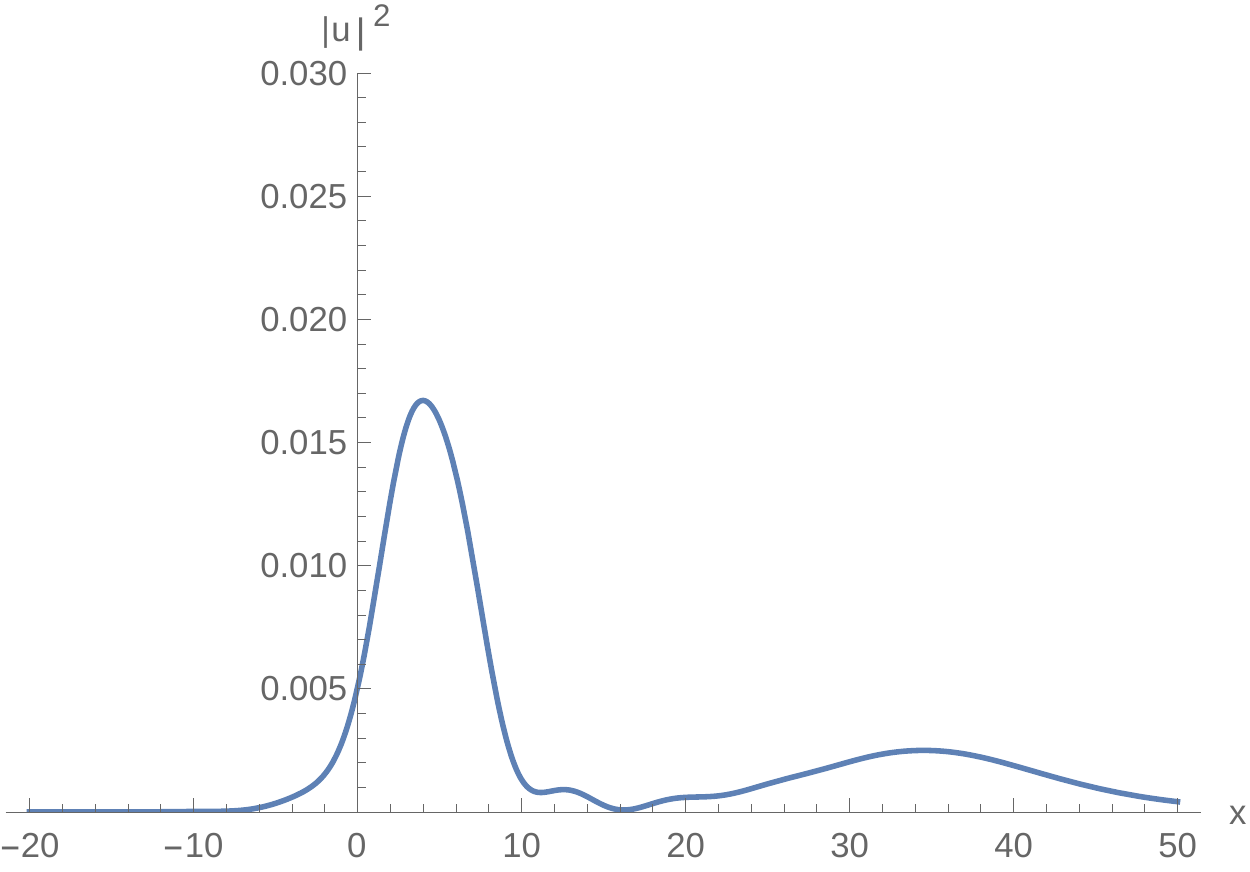}}}
\end{picture}
\end{center}
\vskip 20pt
\begin{center}
\begin{minipage}{16cm}{\footnotesize
\quad\qquad\qquad\qquad(c)\qquad\qquad\qquad\qquad \qquad\quad \qquad \qquad  (d) \qquad\qquad \qquad \qquad\qquad\qquad\qquad (e)\\
{\bf Fig. 9} (a) shape and motion of 2-soliton solutions $|u_{22}|^2_{(\varepsilon=1,\sigma=1)}$ given by \eqref{q12-csG-2ss} for $k_1=0.15+0.5i,~k_2=0.05+0.01i$
and $c_1=c_2=d_1=d_2=1$. (b) a contour plot of (a) with range $x\in [-20, 20]$ and $t\in [-0.2,0.1]$.
(c) 2D-plot of (a) at $t=-0.05$. (d) 2D-plot of (a) at $t=0$. (e) 2D-plot of (a) at $t=0.03$.}
\end{minipage}
\end{center}



\begin{center}
\begin{picture}(120,100)
\put(-150,-23){\resizebox{!}{4.5cm}{\includegraphics{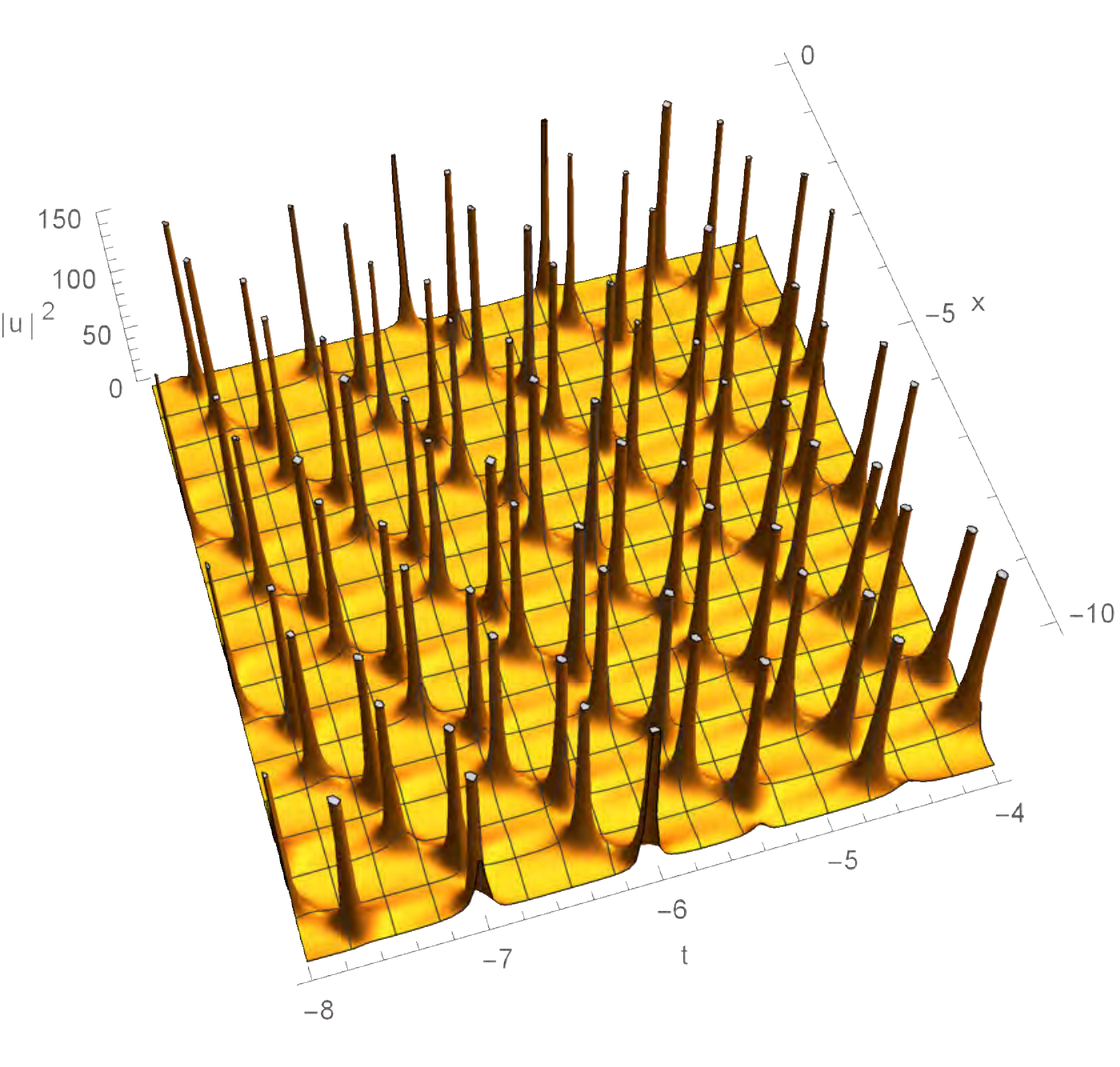}}}
\put(10,-23){\resizebox{!}{3.5cm}{\includegraphics{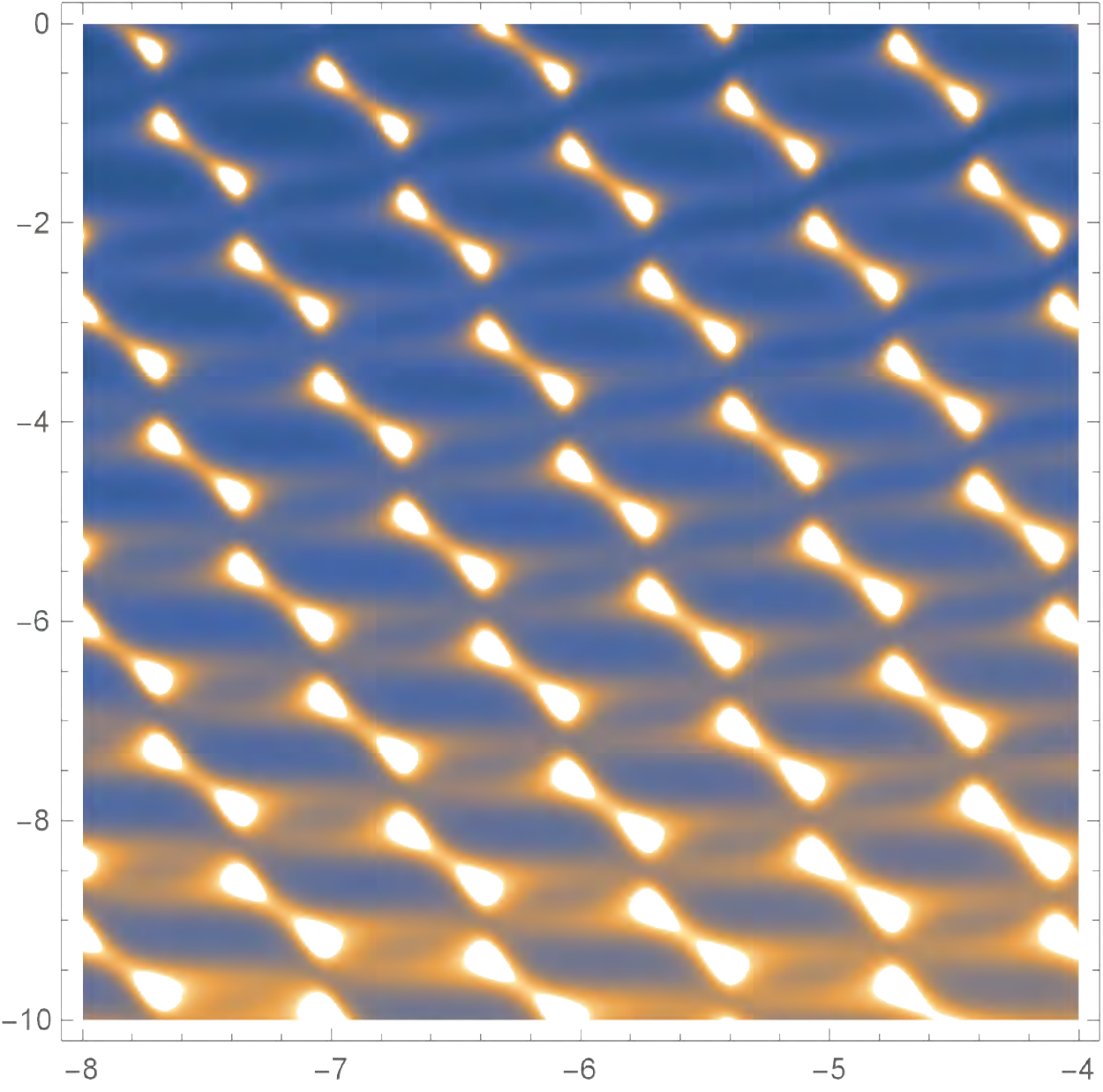}}}
\put(150,-23){\resizebox{!}{3.5cm}{\includegraphics{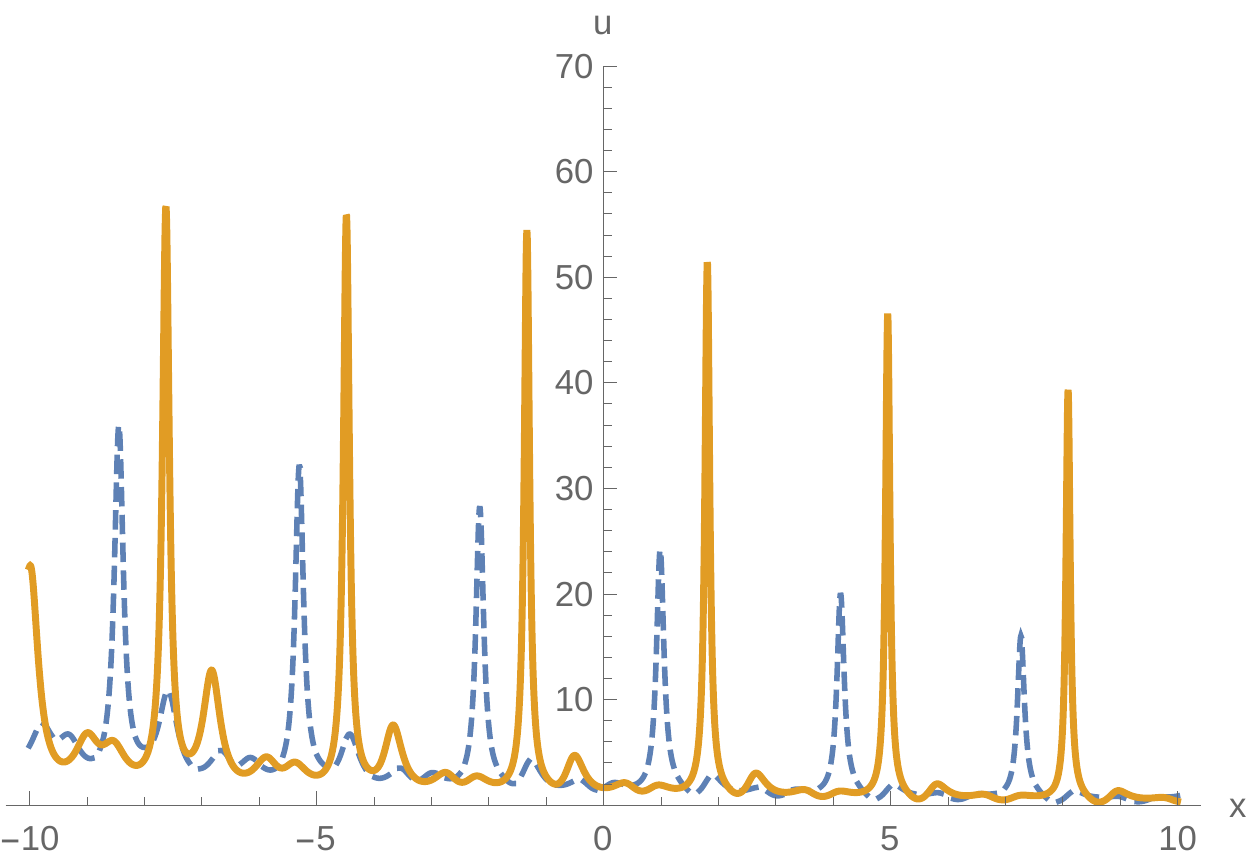}}}
\end{picture}
\end{center}
\vskip 20pt
\begin{center}
\begin{minipage}{16cm}{\footnotesize
\quad\quad\qquad\qquad\qquad\quad(a)\quad\qquad\qquad\qquad \qquad\quad \qquad \qquad (b) \quad\qquad \qquad \qquad\qquad\qquad\qquad \qquad (c)\\
{\bf Fig. 10} (a) shape and motion of 2-soliton solutions $|u_{22}|^2_{(\varepsilon=1,\sigma=-1)}$
 given by \eqref{q34-csG-2ss} for $k_1=0.03+2.5i,~k_2=0.01+0.5i,~c_1=2$ and $c_2=d_1=d_2=1$. (b) a contour plot of (a) with range
$x\in [-10,0]$ and $t\in [-8,-4]$. (c) waves in solid and dotted line stand for plot (a) at $t=-6$ and $t=-5$, respectively.}
\end{minipage}
\end{center}


For the simplest Jordan block solutions, we start with the case $(\varepsilon=1,\sigma=1)$. In this case, solution is
$u_{23,(\varepsilon=1,\sigma=1)}=\frac{u'_{\ty{J},2}}{u'_{\ty{J},1}}$, in which
\begin{subequations}
\label{u12-csG}
\begin{align}
& u'_{\ty{J},1}=(|\vartheta_1|^2e^{\zeta_1+\zeta_1^*}+e^{-2\theta_{11}})^2+|\vartheta_1|^2e^{\zeta_1+\zeta_1^*-3\theta_{11}}
\big(1+(t-x)^2-2e^{\theta_{11}}\big)^2, \\
& u'_{\ty{J},2}=\vartheta_1e^{\zeta_1-4\theta_{11}}\big(1-|\vartheta_1|^2e^{\zeta_1+\zeta_1^*+2\theta_{11}}\big(1-4e^{\theta_{11}}+(e^{-\frac{1}{2}\theta_{11}}(t-x)+2)(t-x)^2\nn\\
&\qquad\quad-2\sqrt{2}(t-x)\sinh((\theta_{11}+\ln 2)/2)\big)\big).
\end{align}
\end{subequations}
We depict $|u_{23}|^2_{(\varepsilon=1,\sigma=1)}$ in Fig. 11.

	
\begin{center}
\begin{picture}(120,100)
\put(-80,-23){\resizebox{!}{4.0cm}{\includegraphics{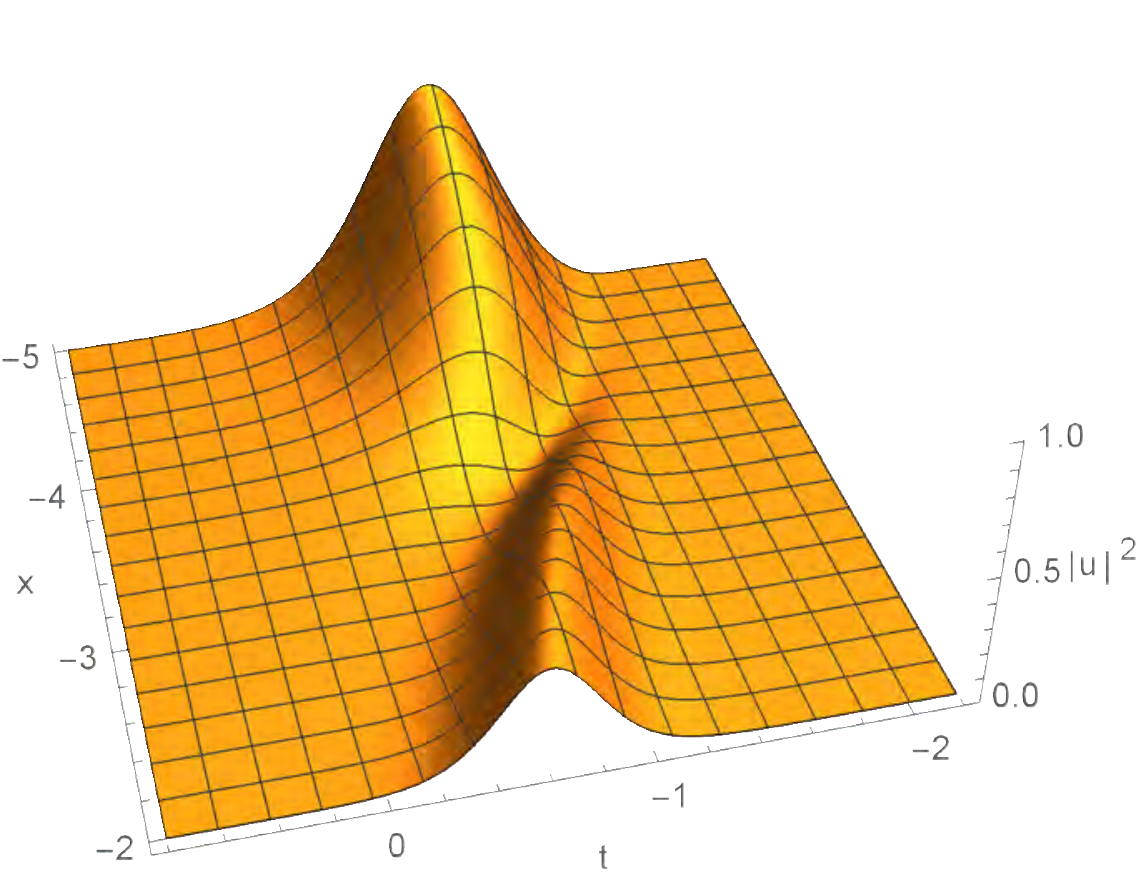}}}
\put(100,-23){\resizebox{!}{3.5cm}{\includegraphics{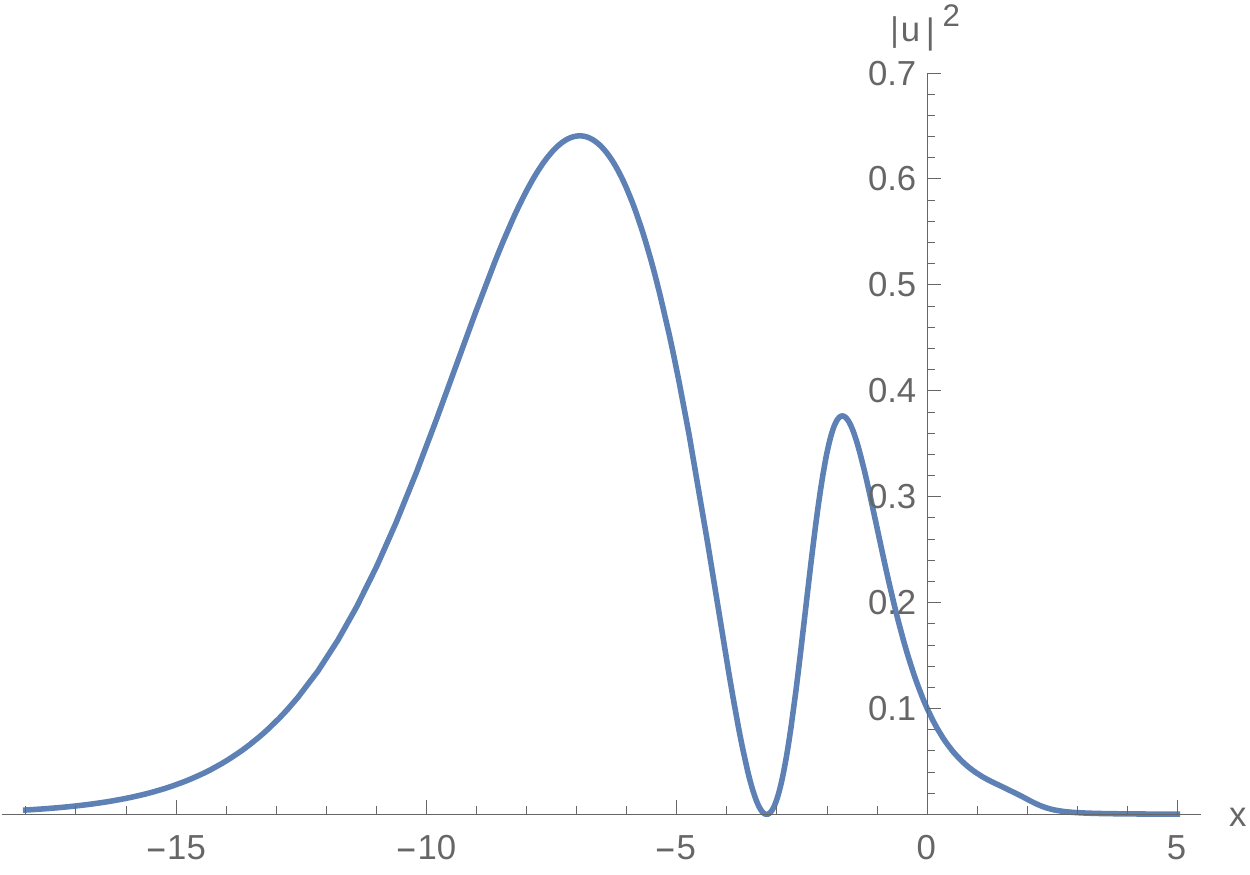}}}
\end{picture}
\end{center}
\vskip 20pt
\begin{center}
\begin{minipage}{16cm}{\footnotesize\quad\qquad\qquad\qquad\qquad\qquad\qquad\qquad \qquad(a)\quad \qquad\qquad\qquad \qquad \qquad\qquad\qquad\qquad \qquad (b)\\
{\bf Fig. 11} (a) shape and motion of Jordan block solution $|u_{23}|^2_{(\varepsilon=1,\sigma=1)}$ given by \eqref{u12-csG} for $k_1=0.25+0.3i,~
c_1=2i$ and $d_1=1$. (b) 2D-plot of (a) at $t=-0.5$.}
\end{minipage}
\end{center}


In the case of $(\varepsilon=1,\sigma=-1)$, the Jordan block
solution reads $u_{23,(\varepsilon=1,\sigma=-1)}=\frac{u'_{\ty{J},4}}{u'_{\ty{J},3}}$ with
\begin{subequations}
\label{u34-csG}
\begin{align}
& u'_{\ty{J},3}=(|\vartheta_1|^2e^{\zeta_1-\zeta_1^*}+e^{-2\epsilon_{11}})^2+|\vartheta_1|^2e^{\zeta_1-\zeta_1^*-3\epsilon_{11}}
\big(4e^{2\epsilon_{11}}+8e^{\epsilon_{11}}(t-x)^2 \nn \\
&\qquad\quad+(1+(t-x)^2)^2+4e^{\frac{1}{2}\epsilon_{11}}(1+(t-x)^2+2e^{\epsilon_{11}})(t-x)\big), \\
& u'_{\ty{J},4}=\vartheta_1e^{\zeta_1-4\epsilon_{11}}\big(1+|\vartheta_1|^2e^{\zeta_1-\zeta_1^*+2\epsilon_{11}}
\big(1+4(t-x)^2+4e^{\epsilon_{11}}+6(t-x)e^{\frac{1}{2}\epsilon_{11}}\nn \\
& \qquad\quad +e^{-\frac{1}{2}\epsilon_{11}}(t-x)(1+(t-x)^2)\big)\big).
\end{align}
\end{subequations}
We depict $|u_{23}|^2_{(\varepsilon=1,\sigma=-1)}$ in Fig. 12.
	
\begin{center}
\begin{picture}(120,100)
\put(-80,-23){\resizebox{!}{4.5cm}{\includegraphics{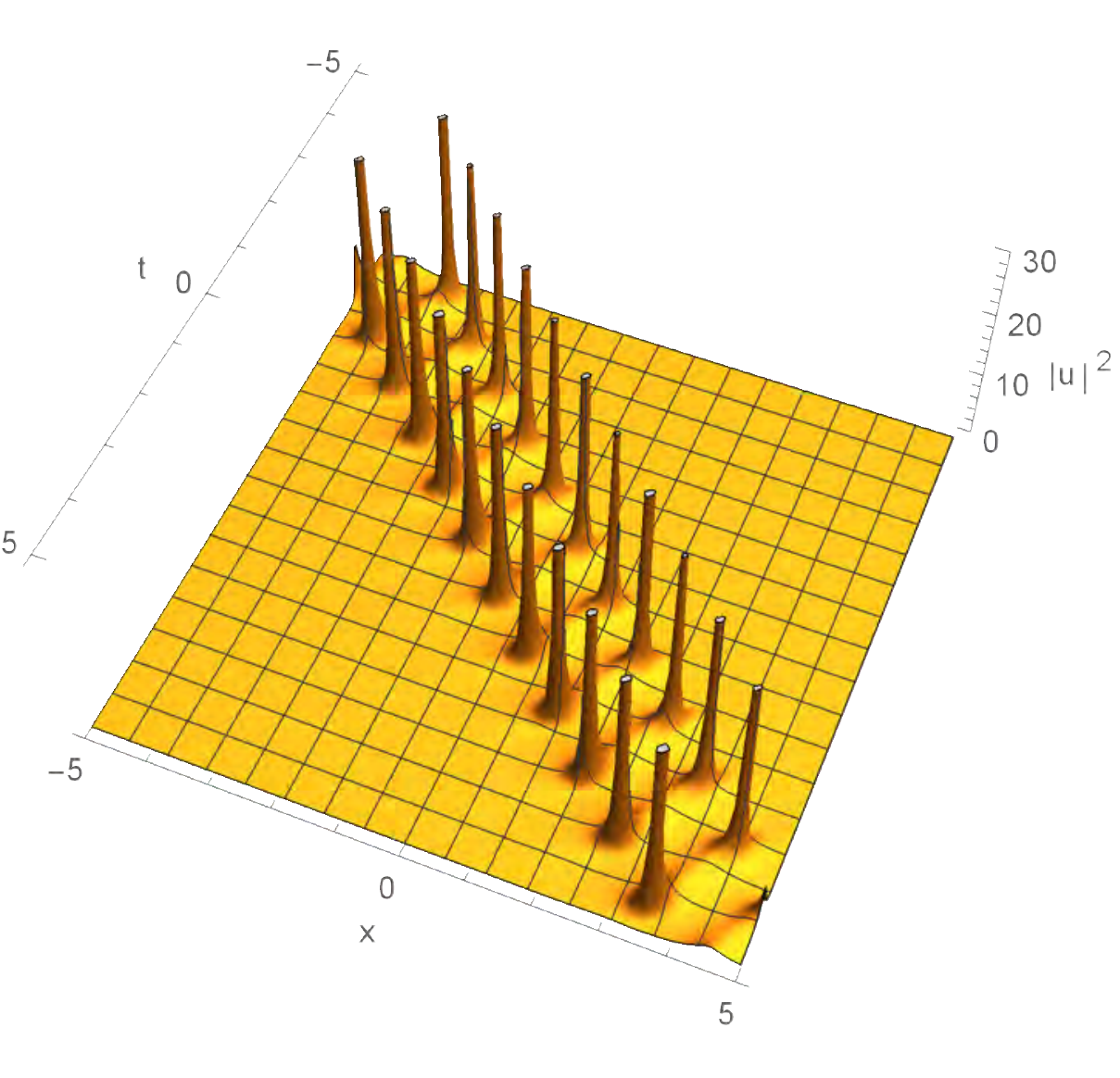}}}
\put(100,-23){\resizebox{!}{3.5cm}{\includegraphics{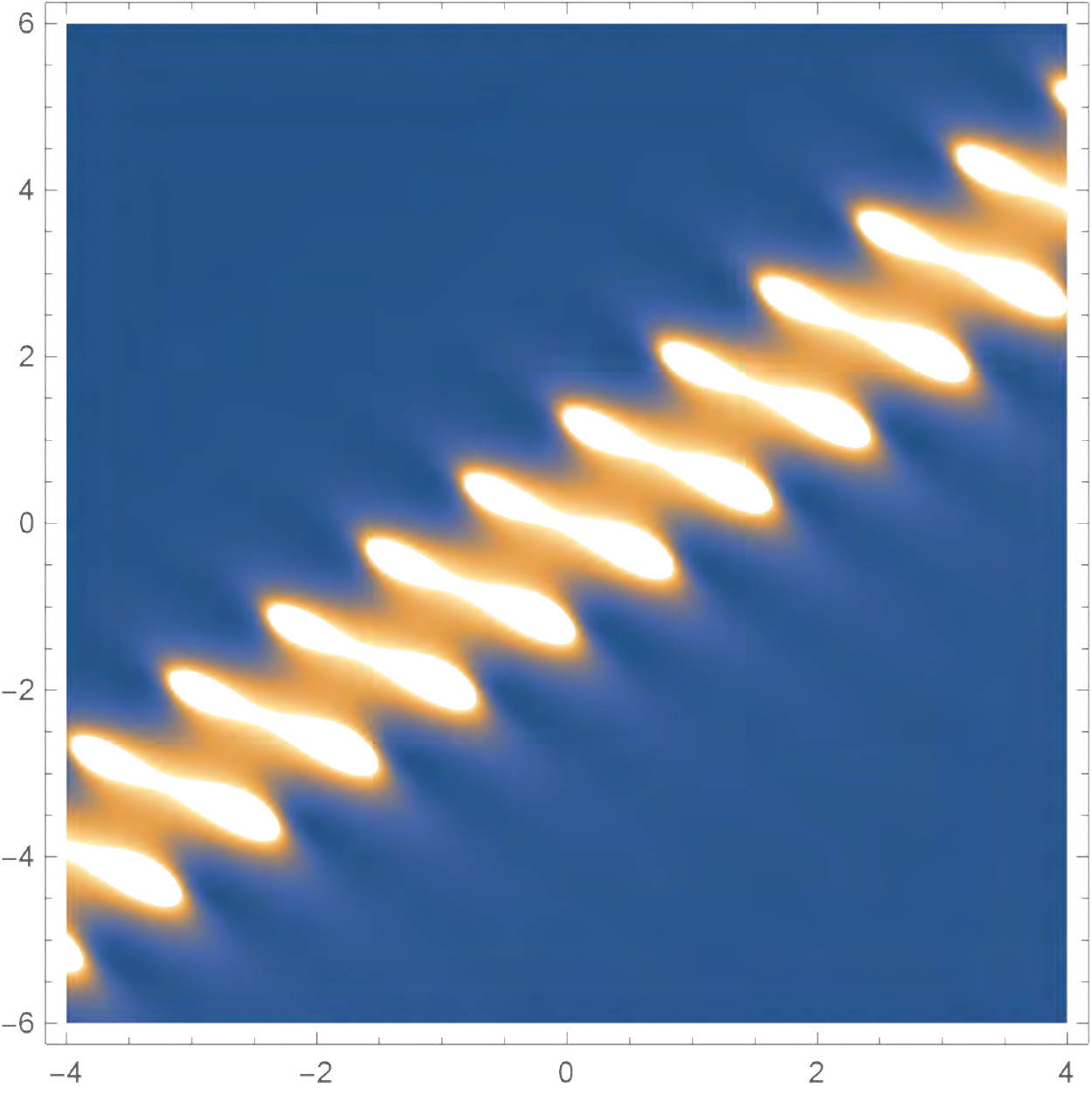}}}
\end{picture}
\end{center}
\vskip 20pt
\begin{center}
\begin{minipage}{16cm}{\footnotesize\quad\qquad\qquad\qquad\qquad\qquad\qquad \qquad(a)\quad\quad\qquad  \qquad\qquad\qquad \qquad \qquad\qquad\qquad(b)\qquad \qquad}
\end{minipage}
\end{center}
\vskip 10pt
\begin{center}
\begin{picture}(120,80)
\put(-150,-23){\resizebox{!}{3.5cm}{\includegraphics{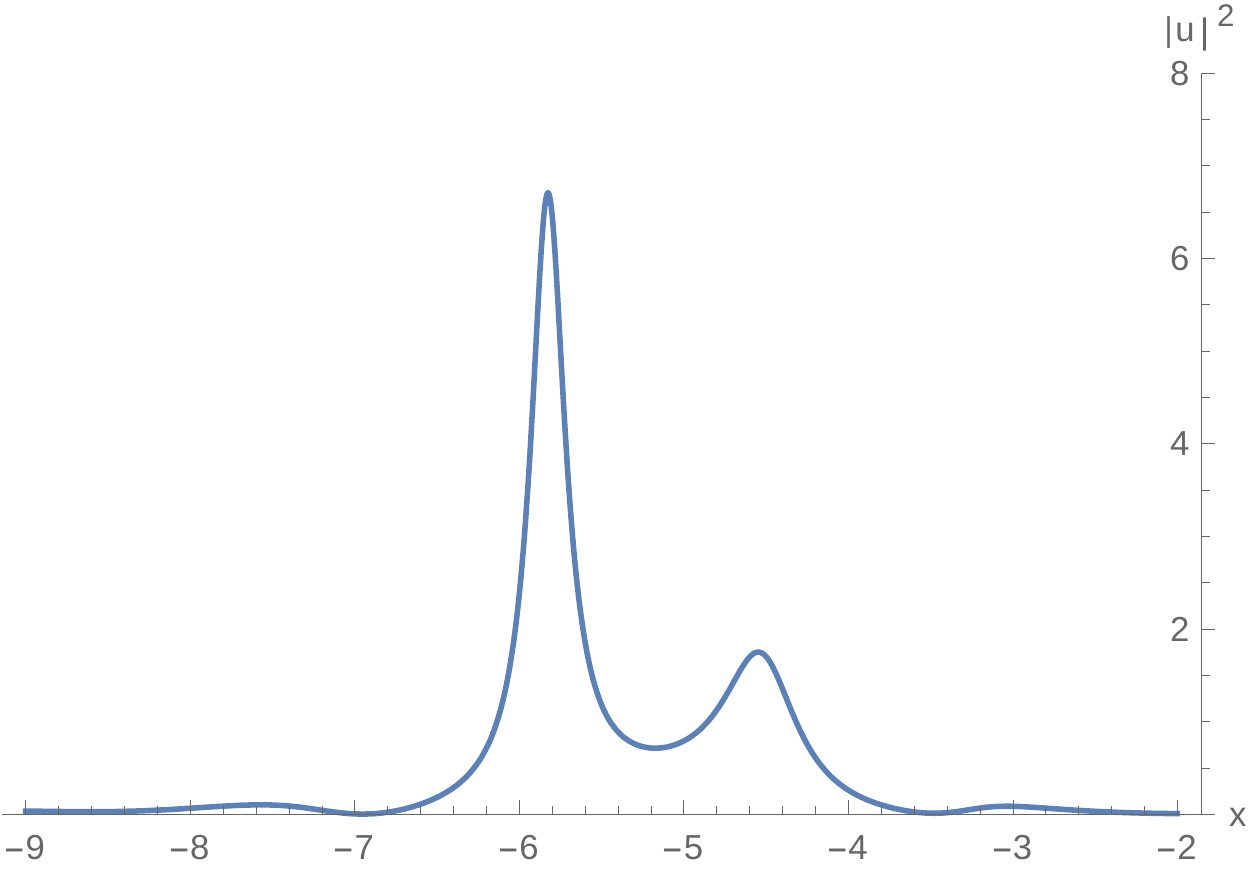}}}
\put(0,-23){\resizebox{!}{3.5cm}{\includegraphics{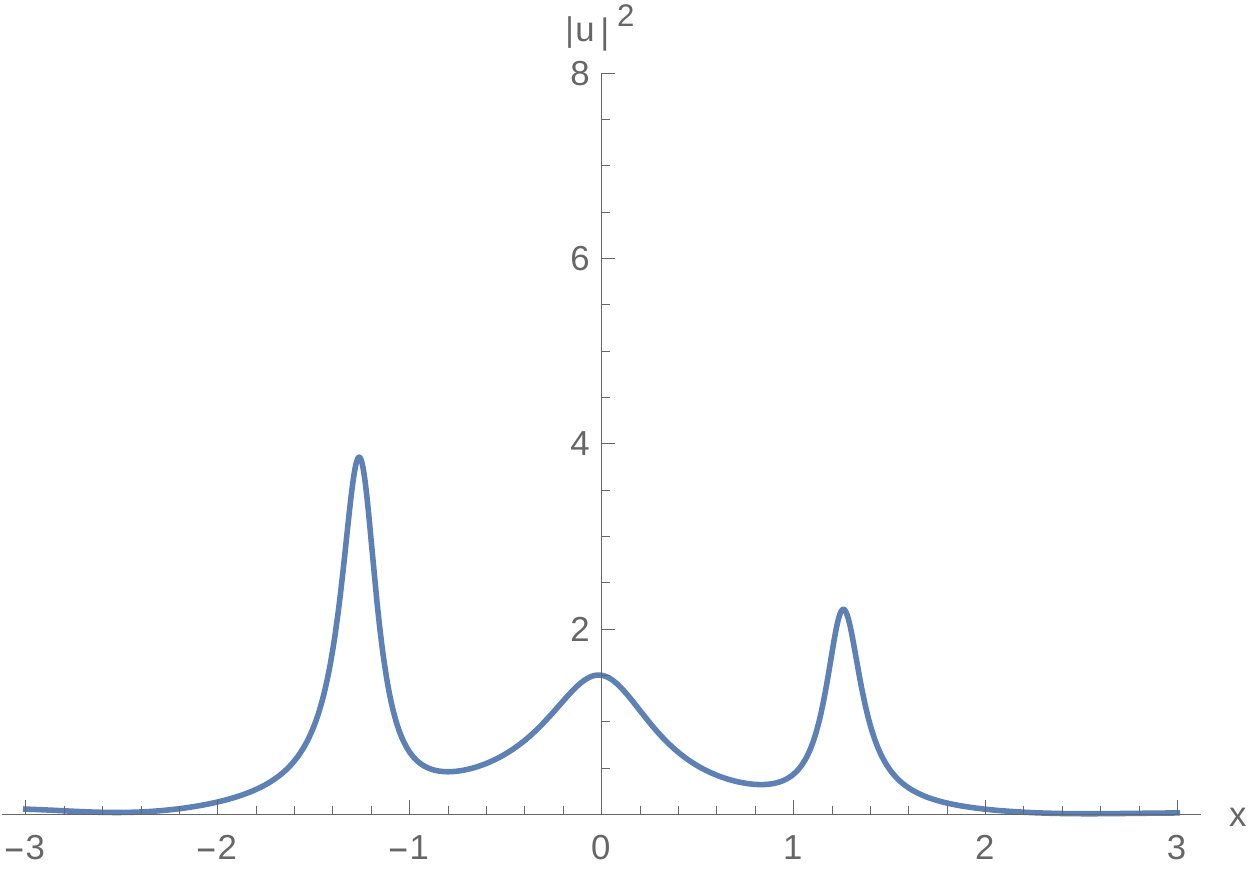}}}
\put(150,-23){\resizebox{!}{3.5cm}{\includegraphics{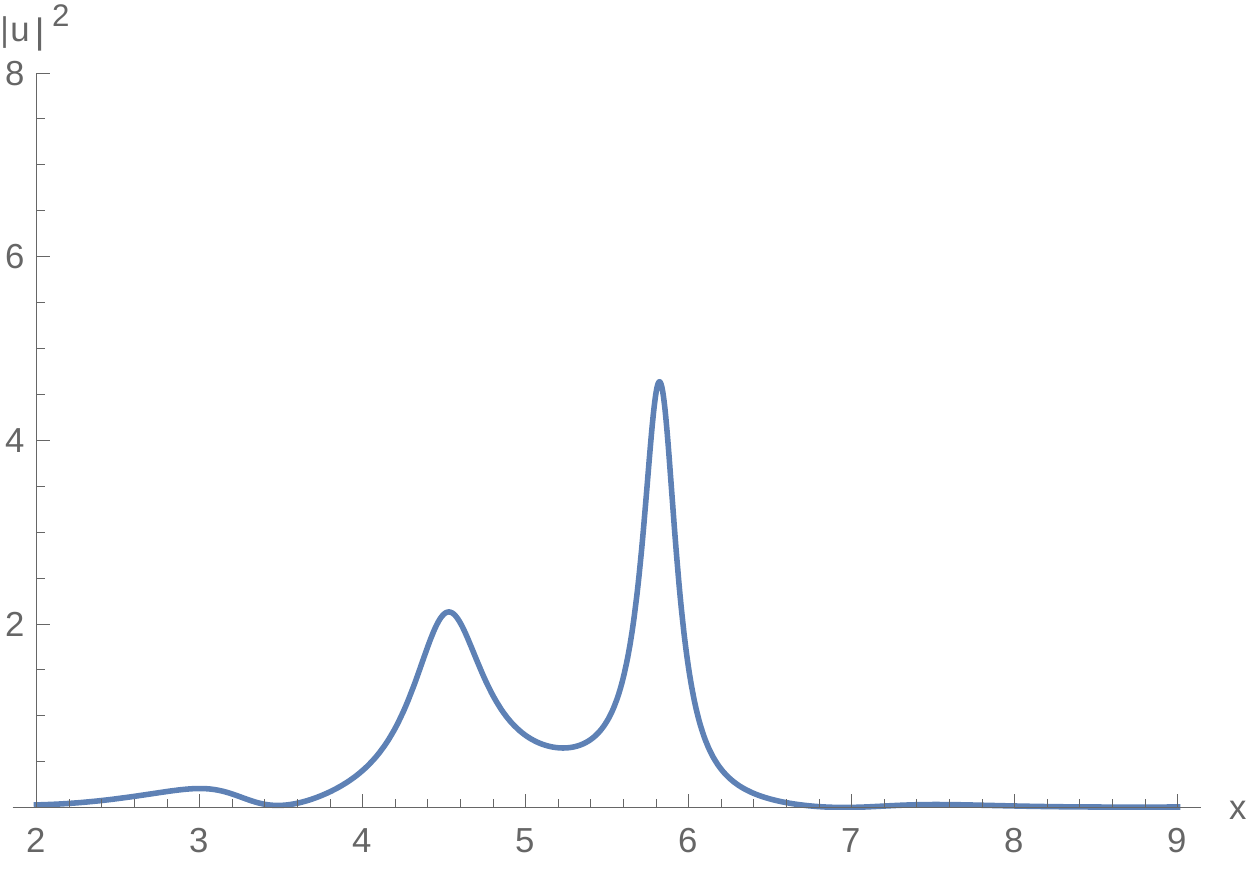}}}
\end{picture}
\end{center}
\vskip 20pt
\begin{center}
\begin{minipage}{16cm}{\footnotesize\quad\qquad\qquad\qquad\quad \quad (c)\quad\qquad\qquad \qquad \qquad \qquad \qquad\qquad(d)\qquad \qquad \qquad\qquad\qquad\qquad\qquad\quad (e)\\
{\bf Fig. 12} (a) shape and motion of Jordan block solution $|u_{23}|^2_{(\varepsilon=1,\sigma=-1)}$
given by \eqref{u34-csG} for $k_1=0.055+i$ and $c_1=d_1=1$.  (b) a contour plot of (a) with range $x\in [-6, 6]$ and $t\in [-4,4]$.
(c) 2D-plot of (a) at $t=-5$. (d) 2D-plot of (a) at $t=0$. (e) 2D-plot of (a) at $t=5$.}
\end{minipage}
\end{center}

\section{Cauchy matrix solutions for the equations \eqref{n-pNLS-1} and \eqref{n-pmKdV-1}}

In this section, we turn our attention to the Theorem \ref{Thm-2} and would like to construct
Cauchy matrix solutions to the local and nonlocal potential NLS equation \eqref{n-pNLS-1} and the local and nonlocal
potential cmKdV equation \eqref{n-pmKdV-1}. We shall use a similar strategy, in other words,
imposing suitable constraints on the pairs $(\Og_1,\Og_2)$, $(\br_1,\br_2)$, $(\bs_1,\bs_2)$
and $(\bM_1, \bM_2)$ in the determining equation set \eqref{DES-C} so that \eqref{2th-pAKNS-so}
coincides with the reductions \eqref{2pAKNS-red} and \eqref{3pAKNS-red}.

\subsection{Local and nonlocal potential NLS equation}

To find the solutions for local and nonlocal potential NLS equation
\eqref{n-pNLS-1} we reduce from the pAKNS(2) equation solution \eqref{2th-pAKNS-so}
led now by the requirement that, we choose $n=2$. The following theorem reveals
the Cauchy matrix solutions for local and nonlocal potential NLS equation \eqref{n-pNLS-1}.
\begin{Thm}
\label{so-nlpNLS}
The function
\begin{align}\label{nlpNLS-solu}
q=\bs^{\st}_2(\bI-\bM_2\bM_1)^{-1}\Og_2^{-1}\br_2-\bs^{\st}_2\bM_2(\bI-\bM_1\bM_2)^{-1}\Og_1^{-1}\br_1-1,
\end{align}
solves the local and nonlocal potential NLS equation \eqref{n-pNLS-1},
where the entities satisfy \eqref{DES-C} $(n=2)$ and simultaneously obey the constraints
\begin{align}\label{nlpNLS-M1M12}
\br_1=\varepsilon \bT \br^*_2(\sigma x),\quad \bs_1=\varepsilon \bT^{\st^{-1}}\bs_2^*(\sigma x),\quad
\bM_1=\bT \bM^*_2(\sigma x)\bT^*,
\end{align}
in which $\bT\in \mathbb{C}_{N\times N}$ is a constant matrix satisfying
\begin{align}\label{nlpNLS-at-eq}
\Og_1\bT+\sigma\bT\Og^*_2=0,\quad \bC^+_1=\varepsilon\bT\bC_2^{+^*},\quad
\bD^+_1=\varepsilon\bT^{\st^{-1}}\bD_2^{+^*},\quad \varepsilon^2=\varepsilon^{*^2}=-\sigma.
\end{align}
\end{Thm}

\begin{proof}
Again this is proven by direct computation.
In terms of the assumptions \eqref{nlpNLS-at-eq}, we have
\begin{align*}
\br_1(x)=& \exp{\Bigl(\Og_1 x+i \Og_1^{2}t \Bigr)}\bC_1^+ \nn \\
=& \exp{\Bigl(-\bT\Og^*_2 \sigma x\bT^{-1}-\bT(i\Og^{2}_2)^*t\bT^{-1}\Bigr)}\bC_1^+ \nn \\
=& \bT\exp{\Bigl(-\Og^*_2 \sigma x-(i\Og^{2}_2)^*t\Bigr)}\bT^{-1}\bC_1^+ \nn \\
=& \varepsilon \bT \br^*_2(\sigma x),
\end{align*}
and similarly one can get the rest relations in \eqref{nlpNLS-M1M12}.
Using the relations \eqref{nlpNLS-M1M12} and \eqref{nlpNLS-at-eq},
a straightforward calculation then yields
\begin{align*}
r(x)=& \bs^{\st}_1(\bI-\bM_1\bM_2)^{-1}\Og_1^{-1}\br_1
-\bs^{\st}_1\bM_1(\bI-\bM_2\bM_1)^{-1}\Og_2^{-1}\br_1-1 \nn \\
=& \bs^{\st^*}_2(\sigma x)(\bI-\bM^*_2(\sigma x)\bM^*_1(\sigma x))^{-1}\Og_2^{*^{-1}}\br^*_2(\sigma x) \nn \\
& \quad  -\bs^{*^{\st}}_2(\sigma x)\bM^*_2(\sigma x)(\bI-\bM^*_1(\sigma x)\bM^*_2(\sigma x))^{-1}\Og_1^{*^{-1}}\br^*_1(\sigma x)-1 \nn \\
=& q^*(\sigma x),
\end{align*}
which coincides with the reduction \eqref{2pAKNS-red}.
\end{proof}

\subsubsection{Exact solutions}

In terms of \eqref{nlpNLS-M1M12} and \eqref{nlpNLS-at-eq}, we write solution \eqref{nlpNLS-solu} in an equivalent form
\begin{subequations}
\label{nlpNLS-solu-r1}
\begin{align}
\label{nlpNLS-solu-1}
q=\bs^{\st}_2(\bI-\bM_2 \bM^*_2(\sigma x))^{-1}\Og_2^{-1}\br_2-\varepsilon^*\bs^{\st}_2\bM_2(\bI-\bM^*_2(\sigma x)\bM_2)^{-1}\Og_2^{*^{-1}}\br^*_2(\sigma x)-1,
\end{align}
where the components $\br_2$ and $\bs_2$ are described as
\begin{align}
\label{rs-2-pNLS}
\br_{2}=\mbox{exp}(-\Og_{2} x-i\Og^2_{2}t)\bC^+_{2},\quad \bs_{2}=\mbox{exp}(-\Og^{\st}_{2}x-i(\Og^{\st}_{2})^2t)\bD^+_{2},
\end{align}
and $\bM_2$ is determined by
\begin{align}
\label{DES-M12-C-pNLS}
\Og_2 \bM_2+\sigma\bM_2\Og^*_2=\varepsilon\br_2\, \bs^{*^{\st}}_2(\sigma x).
\end{align}
\end{subequations}
Here we have used the similar treatment of Remark 1, i.e., absorbing the matrix $\bT$ into $\bM_2$.

In what follows, we construct soliton solutions and Jordan block solutions for the
local and nonlocal potential NLS equation \eqref{n-pNLS-1}. In order to seek the soliton solutions,
we take $\Og_2=\Og_{\ty{D},2}$. In this case, \eqref{nlpNLS-solu-1} can be expressed as
\begin{align}
\label{nlpNLS-solu-so}
q_1=\bs^{\st}_2(\bI-\bM_2\bM^*_2(\sigma x))^{-1}\Og_{\ty{D},2}^{-1}\br_2
-\varepsilon^*\bs^{\st}_2\bM_2(\bI-\bM^*_2(\sigma x)\bM_2)^{-1}\Og_{\ty{D},2}^{*^{-1}}\br^*_2(\sigma x)-1.
\end{align}
We denote, respectively, $\bM^{(1)}_2(x)$, $\bM^{(2)}_2(x)$, $\bM^{(3)}_2(x)$ and $\bM^{(4)}_2(x)$ as
solutions for \eqref{DES-M12-C-pNLS} with $\varepsilon=i$, $\varepsilon=-i$, $\varepsilon=1$ and $\varepsilon=-1$.
It is worthy to note that solution $q_{1,(\varepsilon=-i)}$
is the same as solution $q_{1,(\varepsilon=i)}$ because of $\bM^{(2)}_2(x)=-\bM^{(1)}_2(x)$. Besides,
solution $q_{1,(\varepsilon=-1)}$ is the same as solution $q_{1,(\varepsilon=1)}$ because of $\bM^{(4)}_2(x)=-\bM^{(3)}_2(x)$.
Thus in the following we just present solutions for $\varepsilon=i$ and $\varepsilon=1$.

When $\varepsilon=i$, soliton solution for equation \eqref{n-pNLS-1} reads
\begin{align}\label{nlpNLS-solu-(i)}
q_{1,(\varepsilon=i)}=\bs^{\st}_2(\bI-\bM^{(1)}_2\bM^{(1)^*}_2)^{-1}\Og_2^{-1}\br_2+i\bs^{\st}_2\bM^{(1)}_2(\bI-\bM^{(1)^*}_2\bM^{(1)}_2)^{-1}\Og_2^{*^{-1}}\br^*_2-1,
\end{align}
in which
\begin{subequations}
\begin{align}
& \label{rs-2}
\br_{2}=\mbox{exp}(-\Og_{\ty{D},2} x-i\Og^2_{\ty{D},2}t)\bC^+_{\ty{D},2},\quad \bs_{2}=\mbox{exp}(-\Og_{\ty{D},2}x-i\Og^2_{\ty{D},2}t)\bD^+_{\ty{D},2},\\
& \bM^{(1)}_{2}=i\mbox{exp}(-\Og_{\ty{D},2}x-i\Og^2_{\ty{D},2}t)\bC^-_{\ty{D},2}\cdot
\bG_{\ty{D}}^{(12)}\cdot\bD^{-^*}_{\ty{D},2}\mbox{exp}(-\Og^*_{\ty{D},2}x+i\Og^{*^2}_{\ty{D},2}t),
\end{align}
\end{subequations}
where $\bC^{\pm}_{\ty{D},2},~\bD^{\pm}_{\ty{D},2}$ are given in \eqref{CD-GDD} and $\bG^{(12)}_{\ty{D}}$ is
defined by \eqref{G12}. Specially, 1-soliton solution is
\begin{align}
\label{q-solu-nls-(i)}
q_{1,(\varepsilon=i)}=\frac{k_1^*(\vartheta_1e^{-\varpi}-k_1)-|\vartheta_1|^2k_1^2e^{-(\varpi+\varpi^*)+\theta_{11}}}{
|k_1|^2(1-|\vartheta_1|^2e^{-(\varpi+\varpi^*)+\theta_{11}})},
\end{align}
where $\varpi=2k_1(x+ik_1t)$.

When $\varepsilon=1$, soliton solution of this case is given by
\begin{align}
\label{nlpNLS-solu-(1)}
q_{1,(\varepsilon=1)}=& \bs^{\st}_2(\bI-\bM^{(3)}_2\bM^{(3)^*}_2(-x))^{-1}\Og_2^{-1}\br_2 \nn \\
& -\bs^{\st}_2\bM^{(3)}_2(\bI-\bM^{(3)^*}_2(-x)\bM^{(3)}_2)^{-1}\Og_2^{*^{-1}}\br^*_2(-x)-1,
\end{align}
in which $\br_{2}$ and $\bs_{2}$ are the same as \eqref{rs-2} and
\begin{align}
\bM^{(3)}_{2}=\mbox{exp}(-\Og_{\ty{D},2}x-i\Og^2_{\ty{D},2}t)\bC^-_{\ty{D},2}\cdot
\bG_{\ty{D}}^{(34)}\cdot\bD^{-^*}_{\ty{D},2}\mbox{exp}(\Og^*_{\ty{D},2}x+i\Og^{*^2}_{\ty{D},2}t),
\end{align}
where $\bC^{\pm}_{\ty{D},2},~\bD^{\pm}_{\ty{D},2}$ are given in \eqref{CD-GDD} and
$\bG^{(34)}_{\ty{D}}$ is defined by \eqref{G34}. The 1-soliton solution is
\begin{align}
\label{q-solu-nls-(1)}
q_{1,(\varepsilon=1)}=\frac{k_1^*(\vartheta_1e^{-\varpi}-k_1)-|\vartheta_1|^2k_1^2e^{\ell-\varpi+\epsilon_{11}}}{
|k_1|^2(1+|\vartheta_1|^2e^{\ell-\varpi+\epsilon_{11}})},
\end{align}
where $\ell=2k_1^*(x+ik_1^{*}t)$.

For deriving the Jordan block solutions, we take $\Og_2=\Og_{\ty{J},2}$. Some solutions are listed as follows:
\begin{subequations}
\begin{align*}
& q_{2,(\varepsilon=i)}=\bs^{\st}_2(\bI-\wh{\bM}_2\wh{\bM}^*_2)^{-1}\Og^{-1}_{\ty{J},2}\br_2
+i\bs^{\st}_2\wh{\bM}_2(\bI-\wh{\bM}^*_2\wh{\bM}_2)^{-1}\Og_{\ty{J},2}^{*^{-1}}\br^*_2-1, \\
& q_{2,(\varepsilon=1)}=\bs^{\st}_2(\bI-\wt{\bM}_2\wt{\bM}^*_2(-x))^{-1}\Og^{-1}_{\ty{J},2}\br_2
-\bs^{\st}_2\wt{\bM}_2(\bI-\wt{\bM}^*_2(-x)\wt{\bM}_2)^{-1}\Og_{\ty{J},2}^{*^{-1}}\br^*_2(-x)-1,
\end{align*}
\end{subequations}
where
\begin{subequations}
\begin{align*}
& \br_{2}=\mbox{exp}(-\Og_{\ty{J},2} x-i\Og^2_{\ty{J},2}t)\bC^+_{\ty{J},2},\quad \bs_{2}=\mbox{exp}(-\Og^{\st}_{\ty{J},2}x-i(\Og^{\st}_{\ty{J},2})^2t)\bD^+_{\ty{J},2}, \\
& \wh{\bM}_2=i\mbox{exp}(-\Og_{\ty{J},2} x-i\Og^2_{\ty{J},2}t)\bC_{\ty{J},2}^- \cdot\wh{\bG}_{\ty{J}}^{\st}\cdot\bD_{\ty{J},2}^{-^*}
\mbox{exp}(-\Og_{\ty{J},2}^{*^{\st}}x+i(\Og_{\ty{J},2}^{*^{\st}})^2t), \\
& \wt{\bM}_2=\mbox{exp}(-\Og_{\ty{J},2} x-i\Og^2_{\ty{J},2}t)\bC_{\ty{J},2}^- \cdot\wt{\bG}_{\ty{J}}^{\st}\cdot\bD_{\ty{J},2}^{-^*}
\mbox{exp}(\Og_{\ty{J},2}^{*^{\st}}x+i(\Og_{\ty{J},2}^{*^{\st}})^2t).
\end{align*}
\end{subequations}

\subsubsection{Dynamics}

Now let us examine the dynamics of $|q_1|^2$ given by \eqref{q-solu-nls-(i)} and \eqref{q-solu-nls-(1)}.
Importing the expansion \eqref{k1-com} and
\begin{align}
\label{cd1-cmKdV}
c_1=c^{(1)}+i c^{(2)}, \quad d_1=d^{(1)}+i d^{(2)}
\end{align}
into solution \eqref{q-solu-nls-(i)} gives rise to $|q_1|^2_{(\varepsilon=i)}=\frac{q_{\ty{D},2}}{q_{\ty{D},1}}$ with
\begin{subequations}
\label{q12-pNLS}
\begin{align}
& q_{\ty{D},1}=(\alpha^2+\beta^2)^2(4\alpha^2e^{2\varsigma_1}-|\vartheta_1|^2e^{2\varsigma_2})^2, \\
& q_{\ty{D},2}=(\alpha^2+\beta^2)^2(16\alpha^4 e^{4\varsigma_1}+|\vartheta_1|^4e^{4\varsigma_2})+8\alpha^2
[2\alpha^2(\Theta_1^2+\Theta_2^2)+|\vartheta_1|^2(\alpha^4-\beta^4)] \nn \\
&\qquad \cdot e^{2(\varsigma_1+\varsigma_2)}-32 \alpha^4(\alpha^2+\beta^2)[\Theta_1\sin\iota_1+\Theta_2\cos\iota_1]
e^{3\varsigma_1+\varsigma_2}-8\alpha^2|\vartheta_1|^2[(2\alpha\beta \Theta_1 \nn \\
&\qquad +\Theta_2(\alpha^2-\beta^2))\cos\iota_1-(2\alpha\beta \Theta_2-\Theta_1(\alpha^2-\beta^2))\sin\iota_1]e^{\varsigma_1+3\varsigma_2},
\end{align}
\end{subequations}
where and whereafter $\varsigma_1=2\alpha x,~~\varsigma_2=4\alpha\beta t,~~\iota_1=2(\beta x+(\alpha^2-\beta^2)t)$ and
\begin{subequations}
\begin{align}
& \Theta_1=(c^{(1)}d^{(2)}+c^{(2)}d^{(1)})\alpha-(c^{(1)}d^{(1)}-c^{(2)}d^{(2)})\beta, \\
& \Theta_2=(c^{(1)}d^{(1)}-c^{(2)}d^{(2)})\alpha+(c^{(1)}d^{(2)}+c^{(2)}d^{(1)})\beta.
\end{align}
\end{subequations}
There is a quasi-periodic phenomenon because of the involvement
of sine function and cosine function in numerator.
It is obvious that $|q_1|^2_{(\varepsilon=i)}$ has singularity along with straight line
$x=2\beta t-\hbar$. The traveling speed is $2\beta$. It is worthy to notice
$|q_1|^2_{(\varepsilon=i)}(x,t=\frac{2\kappa \pi-\ln \beta}{4\beta^2}) \rightarrow 0$
as $x\rightarrow \frac{\kappa \pi}{\beta}$ ($\kappa \in \mathbb{Z}$) when $k_1=c_1=d_1$. In particular,
$|q_1|^2_{(\varepsilon=i)}(x,t=-\frac{\ln \beta}{4\beta^2})$ is an even function of $x$
and has limit $|q_1|^2_{(\varepsilon=i)}(x,t=-\frac{\ln \beta}{4\beta^2}) \rightarrow 0$
as $x\rightarrow 0$ (see Fig. 13(e)).
Besides, one can recognize that for fixed $t$, function $|q_1|^2_{(\varepsilon=i)}\rightarrow 1$ as
$x\rightarrow \infty$, which means the background level here is 1.
We depict this solution in Fig. 13.


\begin{center}
\begin{picture}(120,100)
\put(-150,-23){\resizebox{!}{3.5cm}{\includegraphics{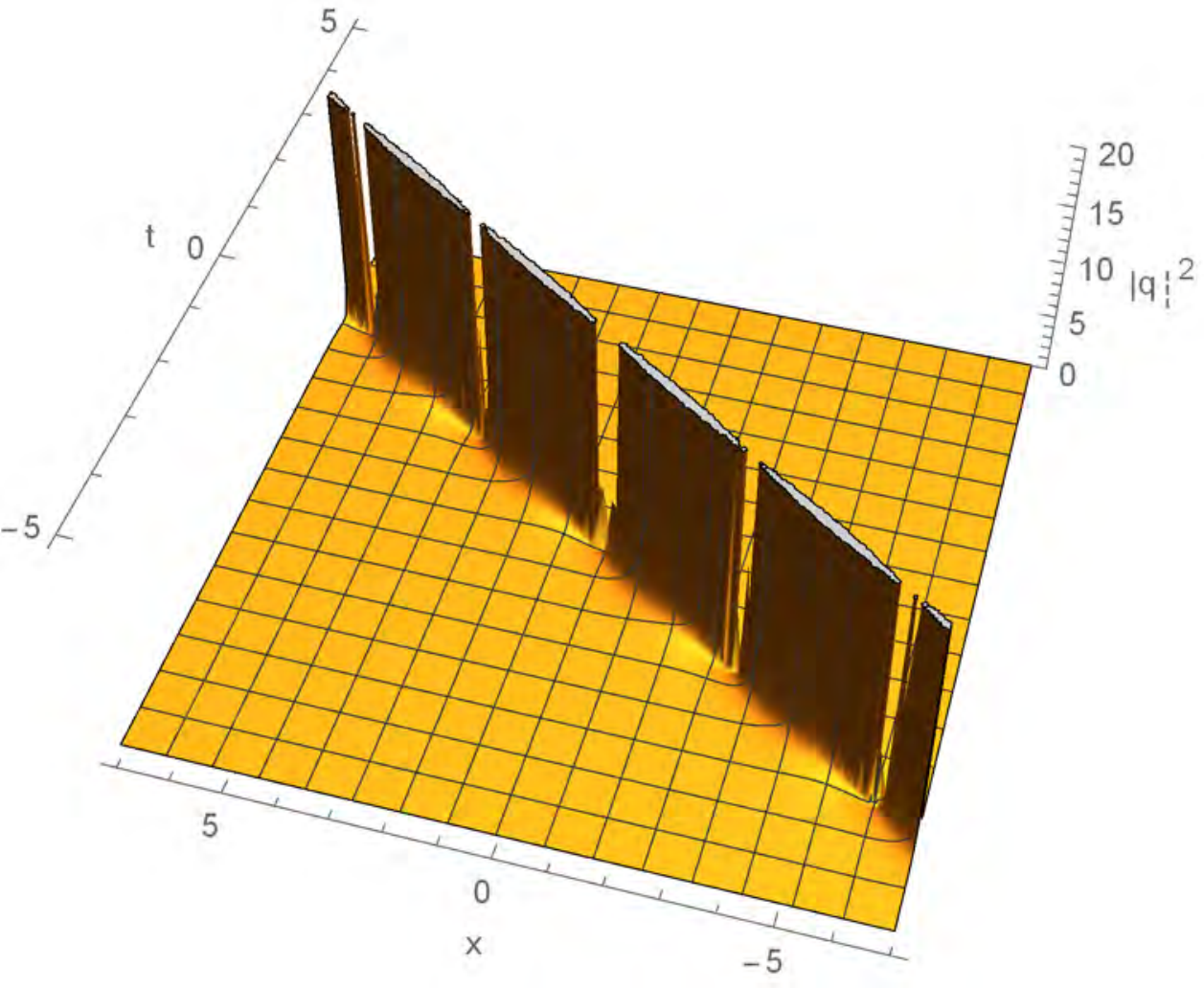}}}
\put(10,-23){\resizebox{!}{3cm}{\includegraphics{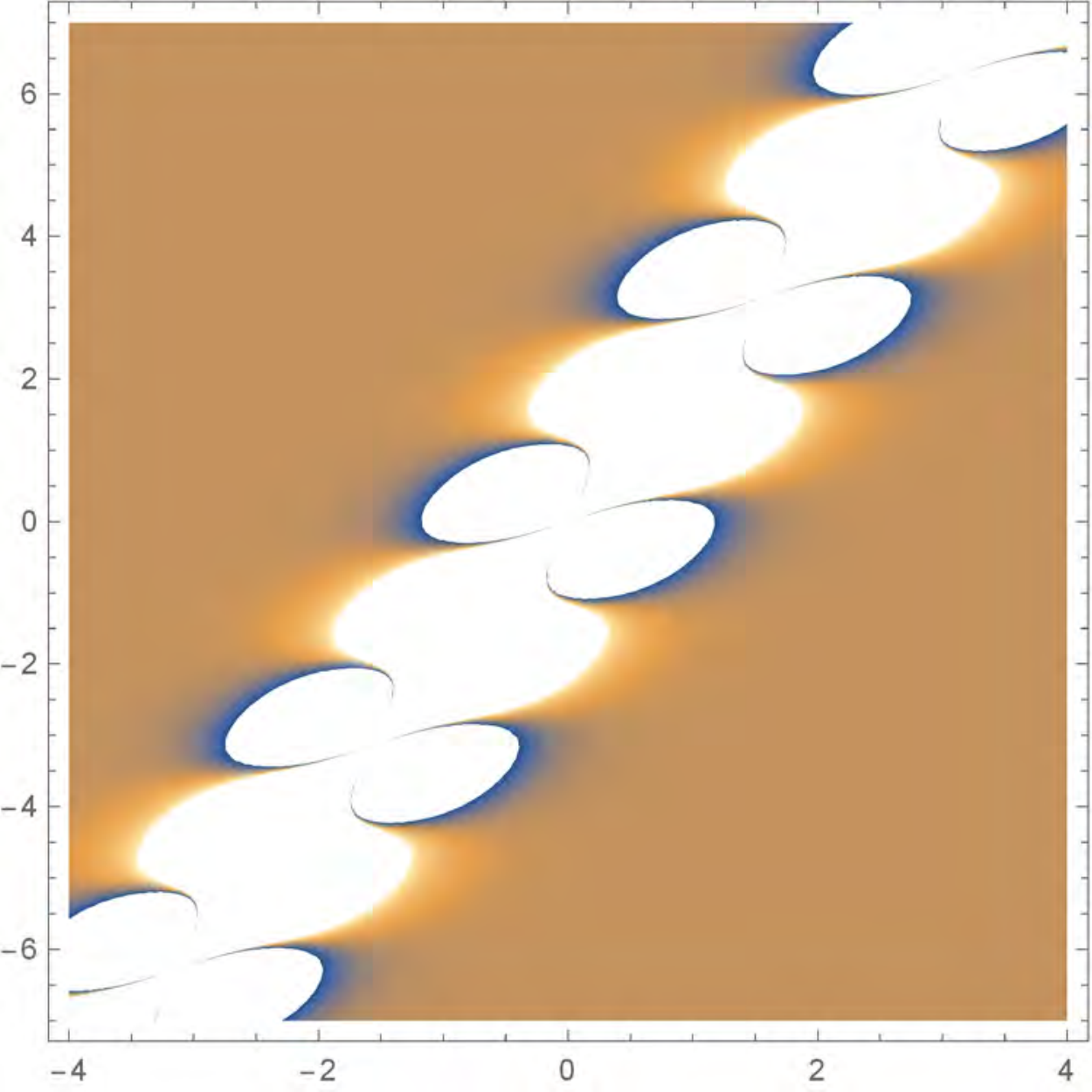}}}
\put(150,-23){\resizebox{!}{3cm}{\includegraphics{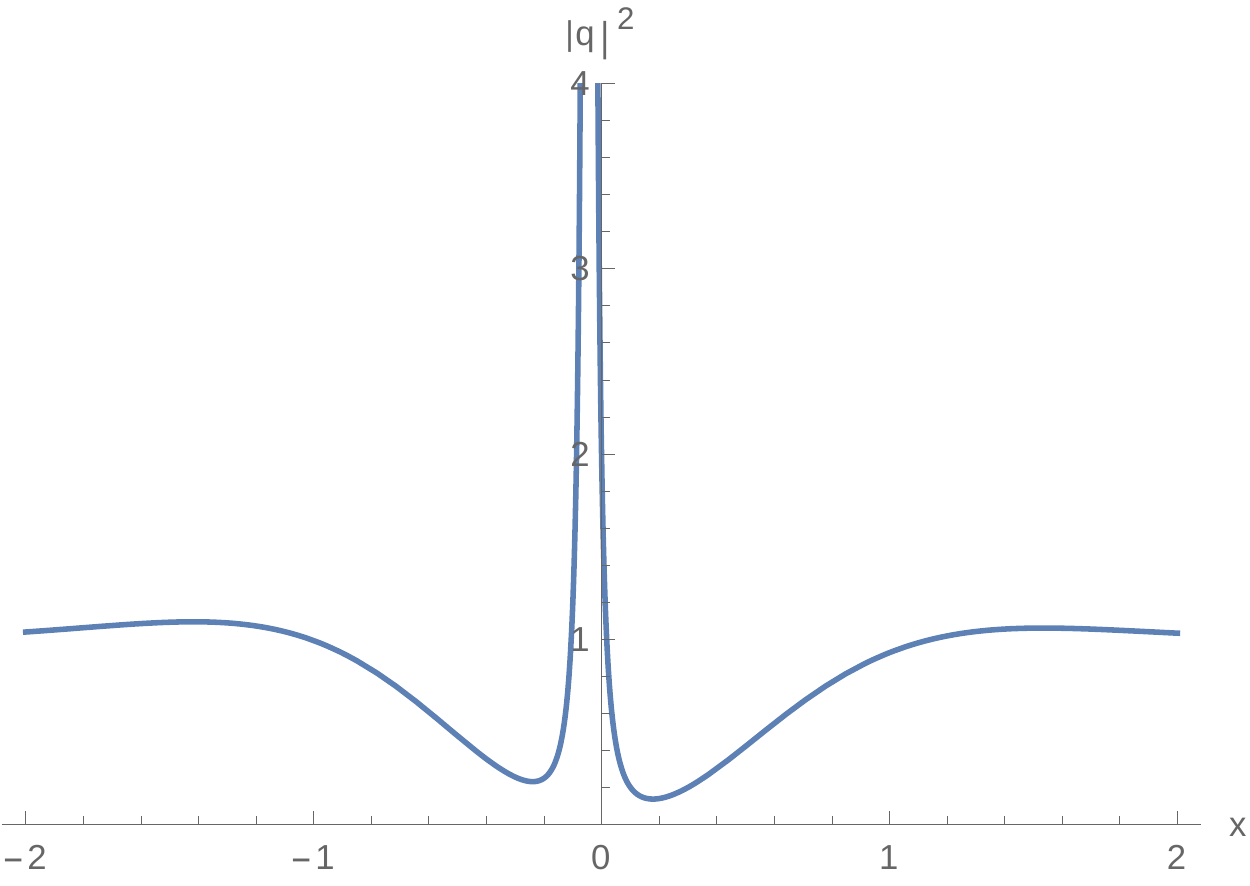}}}
\end{picture}
\end{center}
\vskip 20pt
\begin{center}
\begin{minipage}{16cm}{\footnotesize
\quad\qquad\qquad\qquad(a)\qquad\qquad\qquad\qquad \qquad\quad \qquad \qquad (b) \qquad\qquad \qquad \qquad\qquad\qquad\qquad \qquad (c)}
\end{minipage}
\end{center}
\vskip 10pt
\begin{center}
\begin{picture}(120,80)
\put(-150,-23){\resizebox{!}{3cm}{\includegraphics{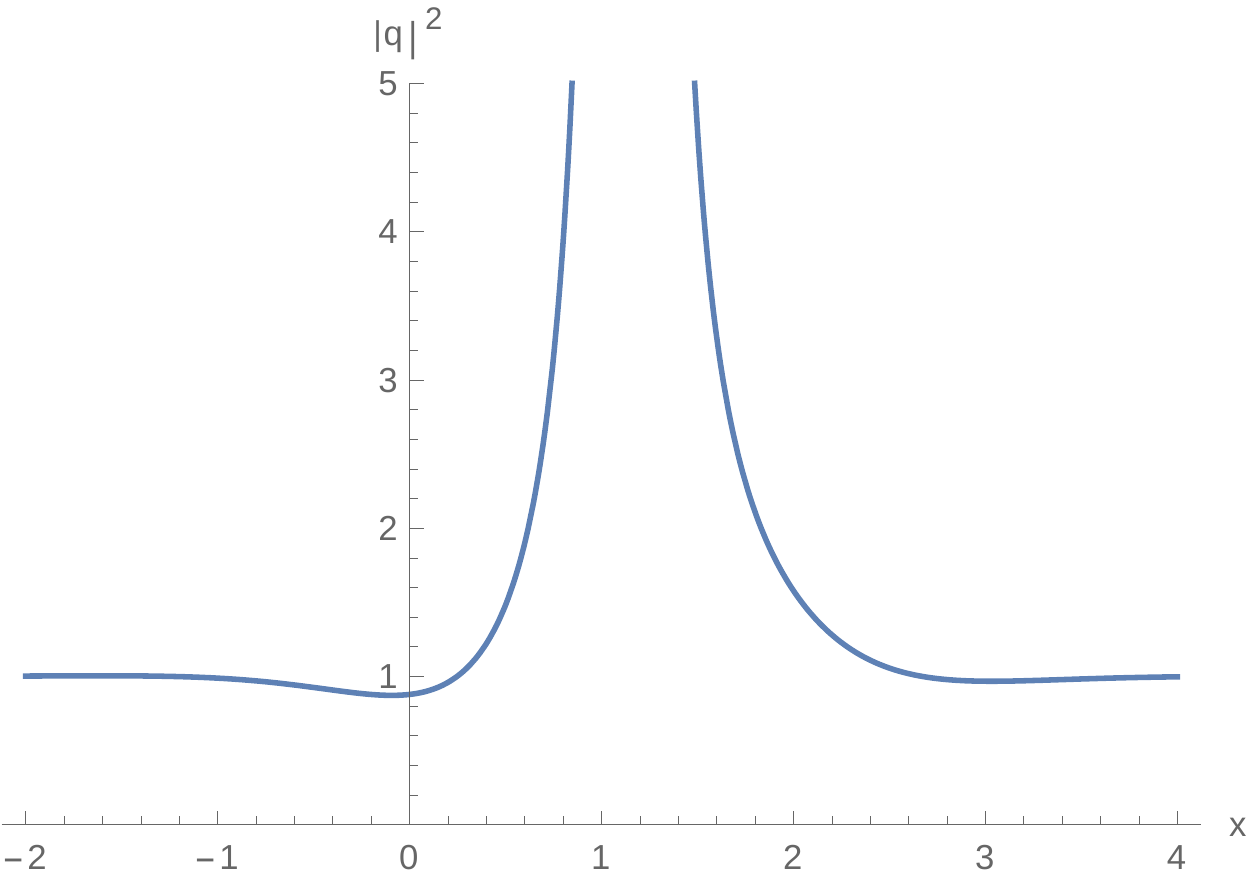}}}
\put(10,-23){\resizebox{!}{3cm}{\includegraphics{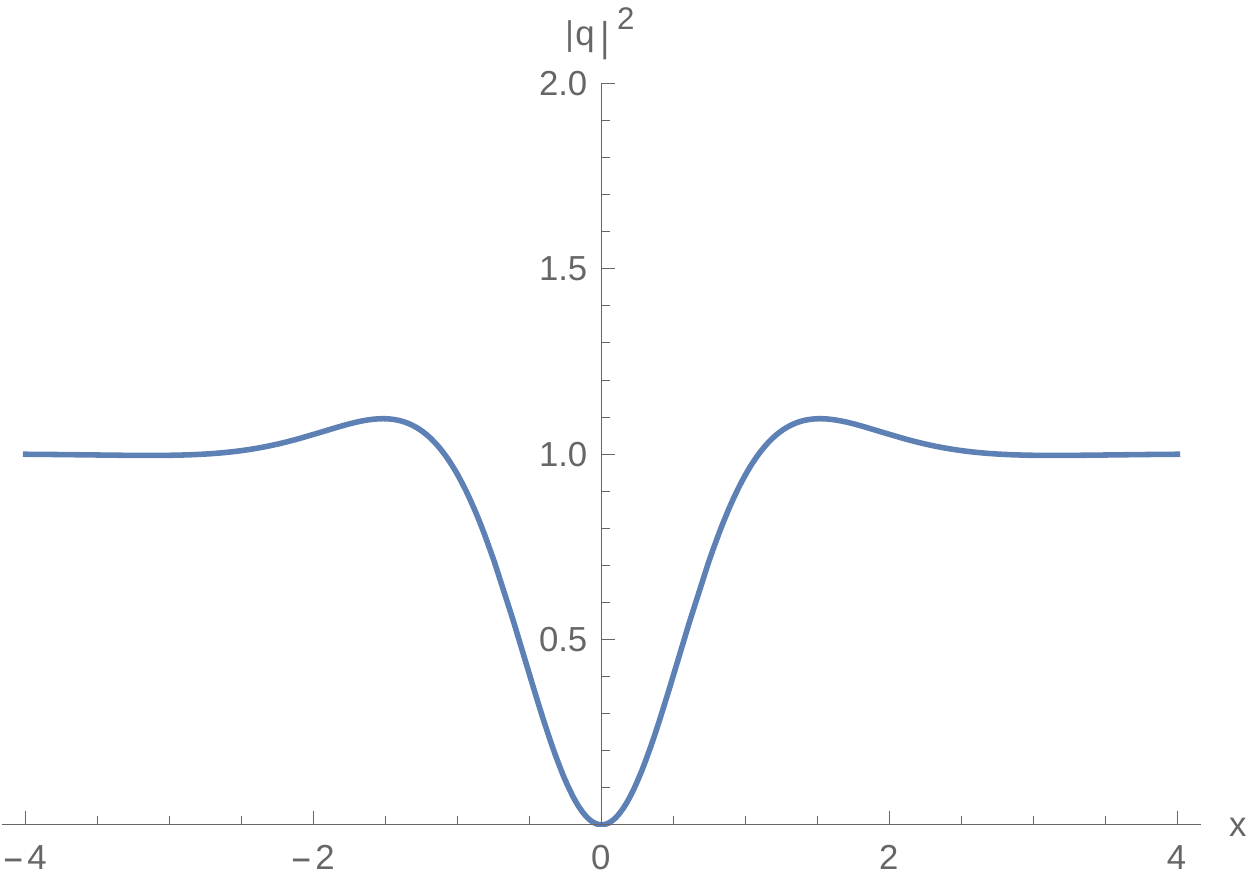}}}
\put(150,-23){\resizebox{!}{3cm}{\includegraphics{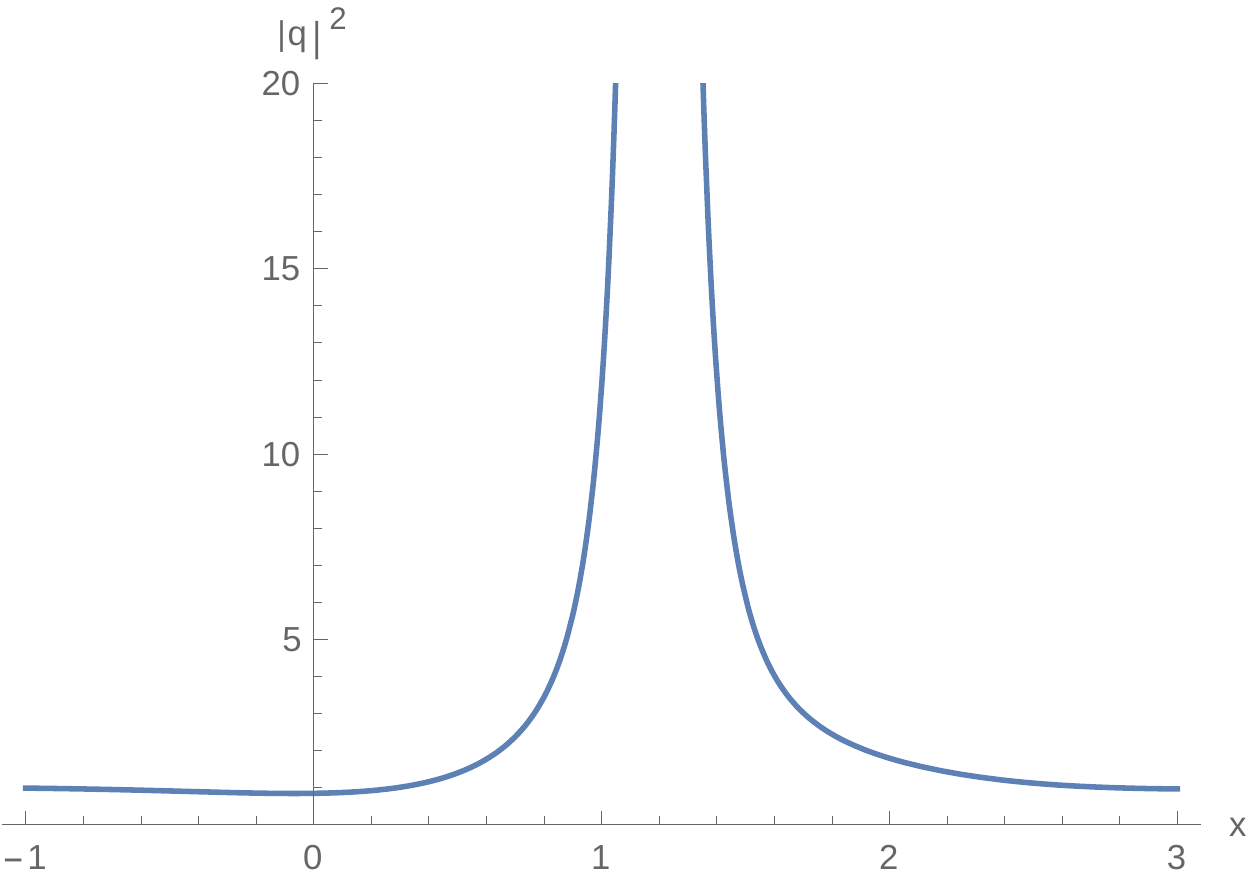}}}
\end{picture}
\end{center}
\vskip 20pt
\begin{center}
\begin{minipage}{16cm}{\footnotesize
\quad\qquad\qquad\qquad(d)\qquad\qquad\qquad\qquad \qquad\quad \qquad \qquad \qquad (e) \qquad \qquad \qquad\qquad\qquad\qquad \qquad (f)\\
{\bf Fig. 13} shape and motion of $|q_1|^2_{(\varepsilon=i)}$ for $c_1=d_1=1+i$: (a)
3D-plot for for $k_1=1+i$. (b) a contour plot of (a) with range
$x\in [-4, 4]$ and $t\in [-7,7]$. (c) and (d) 2D-plot for $k_1=1.1+i$ at $t=0$ and $t=0.6$, respectively. (e) and (f)
2D-plot for $k_1=1+i$ at $t=0$ and $t=0.6$, respectively.}
\end{minipage}
\end{center}


To proceed, we next consider \eqref{q-solu-nls-(1)}. Under the expansions \eqref{k1-com} and
\eqref{cd1-cmKdV} we have $|q_1|^2_{(\varepsilon=1)}=\frac{q_{\ty{D},4}}{q_{\ty{D},3}}$ with
\begin{subequations}
\label{q34-pNLS}
\begin{align}
& q_{\ty{D},4}=(\alpha^2+\beta^2)^2|\vartheta_1|^4e^{4\varsigma_2}+16\beta^4\big[(\Theta^2_1+\Theta^2_2)
e^{2(\varsigma_2-\varsigma_1)}+(\alpha^2+\beta^2)^2-2(\alpha^2+\beta^2)(\Theta_1 \sin\iota_1 \nn \\
& \quad +\Theta_2 \cos\iota_1)e^{\varsigma_2-\varsigma_1}\big]+8|\beta \vartheta_1|^2 \big[(\beta^4-\alpha^4)
\cos\iota_2-2\alpha\beta(\alpha^2+\beta^2)\sin\iota_2+\big((2\alpha\beta \Theta_1 \nn \\
& \quad +\Theta_2(\alpha^2-\beta^2))\cos(\iota_2-\iota_1)+(2\alpha\beta \Theta_2-\Theta_1(\alpha^2-\beta^2))
\sin(\iota_2-\iota_1)\big)e^{\varsigma_2-\varsigma_1}\big]e^{2\varsigma_2}, \\
& q_{\ty{D},3}=(\alpha^2+\beta^2)^2(16\beta^4+|\vartheta_1|^4e^{4\varsigma_2}-8|\beta \vartheta_1|^2e^{2\varsigma_2}\cos\iota_2),
\end{align}
\end{subequations}
where $\iota_2=4\beta x$. Solution $|q_1|^2_{(\varepsilon=1)}$ has singularities along points
\begin{align}
(x,t)=\left(\frac{\kappa \pi}{2\beta},\frac{1}{4\alpha\beta}\ln \frac{2\beta}{|\vartheta_1|}\right), \quad \kappa\in \mathbb{Z}.
\end{align}
and possesses quasi-periodic phenomenon.
We depict this solution in Fig. 14.


\begin{center}
\begin{picture}(120,90)
\put(-150,-23){\resizebox{!}{3.5cm}{\includegraphics{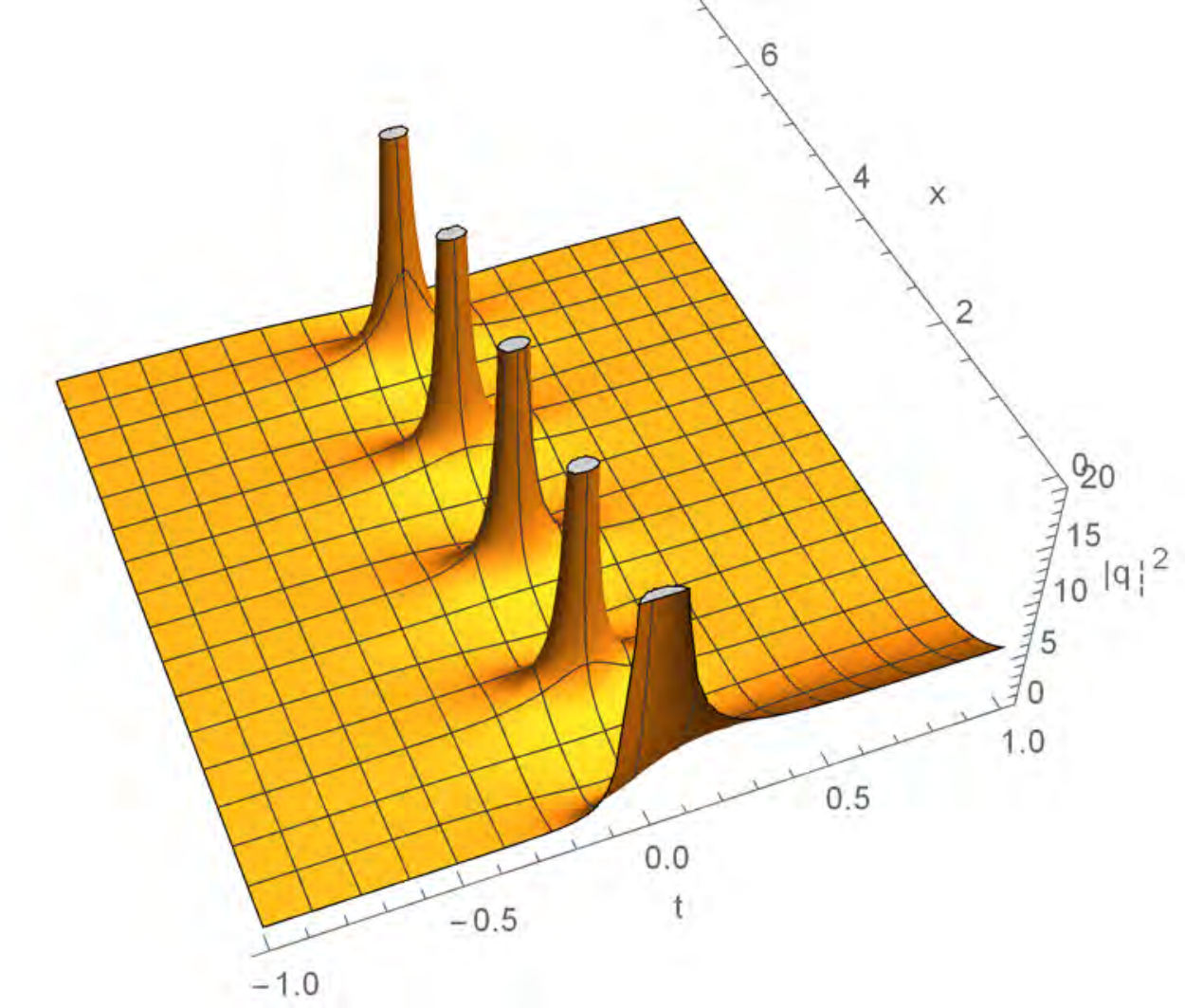}}}
\put(10,-23){\resizebox{!}{3cm}{\includegraphics{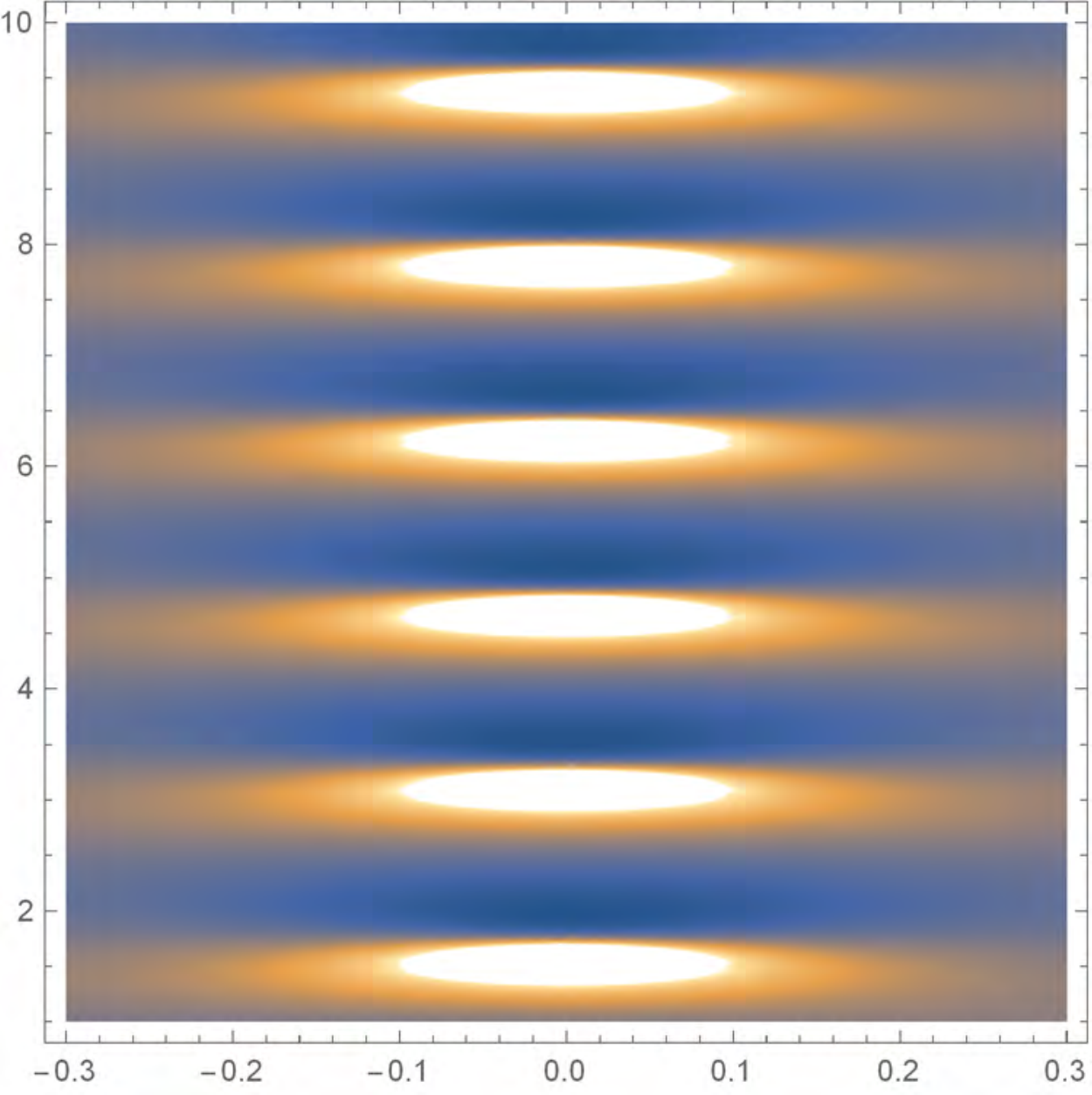}}}
\put(150,-23){\resizebox{!}{3cm}{\includegraphics{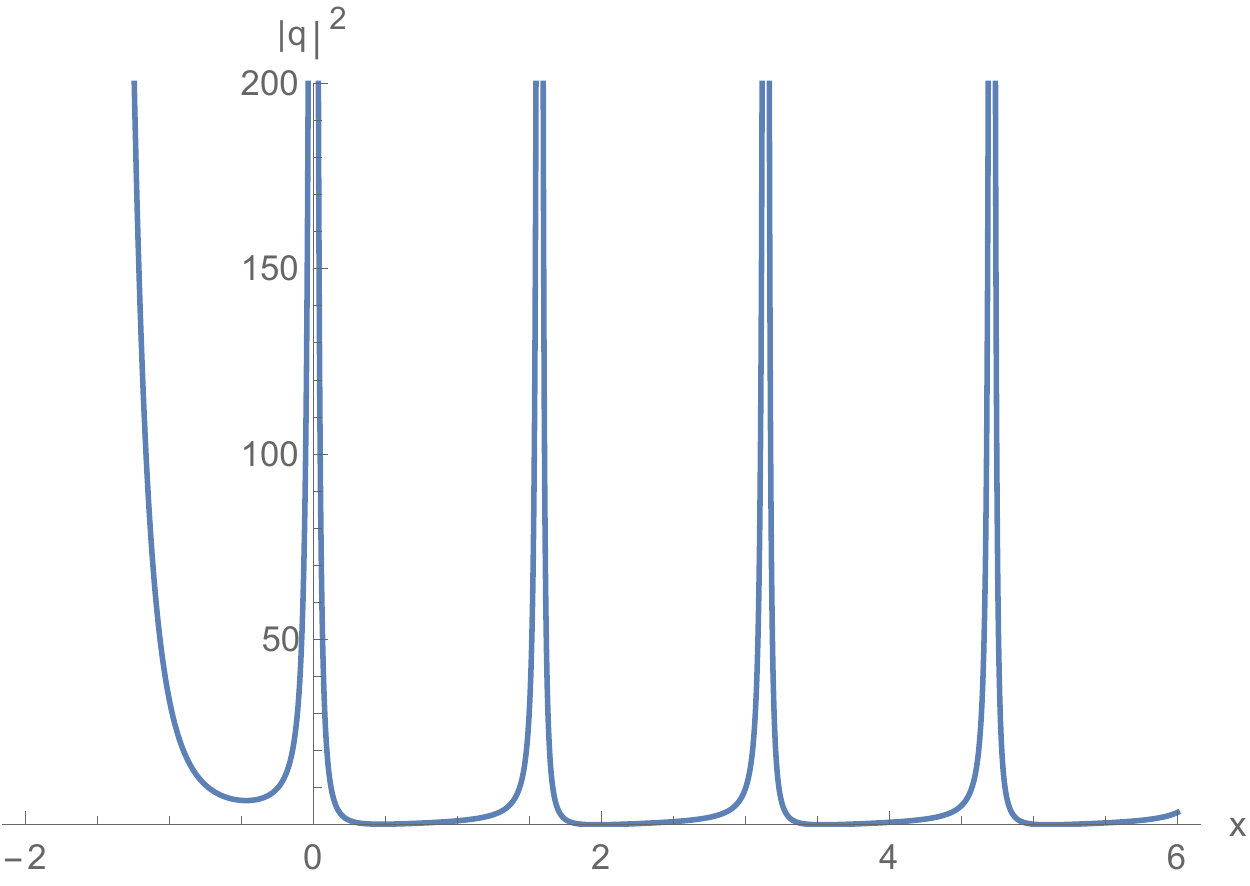}}}
\end{picture}
\end{center}
\vskip 20pt
\begin{center}
\begin{minipage}{16cm}{\footnotesize
\quad\qquad\qquad\qquad(a)\qquad\qquad\qquad\qquad \qquad\quad \qquad \qquad (b) \qquad\qquad \qquad \qquad\qquad\qquad\qquad \qquad (c)\\
{\bf Fig. 14} (a) shape and motion of $|q_1|^2_{(\varepsilon=1)}$ for $k_1=c_1=d_1=1+i$.
(b) a contour plot of (a) with range $x\in [1, 10]$ and $t\in [-0.3,0.3]$. (c) 2D-plot at $t=0$.}
\end{minipage}
\end{center}


\subsection{Local and nonlocal potential cmKdV equation}

We finally consider the construction of Cauchy matrix solutions to the local and nonlocal potential cmKdV equation
\eqref{n-pmKdV-1}. To do so, we take $n=3$ in \eqref{DES-C} and summarize the result in the following Theorem.
\begin{Thm}
\label{so-nlcpmKdV}
The function
\begin{align}\label{nlcpmKdV-solu}
q=\bs^{\st}_2(\bI-\bM_2\bM_1)^{-1}\Og_2^{-1}\br_2-\bs^{\st}_2\bM_2(\bI-\bM_1\bM_2)^{-1}\Og_1^{-1}\br_1-1,
\end{align}
solves the local and nonlocal potential cmKdV equation \eqref{n-pmKdV-1},
where the entities satisfy \eqref{DES-C} $(n=3)$ and simultaneously obey the constraints
\begin{align}\label{nlcpmKdV-M1M12}
\br_1=\varepsilon \bT \br^*_2(\sigma x,\sigma t),\quad \bs_1=\varepsilon \bT^{\st^{-1}}\bs_2^*(\sigma x,\sigma t),\quad
\bM_1=\bT \bM^*_2(\sigma x,\sigma t)\bT^*,
\end{align}
in which $\bT\in \mathbb{C}_{N\times N}$ is a constant matrix satisfying
\begin{align}\label{nlcpmKdV-at-eq}
\Og_1\bT+\sigma\bT\Og^*_2=0,\quad \bC^+_1=\varepsilon\bT\bC_2^{+^*},\quad
\bD^+_1=\varepsilon\bT^{\st^{-1}}\bD_2^{+^*},\quad \varepsilon^2=\varepsilon^{*^2}=-\sigma.
\end{align}
\end{Thm}
We skip the verification here and proceed now with deriving Cauchy matrix-type soliton solutions and
Jordan block solutions for the local and nonlocal potential cmKdV equation \eqref{n-pmKdV-1}.

\subsubsection{Exact solutions}

Applying \eqref{nlcpmKdV-M1M12} and \eqref{nlcpmKdV-at-eq} to \eqref{nlcpmKdV-solu}, we can rewrite solution \eqref{nlcpmKdV-solu} as
\begin{align}\label{nlcpmKdV-solu-1}
& q=\bs^{\st}_2(\bI-\bM_2\bM^*_2(\sigma x,\sigma t))^{-1}\Og_2^{-1}\br_2 \nn \\
& \qquad -\varepsilon^*\bs^{\st}_2\bM_2(\bI-\bM^*_2(\sigma x,\sigma t)\bM_2)^{-1}\Og_2^{*^{-1}}\br^*_2(\sigma x,\sigma t)-1,
\end{align}
where we have absorbed $\bT$ into $\bM_2$ and the components are given by \eqref{rsj-mKdV} and \eqref{DES-M12-C-mKdV-nt}. Analogously, we just
present the solutions correspond to $\varepsilon=i$ and $\varepsilon=1$.

Soliton solutions in these two cases read
\begin{subequations}
\begin{align}\
\label{nlcpmKdV-solu-(i)}
q_{3,(\varepsilon=i)}= & \bs^{\st}_2(\bI-\bM^{(1)}_2\bM^{(1)^*}_2)^{-1}\Og_{\ty{D},2}^{-1}\br_2
+i\bs^{\st}_2\bM^{(1)}_2(\bI-\bM^{(1)^*}_2\bM^{(1)}_2)^{-1}\Og_{\ty{D},2}^{*^{-1}}\br^*_2-1, \\
\label{nlcpmKdV-solu-(1)}
q_{3,(\varepsilon=1)}= &  \bs^{\st}_2(\bI-\bM^{(3)}_2\bM^{(3)^*}_2(-x,-t))^{-1}\Og_{\ty{D},2}^{-1}\br_2 \nn \\
& -\bs^{\st}_2\bM^{(3)}_2(\bI-\bM^{(3)^*}_2(-x,-t)\bM^{(3)}_2)^{-1}\Og_{\ty{D},2}^{*^{-1}}\br^*_2(-x,-t)-1,
\end{align}
\end{subequations}
where $\br_{2},~\bs_{2},~\bM^{(1)}_{2}$ and $\bM^{(3)}_{2}$ are
given by \eqref{rsM-mKdV-D}. The corresponding 1-soliton solutions are expressed by
\begin{subequations}
\label{q-solu-mKdV-(i1)}
\begin{align}
\label{q-solu-mKdV-(i)}
q_{3,(\varepsilon=i)}=& \frac{k_1^*(\vartheta_1e^{\xi_1}-k_1)-|\vartheta_1|^2k_1^2e^{\xi_1+\xi^*_1+\theta_{11}}}{
|k_1|^2(1-|\vartheta_1|^2e^{\xi_1+\xi^*_1+\theta_{11}})}, \\
\label{q-solu-cpmKdV-(1)}
q_{3,(\varepsilon=1)}=& \frac{k_1^*(\vartheta_1e^{\xi_1}-k_1)-|\vartheta_1|^2k_1^2e^{\xi_1-\xi^*_1+\epsilon_{11}}}{
|k_1|^2(1+|\vartheta_1|^2e^{\xi_1-\xi^*_1+\epsilon_{11}})}.
\end{align}
\end{subequations}

The Jordan block solutions are
\begin{subequations}
\begin{align*}
q_{4,(\varepsilon=i)} = & \bs^{\st}_2(\bI-\wh{\bM}_2\wh{\bM}^*_2)^{-1}\Og^{-1}_{\ty{J},2}\br_2
+i\bs^{\st}_2\wh{\bM}_2(\bI-\wh{\bM}^*_2\wh{\bM}_2)^{-1}\Og_{\ty{J},2}^{*^{-1}}\br^*_2-1, \\
q_{4,(\varepsilon=1)}= & \bs^{\st}_2(\bI-\wt{\bM}_2\wt{\bM}^*_2(-x,-t))^{-1}\Og^{-1}_{\ty{J},2}\br_2 \\
& -\bs^{\st}_2\wt{\bM}_2(\bI-\wt{\bM}^*_2(-x,-t)\wt{\bM}_2)^{-1}\Og_{\ty{J},2}^{*^{-1}}\br^*_2(-x,-t)-1,
\end{align*}
\end{subequations}
where $\br_2,~\bs_2,~\wh{\bM}_2$ and $\wt{\bM}_2$ are given by \eqref{rsj-mKdV-OM}.

\subsubsection{Dynamics}

We now study the dynamics of soliton solutions \eqref{q-solu-mKdV-(i1)}.
Taking \eqref{k1-com} and \eqref{cd1-cmKdV} into \eqref{q-solu-mKdV-(i)} we get
$ |q_3|^2_{(\varepsilon=i)}=\frac{q'_{\ty{D},2}}{q'_{\ty{D},1}}$ with
\begin{subequations}
\begin{align*}
& q'_{\ty{D},2}=16(\alpha^4(\Theta_1^2+\Theta_2^2)-|\alpha\beta \vartheta_1|^2(\alpha^2+\beta^2))e^{2(\mu_1+\mu_2)} \nn \\
&\quad\quad +(\alpha^2+\beta^2)^2(4\alpha^2e^{2\mu_1}+|\vartheta_1|^2e^{2\mu_2})^2 -32\alpha^4(\alpha^2+\beta^2)(\Theta_1\sin\nu_1+
\Theta_2\cos\nu_1)e^{3\mu_1+\mu_2}\nn \\
&\quad\quad
-8|\alpha \vartheta_1|^2\big((\Theta_1(\alpha^2-\beta^2)-2\alpha\beta \Theta_2)\sin\nu_1+(2\alpha\beta \Theta_1+(\alpha^2-\beta^2)\Theta_2)\cos\nu_1\big)e^{\mu_1+3\mu_2}, \\
& q'_{\ty{D},1}=((\alpha^2+\beta^2)(4\alpha^2e^{2\mu_1}-|\vartheta_1|^2e^{2\mu_2}))^2,
\end{align*}
\end{subequations}
where $\mu_1=2\alpha(x+3\beta^2t)$, $\mu_2=2\alpha^3t$ and $\nu_1=2\beta(x-(3\alpha^2-\beta^2)t)$.
There is a quasi-periodic phenomenon and $|q_3|^2_{\varepsilon=i}$ has singularity along with straight line
$x=(\alpha^2-3\beta^2) t-\hbar$. The speed is $\alpha^2-3\beta^2$.
We depict this solution in Fig. 15. When $k_1=c_1=d_1=1+i$,
$|q_3|^2_{(\varepsilon=i)}(x,t=\frac{\kappa\pi}{2})\rightarrow 0$
as $x\rightarrow -\kappa\pi$. Specially, $|q_3|^2_{(\varepsilon=i)}(x,t=0)$ is an even function of $x$ (see Fig. 15(f)).
The background level of here is 1.


\begin{center}
\begin{picture}(120,100)
\put(-150,-23){\resizebox{!}{3.5cm}{\includegraphics{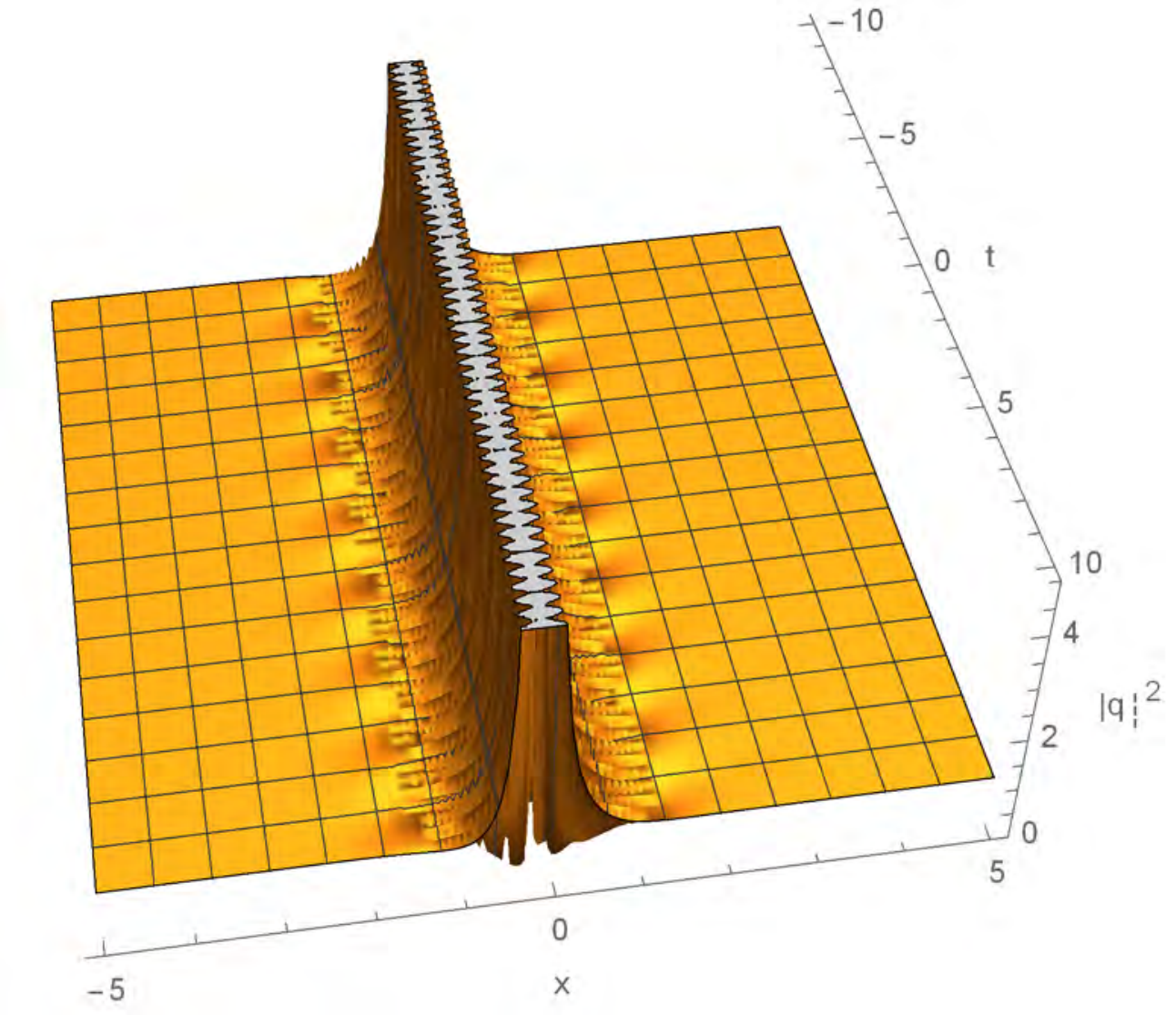}}}
\put(10,-23){\resizebox{!}{3cm}{\includegraphics{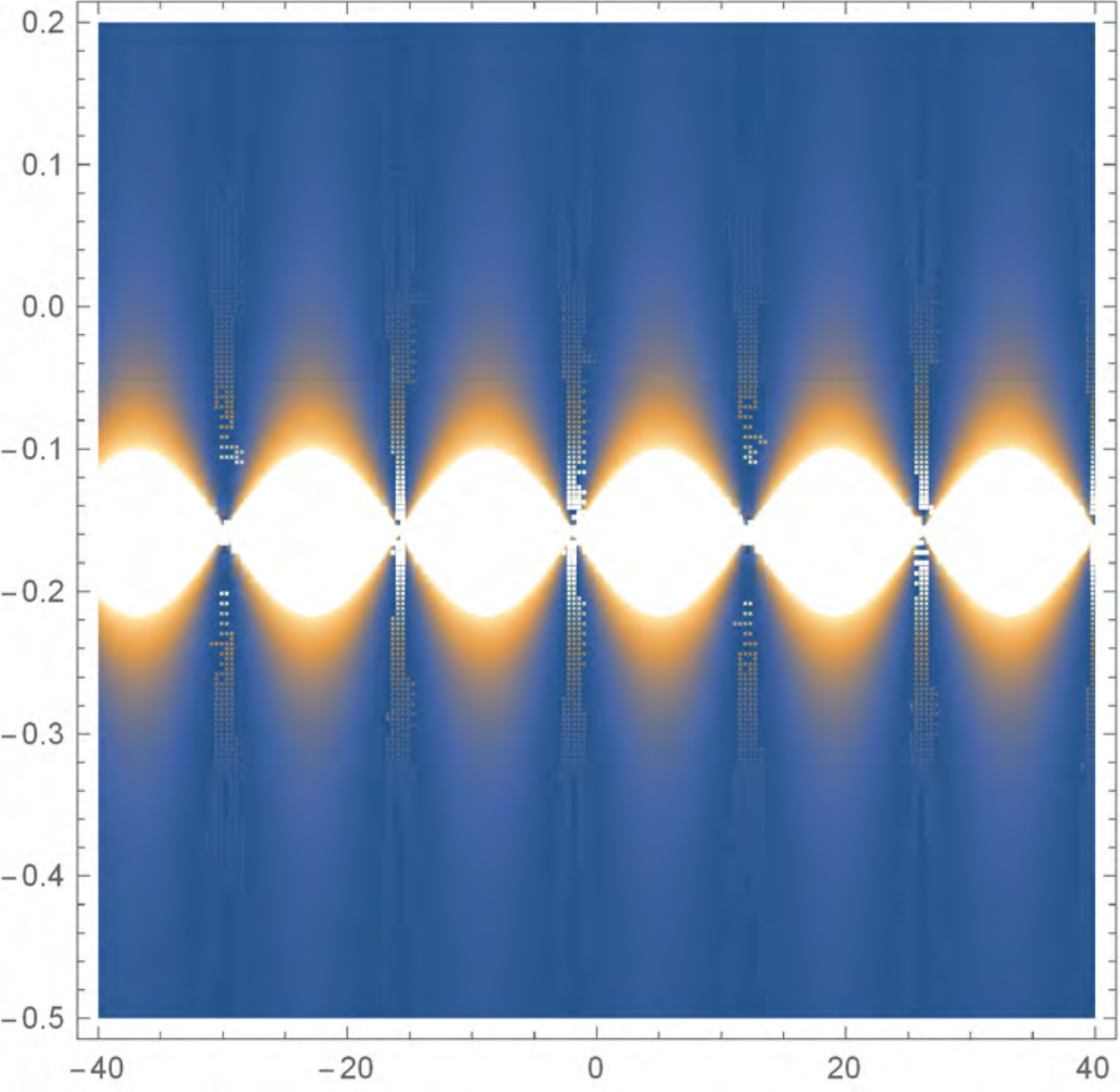}}}
\put(150,-23){\resizebox{!}{3cm}{\includegraphics{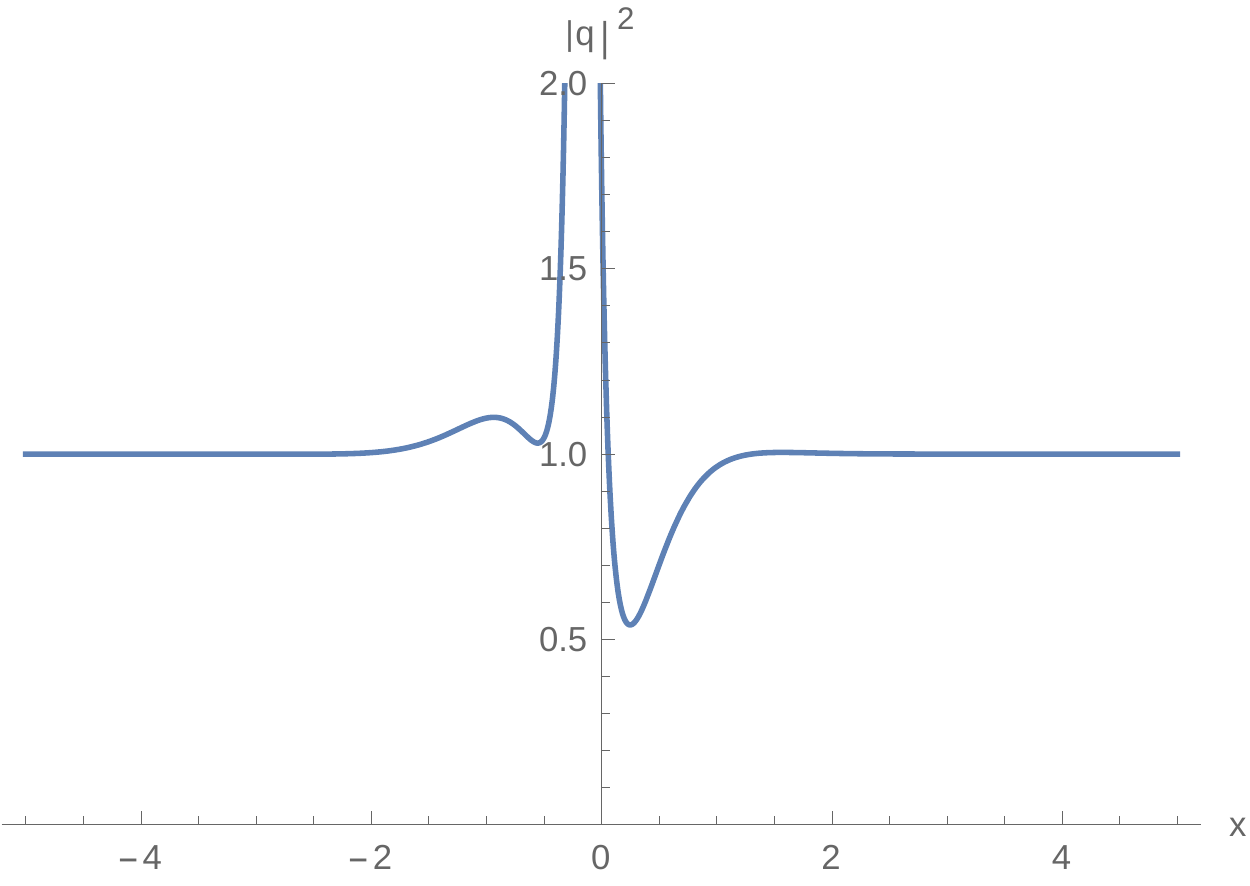}}}
\end{picture}
\end{center}
\vskip 20pt
\begin{center}
\begin{minipage}{16cm}{\footnotesize
\quad\qquad\qquad\qquad(a)\qquad\qquad\qquad\qquad \qquad\quad \qquad \qquad (b) \quad\qquad \qquad \qquad\qquad\qquad\qquad \qquad (c)}
\end{minipage}
\end{center}
\vskip 10pt
\begin{center}
\begin{picture}(120,80)
\put(-150,-23){\resizebox{!}{3.5cm}{\includegraphics{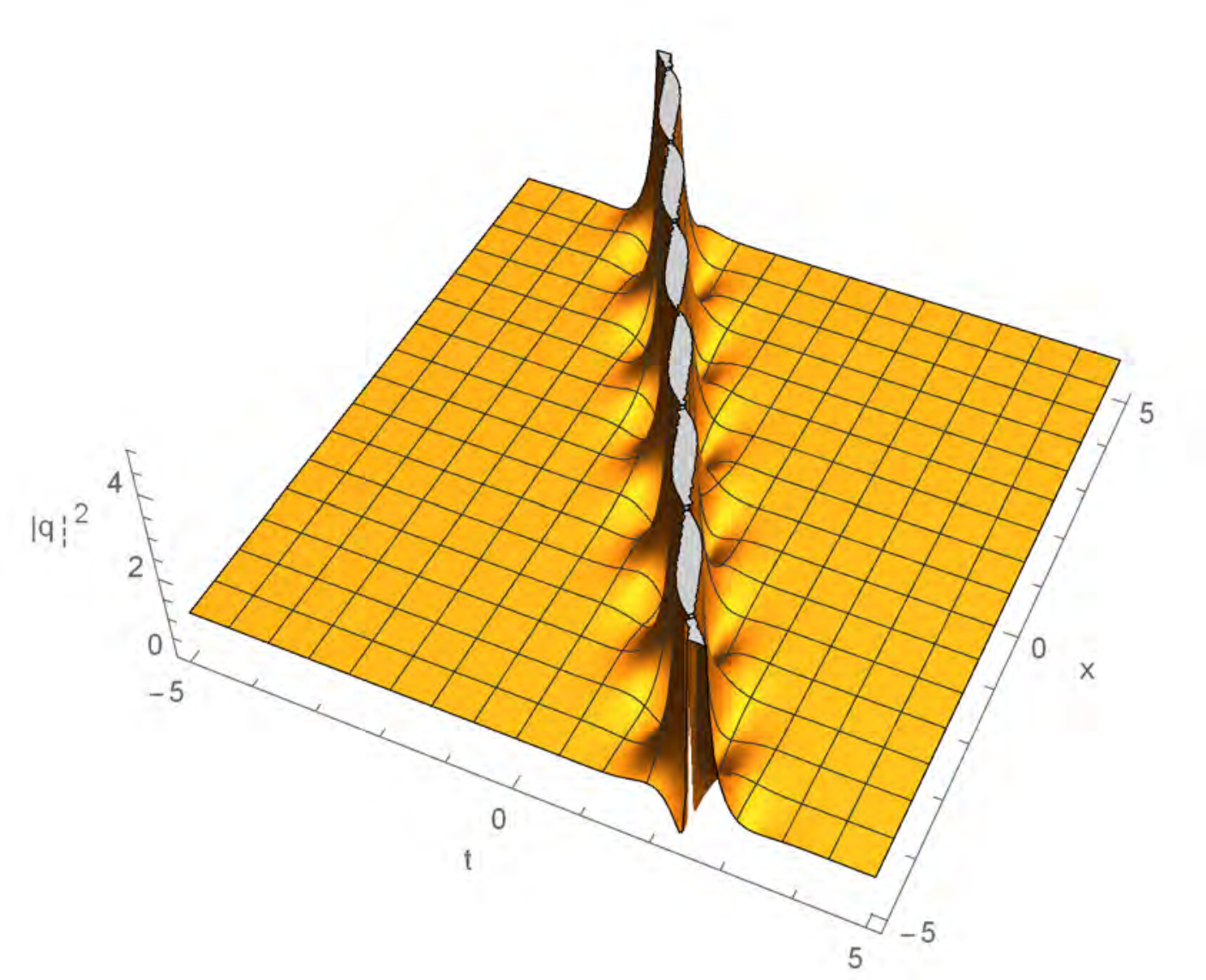}}}
\put(10,-23){\resizebox{!}{3cm}{\includegraphics{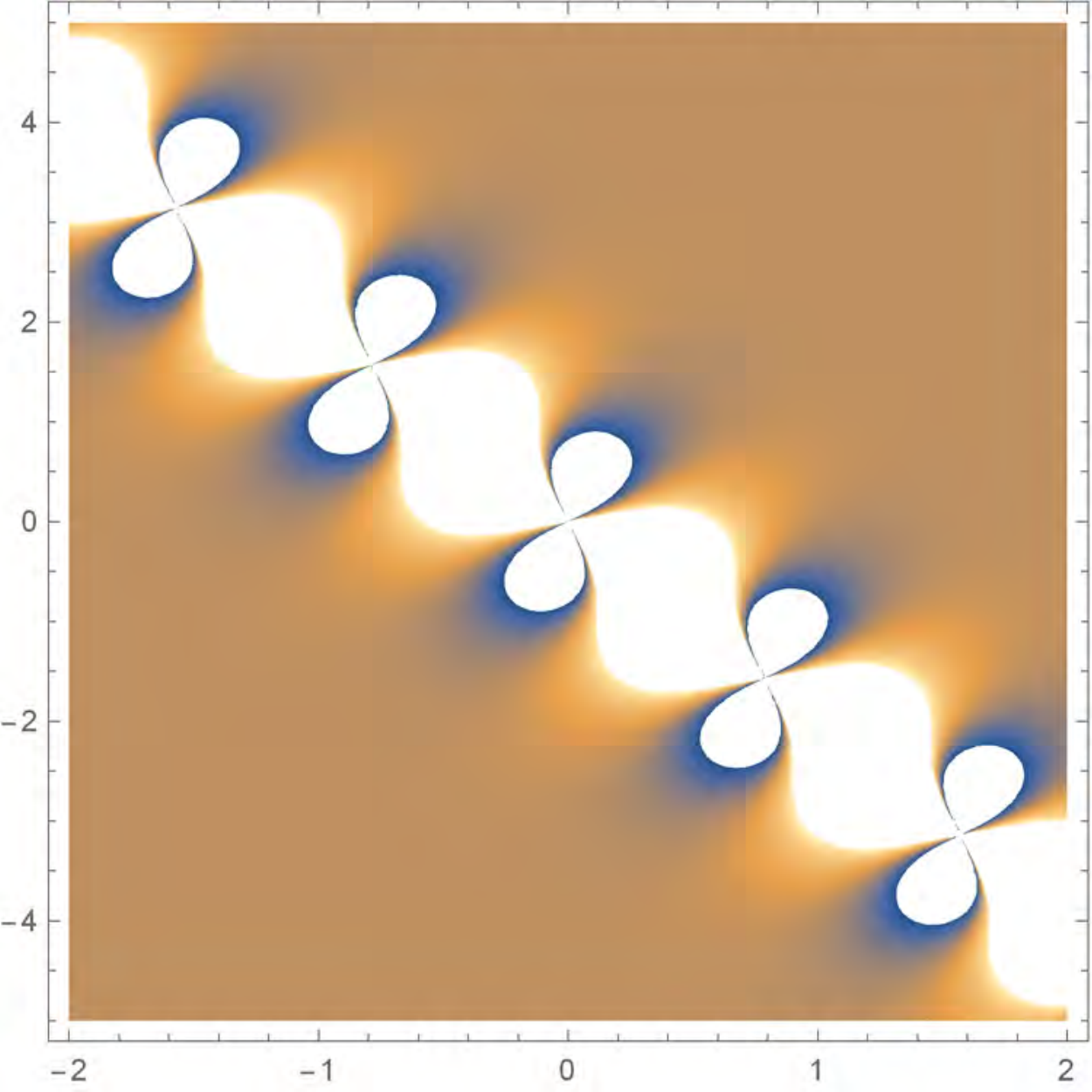}}}
\put(150,-23){\resizebox{!}{3cm}{\includegraphics{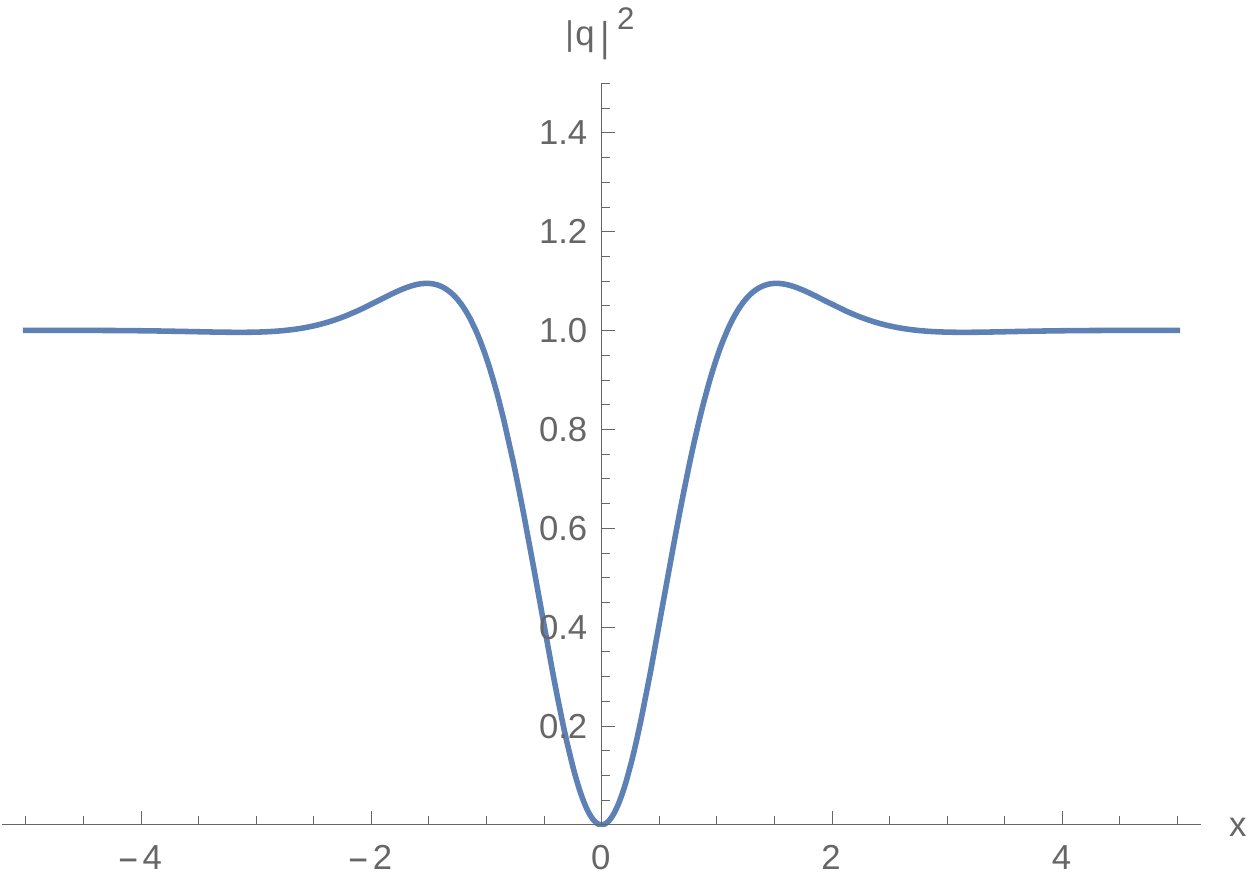}}}
\end{picture}
\end{center}
\vskip 20pt
\begin{center}
\begin{minipage}{16cm}{\footnotesize
\quad\qquad\qquad\qquad(d)\qquad\qquad\qquad\qquad \qquad\quad \qquad \qquad (e) \qquad\qquad \qquad \qquad\qquad\qquad\qquad \quad (f)\\
{\bf Fig. 15} Shape and motion of $|q_3|^2_{(\varepsilon=i)}$ for $c_1=d_1=1+i$:
(a) wave for $k_1=\sqrt{3}+i$.
(b) a contour plot of (a) with range $x\in [-0.5, 0.2]$ and $t\in [-40,40]$.
(c) 2D-plot of (a) at $t=0$. (d) wave for $k_1=1+i$.
(e) a contour plot of (d) with range $x\in [-5, 5]$ and $t\in [-2,2]$.
(f) 2D-plot of (d) at $t=0$.}
\end{minipage}
\end{center}


The solution \eqref{q-solu-cpmKdV-(1)} yields
$|q_3|^2_{(\varepsilon=1)}=\frac{q'_{\ty{D},4}}{q'_{\ty{D},3}}$, where
\begin{subequations}
\label{q34}
\begin{align}
& q'_{\ty{D},4}=16\beta^4(\Theta_1^2+\Theta_2^2)e^{2\mu_3}+8\beta^2
\big[\big(|\vartheta_1|^2\big((\alpha^2-\beta^2)\Theta_2+2\alpha\beta \Theta_1\big)-
4\beta^2(\alpha^2+\beta^2)\Theta_2\big)\cos\nu_1 \nn \\
& \quad\quad +\big(|\vartheta_1|^2\big(2\alpha\beta \Theta_2+(\beta^2-\alpha^2)\Theta_1\big)
-4\beta^2 (\alpha^2+\beta^2)\Theta_1\big)\sin\nu_1\big]e^{\mu_2+\mu_3+\mu_4} +
(\alpha^2+\beta^2)\nn \\
& \quad\quad \cdot \big[(\alpha^2+\beta^2)(|\vartheta_1|^4+16\beta^4)+
8|\beta \vartheta_1|^2\big((\beta^2-\alpha^2)\cos 2\nu_1-2\alpha\beta\sin2\nu_1\big)\big]e^{2(\mu_2+\mu_4)}, \\
& q'_{\ty{D},3}=(\alpha^2+\beta^2)^2(|\vartheta_1|^4+16\beta^4-8|\beta\vartheta_1|^2\cos2\nu_1)e^{2(\mu_2+\mu_4)},
\end{align}
\end{subequations}
where $\mu_3=4\alpha(\alpha^2-3\beta^2) t$ and $\mu_4=2\alpha(x-3\beta^2t)$.
Solution $|q_3|^2_{(\varepsilon=1)}$ has singularities along points
\begin{align}
x=(3\alpha^2-\beta^2)t+\frac{\kappa \pi}{2\beta}, \quad \kappa\in \mathbb{Z}.
\end{align}
Similarly, it has quasi-periodic phenomenon. We depict this solution in Fig. 16.


\begin{center}
\begin{picture}(120,100)
\put(-150,-23){\resizebox{!}{3.5cm}{\includegraphics{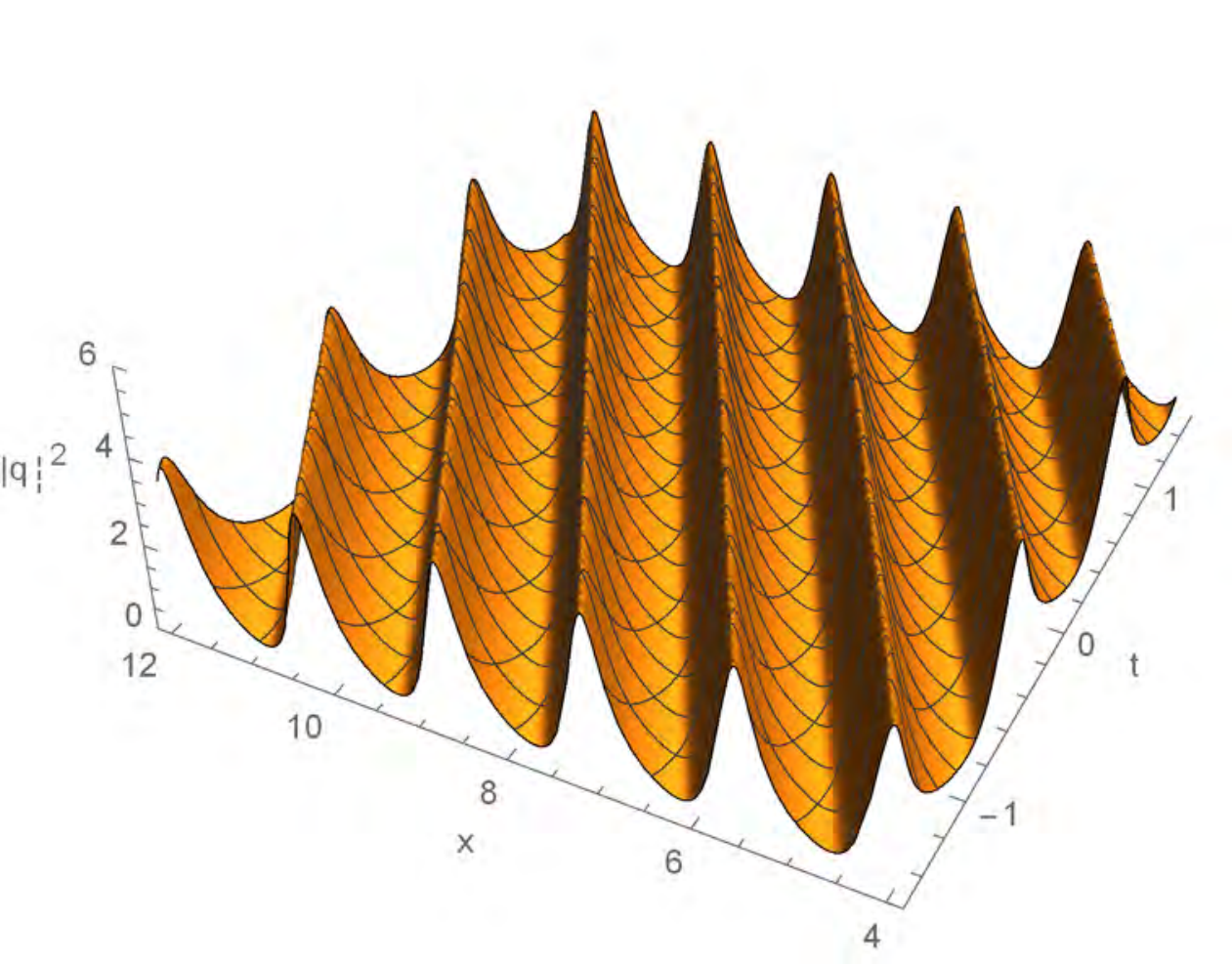}}}
\put(10,-23){\resizebox{!}{3cm}{\includegraphics{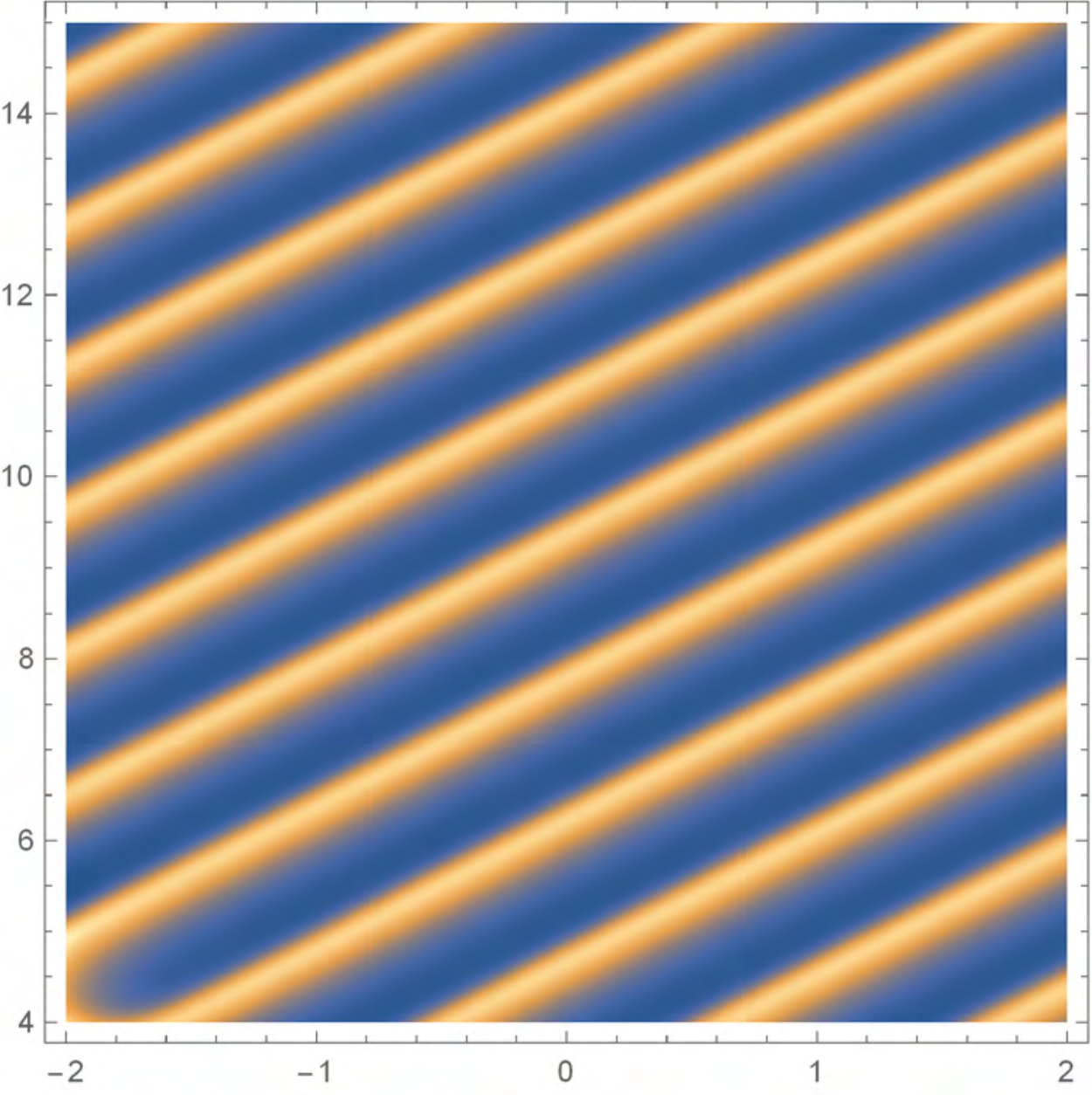}}}
\put(150,-23){\resizebox{!}{3cm}{\includegraphics{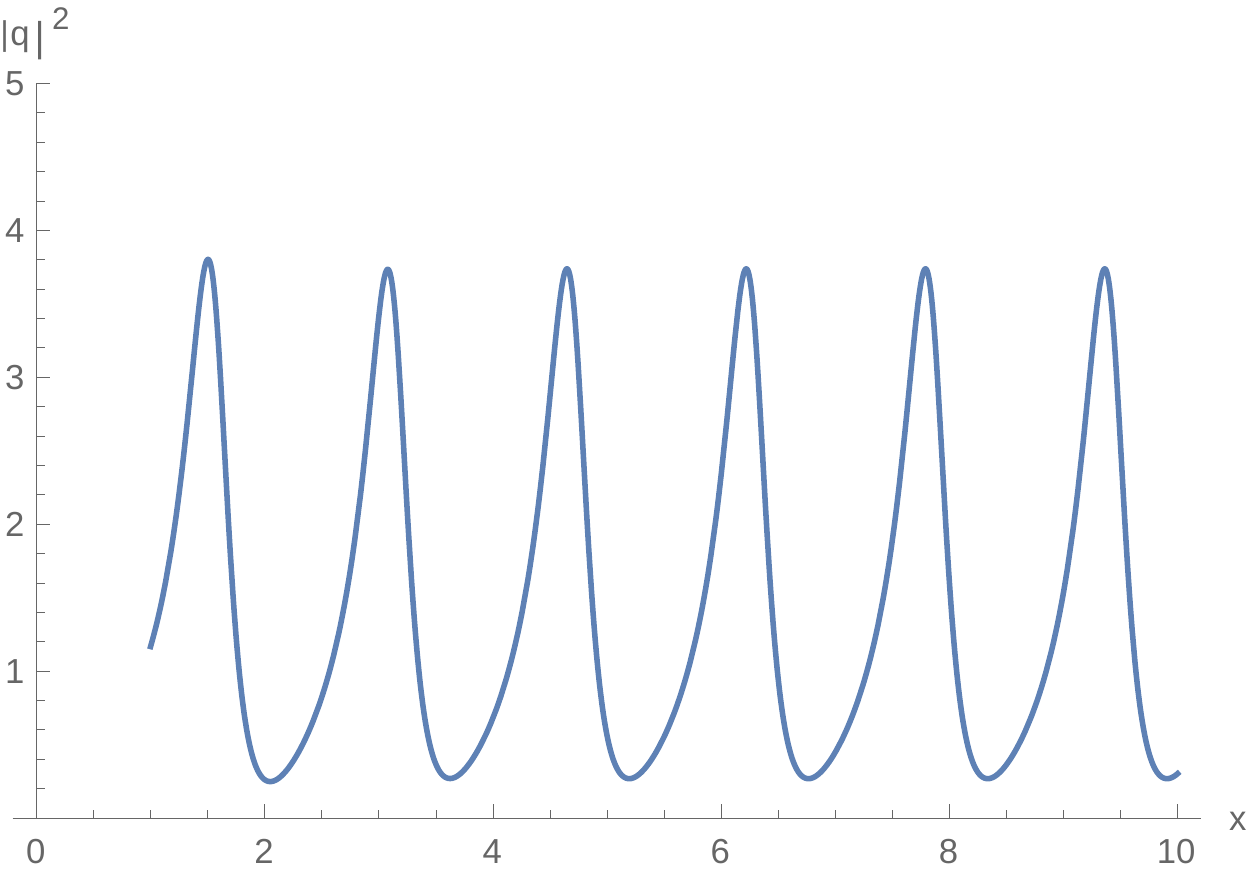}}}
\end{picture}
\end{center}
\vskip 20pt
\begin{center}
\begin{minipage}{16cm}{\footnotesize
\quad\qquad\qquad\qquad(a)\qquad\qquad\qquad\qquad \qquad\quad \qquad \qquad (b) \qquad\qquad \qquad \qquad\qquad\qquad\qquad \qquad (c)}
\end{minipage}
\end{center}
\vskip 10pt
\begin{center}
\begin{picture}(120,80)
\put(-150,-23){\resizebox{!}{3.5cm}{\includegraphics{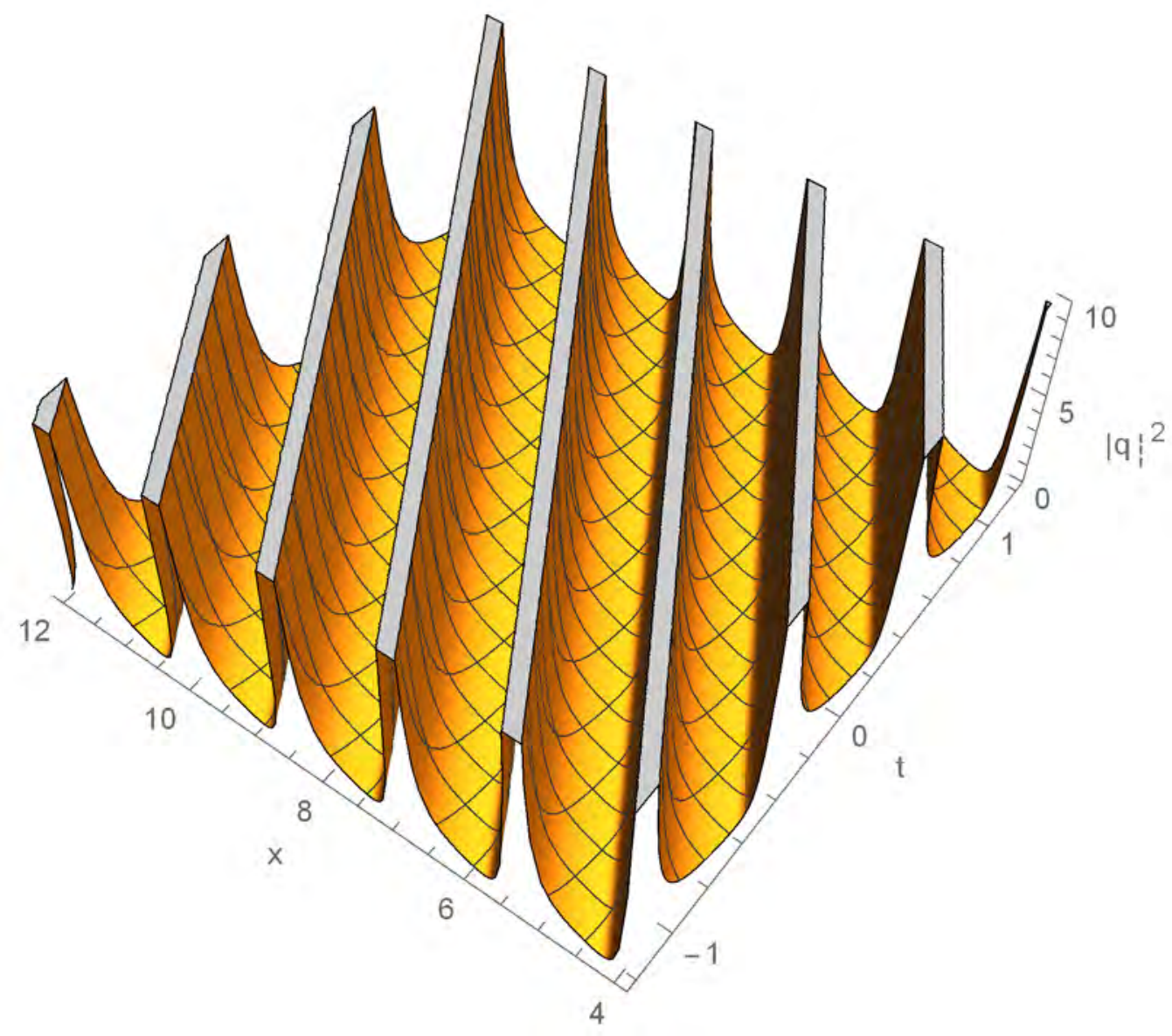}}}
\put(10,-23){\resizebox{!}{3cm}{\includegraphics{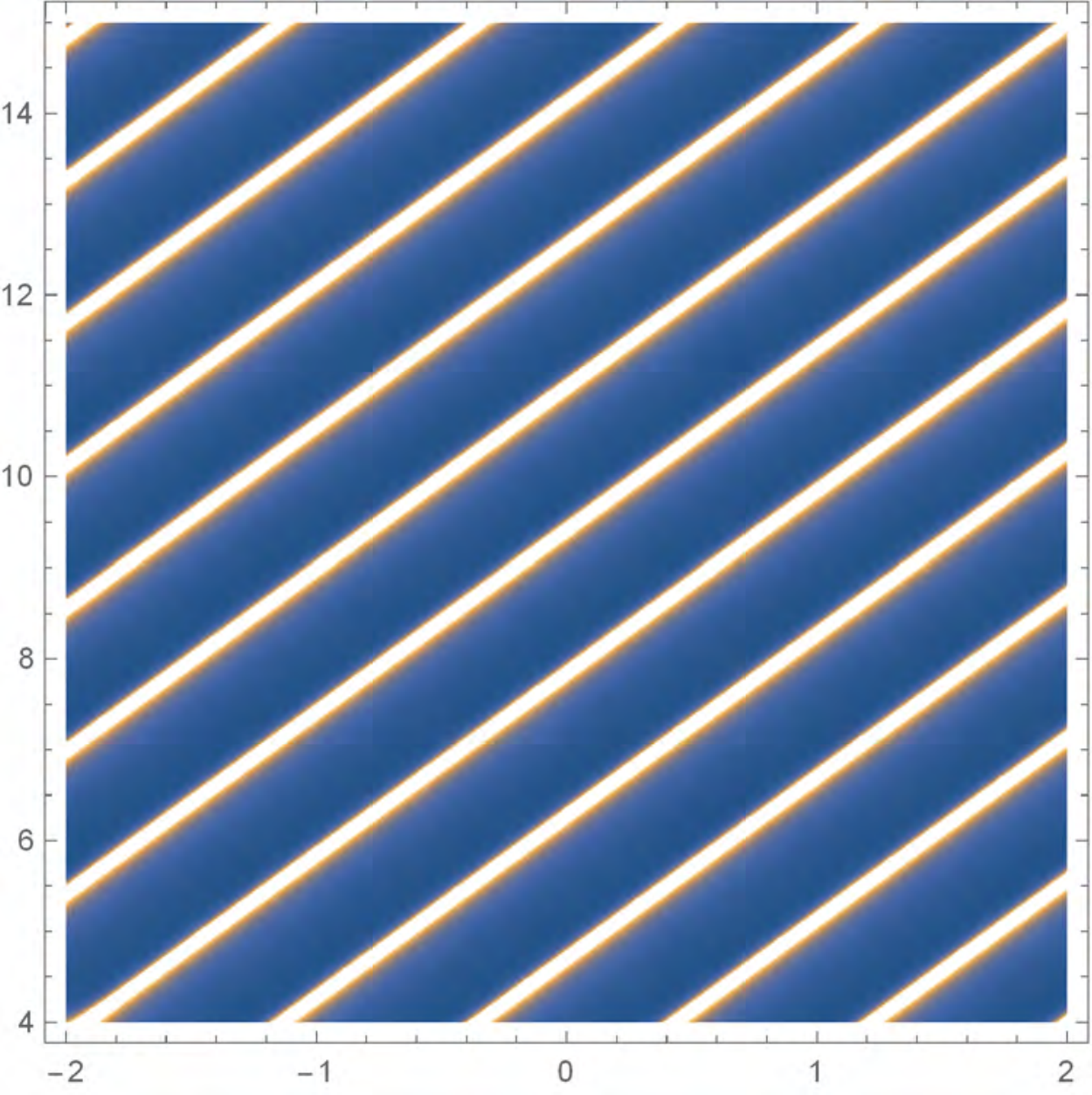}}}
\put(150,-23){\resizebox{!}{3cm}{\includegraphics{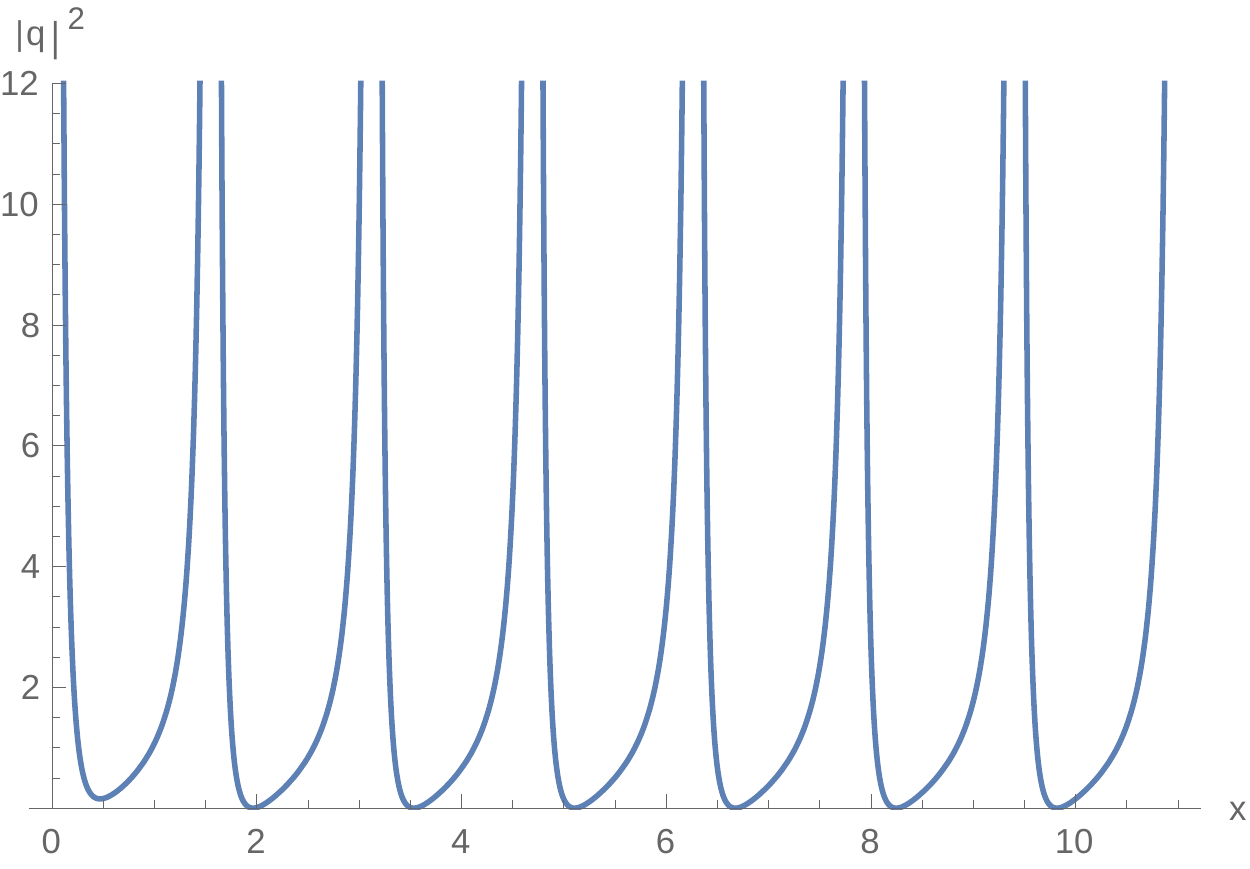}}}
\end{picture}
\end{center}
\vskip 20pt
\begin{center}
\begin{minipage}{16cm}{\footnotesize
\quad\qquad\qquad\qquad(d)\qquad\qquad\qquad\qquad \qquad\quad \qquad \qquad (e) \qquad\qquad \qquad \qquad\qquad\qquad\qquad \qquad (f)\\
{\bf Fig. 16} shape and motion of $|q_3|^2_{(\varepsilon=1)}$: (a) wave for $k_1=0.9+i$, $c_1=2+i$
and $d_1=1+i$. (b) a contour plot of (a) with range $x\in [4, 15]$ and $t\in [-2,2]$.
(c) 2D-plot of (a) at $t=0$. (d) wave for $k_1=c_1=d_1=1+i$.
(e) a contour plot of (d) with range $x\in [4, 15]$ and $t\in [-2,2]$.
(f) 2D-plot of (d) at $t=0$.}
\end{minipage}
\end{center}


\section{Conclusions and some remarks} \label{Con}

In this paper, we have exhibited a Cauchy matrix reduction technique
that enables us to obtain solutions for the reduced local and nonlocal complex equations from the Cauchy matrix
solutions of the original before-reduction systems. We investigated the Cauchy matrix solutions of
four local and nonlocal complex equations, including local and nonlocal cmKdV equation, local and nonlocal csG equation,
local and nonlocal potential NLS equation and local and nonlocal potential cmKdV equation. For each reduced equation, we
gave their soliton solutions and Jordan block solutions, respectively. In particular, we presented the
explicit expressions of 1-soliton solutions, 2-soliton solutions and the simplest Jordan block solutions.
The dynamical behaviors for these solutions were analyzed with graphical illustrations.

We found that the local and nonlocal cmKdV equation \eqref{cmKdV} and the local and nonlocal csG equation \eqref{n-sG},
share the same solutions in formal, but with different dispersion relations, which also holds for the
local and nonlocal potential NLS equation \eqref{n-pNLS-1} and the local and nonlocal potential cmKdV equation \eqref{n-pmKdV-1}.
The matrix $\bT$ is indispensable to relate
$(\Og_1,~\bC^+_1,~\bD^+_1,~\bM_1,~\br_1,~\bs_1)$ and $(\Og_2,~\bC^{+}_2,~\bD^{+}_2,~\bM_2,~\br_2,~\bs_2)$,
while it can be absorbed by $\bM_2$ in the expressions of solution.
Besides, we observed that all the solutions, involving 1-soliton solutions, 2-soliton solutions and the simplest Jordan block solutions,
for the nonlocal complex equations have the quasi-periodic phenomenon.

We finish the paper by the following remarks.

First of all, the reduction technique given in the present paper can be used to
study the local and nonlocal complex reductions of the AKNS hierarchies \eqref{AKNS-hie} and
\eqref{NAKNSH}. The most direct way is to extend dispersion relations by simply introducing
infinite time variables. For example, for the positive AKNS hierarchy
\begin{subequations}
\label{posi-AKNS-hie}
\begin{align}
& 2 u_{t_2}-u_{xx}-8u^2v=0, \\
& 2 v_{t_2}+v_{xx}+8uv^2=0, \\
& 2 u_{t_{n}}+u_{xt_{n-1}}+8u\partial^{-1}(uv)_{t_{n-1}}=0, \quad n=3,4,\ldots, \\
& 2 v_{t_{n}}-v_{xt_{n-1}}-8v\partial^{-1}(uv)_{t_{n-1}}=0, \quad n=3,4,\ldots,
\end{align}
\end{subequations}
its Cauchy matrix solutions are still expressed by
\eqref{AKNS-solu}, with \eqref{DES-M12-C} and general dispersion relations
\begin{subequations}
\begin{align*}
& \br_{j}\rightarrow \mbox{exp}\big((-1)^{j-1}(\Og_{j}x-\sum_{n=2}^{\infty}\Og^n_{j}t_{n})\big)\bC^+_{j},  \\
& \bs_{j}\rightarrow \mbox{exp}\big((-1)^{j-1}(\Og^{\st}_{j}x-\sum_{n=2}^{\infty}(\Og_j^{\st})^nt_{n})\big)\bD^+_{j}.
\end{align*}
\end{subequations}
Here we neglect the trivial member $u_{t_1}-u_{x}=0,~v_{t_1}-v_{x}=0$.
For the even members in the hierarchy \eqref{posi-AKNS-hie}, we consider
complex reduction $v(x,t)=\delta u^*(\sigma x,\sigma t), ~ \delta,~\sigma=\pm 1$.
Table 1 and Table 2, respectively, show soliton solutions and Jordan block solutions
to the higher order local and nonlocal cmKdV equations.
For the odd members, we introduce complex reduction $v(x,t)=\delta u^*(\sigma x, t),~\delta,~\sigma=\pm 1,~t\rightarrow -it$.
The solutions to the higher order local and nonlocal NLS equations can be got without any problem.

What's more, compared with the double Wronskian reduction technique developed in \cite{SAPM-2018},
there is a deficiency of the Cauchy matrix reduction technique, that is, we can't get the Cauchy matrix
solutions for the nonlocal real equations reduced from the AKNS-type equations. For instance, to derive Cauchy matrix solutions for the nonlocal real mKdV equation
\begin{align*}
4 u_{t}+u_{xxx}+24\delta uu(-x,-t)u_{x}=0, \quad \delta=\pm 1,
\end{align*}
we should replace the first equation in
\eqref{nlcmKdV-at-eq} by $\Og_1\bT-\bT\Og_2=0$, where $\bT$ is an invertible matrix. This implies that
matrices $\Og_1$ and $\Og_2$ share the same eigenvalues, which is contradictory to the condition \eqref{con-Og}.

Finally, let us go back to the pAKNS equations \eqref{2th-pAKNS} and \eqref{3th-pAKNS}.
These two equations were firstly derived from \cite{Zhao-ROMP} by utilizing the Cauchy matrix approach.
For the pAKNS equation \eqref{3th-pAKNS}, when $r=q$ it reduces to an equation
\begin{align}
\label{meta}
4 q_{t}+q_{xxx}-3\frac{q_xq_{xx}}{q}=0,
\end{align}
which is always called potential mKdV equation \cite{N-2004-math}. In \cite{Carillo}, a noncommutative version of
\eqref{meta} was introduced. Moreover, its mirror counterpart, recursion operators, hierarchies, and an explicit solution class were also
derived. So how to construct the potential AKNS hierarchy and discuss its integrability is an interesting problem worth consideration.

\vskip 20pt
\subsection*{Acknowledgments}
This project is supported by the National Natural Science Foundation of
China (No. 12071432) and the Natural Science Foundation of Zhejiang Province (Nos. LY17A010024,
LY18A010033).

\end{document}